\documentclass[11pt, letterpaper]{article}
\usepackage{fullpage}

\usepackage[pagebackref]{hyperref}

\usepackage[utf8]{inputenc}
\usepackage{mathtools, nccmath}

\usepackage{array}
\usepackage{amsmath,amsthm,amssymb, color, tikz, mathtools}
\usepackage[ruled,vlined]{algorithm2e}
\usepackage{xcolor}
\usepackage{pgfplots}
\usepackage{pgfplotstable}
\usepackage{thmtools}
\usepackage{thm-restate}
\usepackage{comment}
\usepackage{lipsum}
\usepackage{titlesec}

\titlespacing{\paragraph}{0pt}{*.9}{*.9}

\usepgflibrary{shapes.geometric}
\pgfplotsset{my style/.append style={axis lines=middle, ticks=none, xlabel={$x$}, ylabel={$y$}, every axis x label/.style={at={(current axis.right of origin)},anchor=north west},
every axis y label/.style={at={(current axis.above origin)},anchor=north east}}}
\usetikzlibrary{calc}
\pgfplotsset{compat=newest}
\usepgfplotslibrary{fillbetween}

\usetikzlibrary{positioning}
\usetikzlibrary{shapes}
\usetikzlibrary{decorations.pathreplacing}
\usepackage{enumerate, enumitem} 
\usepackage{hyperref} 
\hypersetup{colorlinks,linkcolor={blue},citecolor={blue},urlcolor={black}}
\allowdisplaybreaks

\newcommand{\braces}[1]{\left\{#1\right\}}
\newcommand{\paren}[1]{\left(#1\right)}

\newcommand{\mathify}[1]{\ifmmode{#1}\else\mbox{$#1$}\fi}
\newcommand{\abs}[1]{\mathify{\left| #1 \right|}}

\newcommand{\wh}[1]{\widehat{#1}}

\renewcommand{\deg}{\mathrm{deg}}

\newcommand{\supp}{\mathrm{supp}}
\newcommand{\spann}{\textnormal{span}}

\newcommand{\cS}{\mathcal{S}}
\newcommand{\N}{\mathbb{N}}

\newcommand{\R}{\mathbb{R}}

\newcommand{\pmone}{\braces{-1, 1}}
\newcommand{\zone}{\braces{0, 1}}
\newcommand{\ftwo}{\mathbb{F}_2}

\newcommand{\calB}{\mathcal{B}}
\newcommand{\calS}{\mathcal{S}}
\newcommand{\prob}[2]{\textnormal{Pr}_{#1}\left[#2\right]}
\newcommand{\Ex}[2]{\mathbb{E}_{#1}\left[#2\right]}

\newcommand{\AND}{\mathsf{AND}}
\newcommand{\XOR}{\mathsf{XOR}}

\newcommand{\AD}{\mathsf{AD}}
\newcommand{\fcube}[1]{\{-1,1\}^{#1}}

\newtheorem{theorem}{Theorem}[section]
\newtheorem{corollary}[theorem]{Corollary}
\newtheorem{remark}[theorem]{Remark}
\newtheorem{lemma}[theorem]{Lemma}
\newtheorem{claim}[theorem]{Claim}
\newtheorem{defi}[theorem]{Definition}

\newtheorem{observation}[theorem]{Observation}

\newtheorem{fact}[theorem]{Fact}

\newcommand{\sparsity}{k}
\newcommand{\seespectrum}{k'}
\newcommand{\seerank}{k''}
\newcommand{\threshold}{t}
\newcommand{\rank}{\mathrm{rank}}

\newcommand{\sparp}{\sparsity_{\{\emptyset\}^{\mathsf{c}}}}
\newcommand{\ind}{\mathbb{I}}
\newcommand{\bin}{\mathrm{bin}}

\newcommand{\bra}[1]{\left(#1\right)}
\newcommand{\cbra}[1]{\left\{#1\right\}}
\newcommand{\abra}[1]{\left \langle #1 \right\rangle}

\newcommand{\degtwo}{\deg_{\ftwo}}
\newcommand{\fmin}{f_{\textnormal{min}}}

\DeclareMathOperator*{\argmin}{\arg\!\min}

\newcommand{\lone}[1]{\|#1\|_1}

\newcommand{\Vb}{{(\Gamma,b)}}

\newcommand{\AB}{\mathsf{AB}_{t',\ell}}
\newcommand{\B}{B}
\newcommand{\AAB}{\mathsf{AAB}_{t,t',\ell}}
\newcommand{\circt}{\circ_{\textnormal{target}}}

\newcommand{\bone}{\mathbf{1}}
\newcommand{\mAND}{\mathsf{mAND}}
\newcommand{\mAD}{\mathsf{mAD}}

\definecolor{ao}{rgb}{0.0, 0.4, 0.0}

\title{Tight Chang's-lemma-type bounds for Boolean functions}
\author{
Sourav Chakraborty\thanks{Indian Statistical Institute, Kolkata. \texttt{sourav@isical.ac.in}} 
        \and Nikhil S.~Mande\thanks{CWI, Amsterdam. Work mostly done while the author was a postdoc at Georgetown University.} 
        \and Rajat Mittal\thanks{Indian Institute of Technology Kanpur. \texttt{rmittal@iitk.ac.in}} 
        \and Tulasimohan Molli\thanks{Tata Institute of Fundamental Research, Mumbai. \texttt{tulasimohanm@gmail.com}} \and Manaswi Paraashar\thanks{Indian Statistical Institute, Kolkata. \texttt{manaswi.isi@gmail.com}} 
        \and Swagato Sanyal\thanks{Indian Institute of Technology Kharagpur. \texttt{swagato@cse.iitkgp.ac.in}}}
\date{}

\begin{document}

\pagenumbering{gobble}

\maketitle
\begin{abstract}
Chang's lemma (Duke Mathematical Journal, 2002) is a classical result in mathematics, with applications spanning across additive combinatorics, combinatorial number theory, analysis of Boolean functions, communication complexity and algorithm design. For a Boolean function $f$ that takes values in $\pmone$ let $r(f)$ denote its Fourier rank (i.e., the dimension of the span of its Fourier support). For each positive threshold $t$, Chang's lemma provides a lower bound on $\delta(f):=\Pr[f(x)=-1]$ in terms of the dimension of the span of its characters with Fourier coefficients of magnitude at least $1/t$. In this work we examine the tightness of Chang's lemma with respect to the following three natural settings of the threshold:
\begin{itemize}
    \item the Fourier sparsity of $f$, denoted $k(f)$,
    \item the Fourier max-supp-entropy of $f$, denoted $k'(f)$, defined to be the maximum value of the reciprocal of the absolute value of a non-zero Fourier coefficient,
    \item the Fourier max-rank-entropy of $f$, denoted $k''(f)$, defined to be the minimum $t$ such that characters whose coefficients are at least $1/t$ in magnitude span a $r(f)$-dimensional space.
\end{itemize}
In this work we prove new lower bounds on $\delta(f)$ in terms of the above measures. One of our lower bounds, $\delta(f)=\Omega\left(r(f)^2/(k(f) \log^2 k(f))\right)$, subsumes and refines the previously best known upper bound on $r(f)$ in terms of $k(f)$ by Sanyal (Theory of Computing, 2019). Another lower bound, $\delta(f) =\Omega\left(r(f)/(\seerank(f) \log k(f))\right)$, is based on our improvement of a bound by Chattopadhyay, Hatami, Lovett and Tal~(ITCS, 2019) on the sum of absolute values of level-$1$ Fourier coefficients in terms of $\ftwo$-degree. We further show that Chang's lemma for the above-mentioned choices of the threshold is asymptotically outperformed by our bounds for most settings of the parameters involved.

Next, we show that our bounds are tight for a wide range of the parameters involved, by constructing functions witnessing their tightness. All the functions we construct are modifications of the Addressing function, where we replace certain input variables by suitable functions.
Our final contribution is to construct Boolean functions $f$ for which 
\begin{itemize}
\item our lower bounds asymptotically match $\delta(f)$, and
\item for any choice of the threshold $t$, the lower bound obtained from Chang's lemma is asymptotically smaller than $\delta(f)$.
\end{itemize}
Our results imply more refined deterministic one-way communication complexity upper bounds for XOR functions.
Given the wide-ranging application of Chang's lemma, we strongly feel that our refinements of Chang's lemma will find many more applications.

\end{abstract}
\newpage


\pagenumbering{gobble}

\newpage
\pagenumbering{gobble}



\clearpage 
\pagenumbering{arabic}

\section{Introduction}

Chang's lemma \cite{Chang02, Green04} is a classical result in additive combinatorics. Informally, the  lemma states that all the large Fourier coefficients of the indicator function of a large subset of an Abelian group reside in a low dimensional subspace. The discovery of this lemma was motivated by an application to improve Frieman's theorem on set additions \cite{Chang02}.
The lemma has subsequently found many applications in additive combinatorics and combinatorial number theory. Chang's lemma and the ideas developed in Chang's paper~\cite{Chang02} have been used to prove theorems about arithmetic progressions in sumsets~\cite{Green02, San08}, structure of Boolean functions with small spectral norm~\cite{GS08}, and improved bounds for Roth’s theorem on three-term
arithmetic progressions in the integers~\cite{San11, Bloom16, BS20}. Green and Ruzsa~\cite{GR07} used the ideas of Chang's lemma to prove a generalization of Frieman's theorem for arbitrary Abelian groups.
The Chang's lemma is known to be sharp for various settings of parameters for the group $\mathbb{Z}_N$ \cite{Green03}. 

In this paper, our focus is a specialization of Chang's lemma for the Boolean hypercube. Let $f:\pmone^n \to \pmone$ be a Boolean function. For any positive real number $t$ (which we refer to as the \emph{threshold}) define $\cS_t:=\{S \subseteq [n]:|\wh{f}(S)|\geq \frac{1}{t}\}$.\footnote{The function $f$ is implicit in the definition of $\cS_t$ and will be clear from context.}\textsuperscript{,}\footnote{We refer the reader to Section~\ref{sec:prelims} for preliminaries on Fourier analysis.} 
Viewing elements of $\calS_t$ as vectors in $\ftwo^n$, Chang's lemma gives a lower bound on $\delta(f) :=\Pr[f(x) = -1]$ (called the weight of $f$), in terms of $t$ and the dimension of the span of $\calS_t$ (denoted by $\dim(\cS_t)$). Formally, we have the following lemma, referred to as Chang's lemma in this paper.

\begin{lemma}[Chang's lemma~\cite{Chang02}]
\label{lem:chang}
There exists a universal constant $c>0$ such that the following is true for every integer $n>0$. Let $f:\pmone^n \to \pmone$ be any function and $t$ be any positive real number. Let 
$\delta(f) := \Pr_x[f(x) = -1]$ and 
$d = \dim(\cS_t) > 1$. If $\delta(f)<c$, then
    \[
        \delta(f)=\Omega\left(\frac{\sqrt{d}}{t \sqrt{\log \bra{t^2/d}}}\right).
    \]
\end{lemma}
\begin{remark}
In the literature Chang's lemma is generally stated as an upper bound on $d$ in terms of $\delta(f)$ and $t$. In Section~\ref{sec: appendix} we state the more commonly seen form of Chang's lemma (Lemma~\ref{lem:chang original}) and prove that it is equivalent to Lemma~\ref{lem:chang}.
\end{remark}

This lemma has found numerous applications in complexity theory and algorithms~\cite{BRTW14, CLRS16}, analysis of Boolean functions~\cite{GS08, TWXZ13}, communication complexity~\cite{TWXZ13, HLY19} and extremal combinatorics~\cite{FKKK18}. See \cite{IMR14} for a proof of Lemma~\ref{lem:chang}.

In this paper, we investigate the tightness of Lemma~\ref{lem:chang} for three natural choices of the threshold $t$ based on the Fourier spectrum of the function (see Section~\ref{sec: thresholds considered} for details about these thresholds).
We prove additional lower bounds on $\delta(f)$, and compare relative performances of all the bounds under consideration. Our results imply that the bounds given by Chang's lemma for the choices of the threshold that we consider are asymptotically outperformed by one of the bounds we prove for a broad range of the parameters involved. For most regimes of the parameters we are able to construct classes of functions that witness the tightness of our bounds. 

Interestingly, for each choice of threshold that we consider, $\dim(\cS_t)$ equals the Fourier rank of $f$ (denoted by $r(f)$, see Definition~\ref{defi:rank}). In particular, setting $t$ to be the Fourier sparsity of $f$ (denoted by $k(f)$, see Definition~\ref{defi:sparsity}) leads to a very natural question about the relationship among $r(f), k(f)$ and $\delta(f)$ for a Boolean function $f$. The best known upper bound on $r(f)$ in terms of $k(f)$ is  $r(f)=O(\sqrt{k(f)}\log k(f))$~\cite{San19}.
We improve upon this bound by incorporating $\delta(f)$ into it, and show 
\[
r(f)=O(\sqrt{k(f)\delta(f)}\log k(f)).
\]
Moreover, we also show that this bound is tight; see Section~\ref{sec: our contributions} for a detailed discussion.

Throughout this paper, we assume that $f$ is not a constant function or a parity or a negative parity (unless mentioned otherwise). In other words, $k(f),r(f) \geq 2$.

\subsection{Thresholds considered for Chang's lemma}
\label{sec: thresholds considered}
For a Boolean function $f$, let $\supp(f)$ denote the Fourier support of $f$ (Definition~\ref{defi: Fouriersupport}). In this section, we discuss and motivate the choices of the threshold $t$ considered in this work.

    \paragraph{The Fourier sparsity of $f$.} It was shown in~\cite[Theorem 3.3]{GOS+} that for all $S \in \supp(f)$, $|\wh{f}(S)| \geq \frac{1}{k(f)}$. It follows that $\cS_{k(f)}=\supp(f)$ and hence $\dim(\cS_{k(f)})=r(f)$. Moreover, there exist functions (e.g.~$f = \AND_n$) for which $\dim(\calS_t) = 0$ for $t = o(k(f))$, justifying the choice of threshold $k(f)$.

This choice also leads us to a fundamental structural problem of bounding the weight of a Boolean function $f$ from below, in terms of its Fourier sparsity and Fourier rank. The \emph{uncertainty principle} (see, for example, \cite{GT13} for a statement and a proof) asserts that $\delta(f) = \Omega\left(\frac{1}{k(f)}\right)$. Chang's lemma with $t=k(f)$ and the fact that $\log\bra{k(f)^2/r(f)}=\Theta(\log k(f))$ (Lemma~\ref{lem:relationships between rk and k'} (part 1)) implies that
\begin{equation}
\label{eq:changs lemma with threshold k}
\delta(f)=\Omega\left(\frac{1}{k(f)}\sqrt\frac{{r(f)}}{\log k(f)}\right),
\end{equation}
thereby subsuming the uncertainty principle (note that $r(f)/\log k(f) \geq 1$) and refining it by incorporating $r(f)$ into the bound.

\paragraph{The Fourier max-supp-entropy of $f$.} The next choice of the threshold that we consider is the \emph{Fourier max-supp-entropy} of $f$, denoted by $k'(f)$,
which we define to be $\max_{S \in \supp(f)}$ $\frac{1}{|\wh{f}(S)|}$ (Definition~\ref{defi:max entropy, max rank entropy}). By its definition $\seespectrum(f)$ is the smallest value of $t$ such that $\cS_t=\supp(f)$.
Since $k'(f) \leq k(f)$ (see the discussion in the last item), the knowledge of $\seespectrum(f)$ can potentially offer us a more fine-grained lower bound on $\delta(f)$ than as in the last item; Chang's lemma with $t=\seespectrum(f)$ and $\log \bra{k'(f)^2/r(f)} = \Theta(\log k'(f))$ (Lemma~\ref{lem:relationships between rk and k'} (part 2)) implies  

\begin{equation}
\label{eq: changs lemma with threshold k'}
\delta(f)=\Omega\left(\frac{1}{k'(f)}\sqrt\frac{{r(f)}}{\log k'(f)}\right).
\end{equation}
Notice that Equation~\eqref{eq: changs lemma with threshold k'} subsumes the bound in Equation~\eqref{eq:changs lemma with threshold k}. 

In~\cite{HKP11} an equivalent statement of the well-known sensitivity conjecture was presented in terms of $k'(f)$.\footnote{In \cite{HKP11} $\log (k'(f)^2)$ is called the Fourier max-entropy while we refer to $k'(f)$ as the Fourier max-supp-entropy.} 
Granularity is another widely-studied measure that is closely associated with Fourier max-supp-entropy.

    \paragraph{The Fourier max-rank-entropy of $f$.} Our final choice of the threshold is the \emph{Fourier max-rank-entropy of $f$}, denoted by $\seerank(f)$,
    which we define to be the smallest positive real number $t$ such that $\dim(\cS_t)=r(f)$ (Definition~\ref{defi:max entropy, max rank entropy}). We have that $\seerank(f)\leq\seespectrum(f)\leq k(f)$ by their definitions.  Amongst all settings of the threshold $t$ for which $\dim(\calS_t) = r(f)$, the value $t = k''(f)$ yields the best lower bound from Chang's lemma.
 Chang's lemma with $t=\seerank(f)$ implies  
\begin{equation}
\label{eq: changs lemma with threshold k''}
\delta(f)=\Omega\left(\frac{1}{\seerank(f)}\sqrt\frac{{r(f)}}{\log \bra{\seerank(f)^2/r(f)}}\right),
\end{equation}
which subsumes the bounds in Equations~\eqref{eq: changs lemma with threshold k'} and~\eqref{eq:changs lemma with threshold k}.

\subsection{Our contributions}
\label{sec: our contributions}
We prove the following results regarding the three natural instantiations of the threshold $t$ (mentioned in the preceding section) for Chang's lemma.
\begin{enumerate}[leftmargin=*]
\item \textbf{The Fourier sparsity of $f$}:
Recall that Chang's lemma with threshold $t=k(f)$ (Equation~\eqref{eq:changs lemma with threshold k})
implies that 
$\delta(f)=\Omega\left(\frac{1}{k(f)}\sqrt\frac{{r(f)}}{\log k(f)}\right)$. 
It was shown in~\cite{ACL+19} that $\delta(f)=\Omega\left(\frac{1}{k(f)}\left(\frac{r(f)}{\log k(f)}\right)\right)$, improving upon this bound asymptotically (note that $r(f)/\log k(f) \geq 1$). In this work we improve their bound further.
\begin{theorem}
\label{thm:delta lower bound in terms of rk only}
    Let $f: \pmone^n \to \pmone$ be any function such that $k(f)>1$. Then
        \[
        \delta(f)=\Omega\left(\frac{1}{k(f)} \left(\frac{r(f)}{\log k(f)}\right)^2\right).
        \]
\end{theorem}
Observe that the statement of Theorem~\ref{thm:delta lower bound in terms of rk only} is equivalent to $r(f)=O(\sqrt{k(f)\delta(f)}\log k(f))$. This bound subsumes the bound $r(f)=O(\sqrt{k(f)} \log k(f))$ shown by Sanyal~\cite{San19}. We prove Theorem~\ref{thm:delta lower bound in terms of rk only} by incorporating $\delta(f)$ in Sanyal's arguments and thereby refining his proof.
See Section~\ref{sec:lower bound using k'} for  the proof of Theorem~\ref{thm:delta lower bound in terms of rk only}.

We also show that Theorem~\ref{thm:delta lower bound in terms of rk only} is tight. For nearly all admissible values of $\rho$ and $\kappa$ we construct many Boolean functions $f$ with $k(f)=O(\kappa)$, $r(f)=O(\rho)$ and $\delta(f)=O\left(\frac{1}{\kappa} \left(\frac{\rho}{\log \kappa}\right)^2\right)$ (Theorem~\ref{thm:delta upper bound in terms of rk and also k'} and Claim~\ref{claim:setting parameters for tightstraightline}).

\textbf{Comparison with Sanyal's bound:}
The bound $r(f)=O(\sqrt{ k(f)} \log k(f))$ proven by Sanyal is a special case of Theorem~\ref{thm:delta lower bound in terms of rk only} for $\delta(f)=\Theta(1)$. It is not known whether the $\log k(f)$ term is required in Sanyal's upper bound on $r(f)$ (when $f$ equals the Addressing function, $r(f) = \Omega(\sqrt{k(f)})$, see Definition~\ref{defi:Addressing} and Observation~\ref{obs:properties of AND, Bent and Addressing}). For all the functions we construct witnessing the tightness of the bound in Theorem~\ref{thm:delta lower bound in terms of rk only}, $\delta(f)=o(1)$. 

We prove Theorem~\ref{thm:delta lower bound in terms of rk only} by generalizing Sanyal's proof. 
As stated before, our bound is tight in this generality, i.e.~the logarithmic factor is required in the upper bound on $r(f)$. This sheds light on the presence of the logarithmic term in the bound $r(f)=O(\sqrt {k(f)} \log k(f))$.

\item \textbf{The Fourier max-supp-entropy of $f$}:
Recall from Section~\ref{sec: thresholds considered} that the Fourier max-supp-entropy of $f$, denoted $k'(f)$, is defined as $k'(f) = \max_{S \in \supp(f)}$ $\frac{1}{|\wh{f}(S)|}$. It can be shown that $\sqrt{\sparsity(f)} \leq \seespectrum(f) \leq \sparsity(f)/2$ (Lemma~\ref{lem:relationships between rk and k'} (part 2)). We prove the following lower bound.
\begin{theorem}
\label{thm:delta lower bound in terms of rk and also k'}
Let $f: \pmone^n \to \pmone$ be any function such that $k(f) > 1$.
Then,
\[
\delta(f)=\Omega\left(\max\left\{\frac{1}{k(f)}\bra{\frac{r(f)}{\log \sparsity(f)}}^2, \frac{\sparsity(f)}{\seespectrum(f)^2}\right\}\right).
\]
\end{theorem}
As is evident from the statement, Theorem~\ref{thm:delta lower bound in terms of rk and also k'} presents two lower bounds, one of which is Theorem~\ref{thm:delta lower bound in terms of rk only}. The other lower bound $\delta(f)\geq\frac{\sparsity(f)}{\seespectrum(f)^2}$ is Claim~\ref{claim:delta at least k/k'^2}.

Chang's lemma with the threshold $t$ set to $\seespectrum(f)$ (Equation~\eqref{eq: changs lemma with threshold k'}), together with the observation that $\log \sparsity(f) = \Theta(\log \seespectrum(f))$, implies $\delta(f) = \Omega\left(\frac{1}{k'(f)}\sqrt{\frac{r(f)}{\log \sparsity(f)}}\right)$. Theorem~\ref{thm:delta lower bound in terms of rk and also k'} subsumes this bound since 
\[
\delta(f)=\Omega\left(\frac{1}{k(f)}\left(\frac{r(f)}{\log \sparsity(f)}\right)^2 \cdot\frac{\sparsity(f)}{\seespectrum(f)^2}\right)^{1/2}=\Omega\left(\frac{1}{k'(f)}\left(\frac{r(f)}{\log \sparsity(f)}\right)\right)=\Omega\left(\frac{1}{k'(f)}\sqrt{\frac{r(f)}{\log \sparsity(f)}}\right),
\]
where the equality follows from $r(f)/\log \sparsity(f) \geq 1$.

In addition, observe from the last equality above that the bound of
Theorem~\ref{thm:delta lower bound in terms of rk and also k'} is asymptotically larger than the bound obtained from Chang's lemma for $t=k'(f)$ (Equation~\eqref{eq: changs lemma with threshold k'}) except when $r(f)/\log k(f) =\Theta(1)$. Theorem~\ref{thm:delta upper bound in terms of rk and also k'} complements Theorem~\ref{thm:delta lower bound in terms of rk and also k'} by showing that for nearly all admissible values of $r(f), k(f)$ and $\seespectrum(f)$, there exists a function for which the larger of the two bounds presented in Theorem~\ref{thm:delta lower bound in terms of rk and also k'} is tight. 
\begin{theorem}
\label{thm:delta upper bound in terms of rk and also k'}
For all $\rho, \kappa, \kappa' \in \N$ such that $\kappa$ is sufficiently large, for all constants $\epsilon > 0$ such that $\log \kappa \leq \rho \leq \kappa^{\frac12 - \epsilon}$  and $\kappa^{\frac12} \leq \kappa' \leq \kappa$, there exists a Boolean function $f_{\rho,\kappa,\kappa'}$ such that $r({f_{\rho,\kappa,\kappa'}}) = \Theta(\rho)$, $k({f_{\rho,\kappa,\kappa'}})  = \Theta(\kappa)$, $ k'({f_{\rho,\kappa,\kappa'}}) = \Theta(\kappa')$ and 
$$\delta(f_{\rho,\kappa,\kappa'}) = \Theta\left(\max\left\{\frac{1}{\kappa}\paren{\frac{\rho}{\log \kappa}}^2,\frac{\kappa}{\kappa'^2} \right\}\right).$$
\end{theorem}

The range of parameters considered in Theorem~\ref{thm:delta upper bound in terms of rk and also k'} is justified by Lemma~\ref{lem:relationships between rk and k'}.
We prove Theorem~\ref{thm:delta upper bound in terms of rk and also k'} in two parts. Fix any $\rho, \kappa$ such that $\log \kappa \leq \rho \leq \kappa^{\frac{1}{2} - \epsilon}$ for some constant $\epsilon > 0$. First, for each value of $\kappa' \in [\frac{\kappa \log \kappa}{\rho}, \kappa]$ we construct a function $f$ for which the first lower bound on $\delta(f)$ from Theorem~\ref{thm:delta lower bound in terms of rk and also k'}  is tight (Claim~\ref{claim:setting parameters for tightstraightline}). Next, for each value of $\kappa' \in [\kappa^{\frac{1}{2}},\frac{\kappa \log \kappa}{\rho}]$ we construct a function $f$ for which the second lower bound on $\delta(f)$ from Theorem~\ref{thm:delta lower bound in terms of rk and also k'} is tight (Claim~\ref{claim:setting parameters for tight curve}). See Figure~\ref{fig: k'} for a graphical visualization of the bounds in Theorem~\ref{thm:delta lower bound in terms of rk and also k'} for any fixed values of $\rho$ and $\kappa$.


\begin{figure}[h]
\centering
\begin{tikzpicture}[scale = 1, xscale = 1.8, yscale=.7]
\def\dummyvar{2};
\def\rcval{2.5}; 
\def\kval{100}; 
\def\rval{(\rcval)^2};
\def\logkval{log2 \kval}; 
\def\rkval{\rval^2/\kval};

    \begin{axis}[
    axis y line = left,
    axis x line = bottom,
    x label style={at={(axis description cs:.25,-0.02)},anchor=north},
    y label style={at={(axis description cs:-0.01,.5)},rotate=0,anchor=south},
    xlabel={$\ \ \ \ \ \ \ \ \ \ $ Fourier max-Entropy ($\kappa'$) $\longrightarrow$},
    ylabel={weight($\delta$) $\longrightarrow$},
    xtick       = {16, 100},
    xticklabels = {$\frac{\kappa}{\rho}\log \kappa\ \ \ \ \ $,$\kappa$},
    ytick       = {1},
    yticklabels = {1},
    xmin = 0, xmax = 102,
    ymin = 0, ymax = 1,
  ]

  
        \addplot[name path=testa, domain=0:100, color = ao, thick, dashed, samples=100]{\rkval} node[above,pos=.8] {\textcolor{black}{\textbf{\color{ao}$k$-line:\color{black}} $\ \delta = \rho^2/(\kappa\log^2 \kappa)$}}; 

   \addplot[name path=testa1, domain=16:100, color = green, ultra thick, samples=75]{\rkval}; 
  
        \addplot[name path=testb, domain=0:100, color = blue, thick, dashed,  samples=75]{\kval/(x^2)} node[above, pos=.83] {\textcolor{black}{\textbf{\color{blue}$k'$-curve:\color{black}} $\ \delta = \kappa/(\kappa')^2$}}; 

        \addplot[name path=testb1, domain=0:16, color = blue, ultra thick,  samples=15]{\kval/(x^2)}; 
 
       \addplot[name path=chang, domain=2.5:100, color = red, thick, dashed, samples=75]{.15+\rcval/x} node[above, pos=0.75] {\textcolor{black}{\textbf{\color{red}CL-$k'$-curve :\color{black}} $\ \delta = \sqrt{\rho}/(\kappa'\log {\kappa'})$}}; 
  
         \addplot [name path=vertline, dotted, thick] coordinates {(\kval/\rval,0) (\kval/\rval, 1)};
       
       \addplot[name path=upper, domain=0:100, color = white, smooth]{1};

        \addplot fill between[ 
    of = chang and upper, 
    split, 
    every even segment/.style = {white!10},
    every odd segment/.style  = {gray!70}
  ];
        
        \addplot fill between[ 
    of = chang and testb, 
    split, 
    every even segment/.style = {gray!20},
    every odd segment/.style  = {white!60}
  ];
  
          \addplot fill between[ 
    of = testa and chang, 
    split, 
    every even segment/.style = {white!60},
    every odd segment/.style  = {gray!20}
  ];
  
  \node[anchor=west] (source) at (axis cs:25,.9){Chang's lemma Bound};
       \node (destination) at (axis cs:5,.6){};
       \draw[->](source)--(destination);
  
  \node[anchor=west] (source2) at (axis cs:45,.7){Our Bounds};
       \node (destination2) at (axis cs:50,.39){};
       \node (destination3) at (axis cs:14,.51){};
       \draw[->](source2)--(destination2);
       \draw[->](source2)--(destination3);

    \end{axis}
\end{tikzpicture}
\caption{This plot is constructed for any fixed values of $\rho, \kappa$ for which $\log \kappa \leq \rho \leq \sqrt{\kappa}$, and depicts the relationship between $\delta(f)$ and $k'(f)$ for functions $f$ with $r(f) = \Theta(\rho)$ and $k(f) = \Theta(\kappa)$. For any fixed values of $\rho, \kappa$, we will refer to this plot as the $(\rho,\kappa)$-$k'$-plot. 
Chang's lemma implies that Boolean functions lie above the CL-$k'$-curve. Theorem~\ref{thm:delta lower bound in terms of rk and also k'} improves upon Chang's lemma and shows that Boolean functions lie above both the $k$-line and the $k'$-curve, highlighted by the dark grey region in the figure. Roughly speaking, Theorem~\ref{thm:delta upper bound in terms of rk and also k'} exhibits functions that lie on the boundary of the dark grey region described by the $k$-line and the $k'$-curve.}
\label{fig: k'}
\end{figure}


\item \textbf{The Fourier max-rank-entropy of $f$}:

Recall from Section~\ref{sec: thresholds considered} that the Fourier max-rank-entropy of $f$, denoted $\seerank(f)$, is the smallest positive real number $t$ such that $\dim(\cS_t)=r(f)$ . It can be shown that $\max\cbra{\sqrt{r(f)}, \frac{r(f)}{\log \sparsity(f)}} \leq \seerank(f) \leq \sparsity(f)$ (Lemma~\ref{lem:relationships between rk and k'} (part 2)). We prove the following lower bound.

\begin{theorem}
\label{thm:delta lower bound in terms of rk and also k''}
Let $f: \pmone^n \to \pmone$ be any function such that $k(f) > 1$.
Then,
\[
\delta(f)=\Omega\left(\max\left\{\frac{1}{\sparsity(f)}\left(\frac{r(f)}{\log \sparsity(f)}\right)^2, \frac{r(f)}{\seerank(f)\log \sparsity(f)}\right\}\right).
\]
\end{theorem}
Theorem~\ref{thm:delta lower bound in terms of rk and also k''} yields a better lower bound than Chang's lemma with the threshold $t=\seerank(f)$ (Equation~\eqref{eq: changs lemma with threshold k''}), except when $r(f) < (\log k(f))^2$ (see the caption of Figure~\ref{fig: k''}).
Theorem~\ref{thm:delta lower bound in terms of rk and also k''} presents two lower bounds: the first one is Theorem~\ref{thm:delta lower bound in terms of rk only}, and the second one is Lemma~\ref{lem:delta lower bound in terms of rk and also k''}. We prove Lemma~\ref{lem:delta lower bound in terms of rk and also k''} by strengthening a bound due to~\cite{CHLT19} on the sum of absolute values of level-$1$ Fourier coefficients of a Boolean function in terms of its $\ftwo$-degree. A proof of Theorem~\ref{thm:delta lower bound in terms of rk and also k''} can be found in Section~\ref{sec:lower bound using k''}.

We also show that for nearly all admissible values of $r(f), k(f)$ and $k''(f)$, there exist functions for which the larger of the two bounds presented in Theorem~\ref{thm:delta lower bound in terms of rk and also k''} is nearly tight.

\begin{theorem}
\label{thm:delta upper bound in terms of rk and also k''}
For all $\rho, \kappa, \kappa'' \in \N$ such that $\kappa$ is sufficiently large, for all $\epsilon > 0$ such that  $\log \kappa \leq \rho \leq \kappa^{\frac{1}{2} - \epsilon}$ and $\rho \leq \kappa'' \leq \kappa$ there exists a Boolean function $f_{\rho,\kappa,\kappa''}$ such that 
$r(f_{\rho,\kappa, \kappa''}) = \Theta(\rho)$, $k(f_{\rho,\kappa, \kappa''})= \Theta(\kappa)$, $k''(f_{\rho,\kappa, \kappa''}) = \Theta(\kappa'')$ and
$$\delta(f_{\rho,\kappa, \kappa''}) = \Theta\left( \max\left\{ \frac{1}{\kappa}\paren{\frac{\rho}{\log \kappa}}^2, \frac{\rho}{\kappa''\log (\kappa''/\rho)} \right\} \right).$$
\end{theorem}

The range of parameters considered in Theorem~\ref{thm:delta upper bound in terms of rk and also k''} is justified by Lemma~\ref{lem:relationships between rk and k'}.
Theorem~\ref{thm:delta upper bound in terms of rk and also k''} is proved in two parts. Fix any $\rho, \kappa$ such that $\log \kappa \leq \rho \leq \kappa^{\frac{1}{2} - \epsilon}$ for some constant $\epsilon > 0$. First, for each value of $\kappa'' \in [\frac{\kappa \log \kappa}{\rho}, \kappa]$ we construct a function $f$ for which the first lower bound on $\delta(f)$ from Theorem~\ref{thm:delta lower bound in terms of rk and also k''} is tight (Claim~\ref{claim:setting parameters for tightstraightline for k''}). In fact these are the same functions that are used to prove the first bound in Theorem~\ref{thm:delta upper bound in terms of rk and also k'}. Next, for each value of $\kappa'' \in [e\rho,\frac{\kappa \log \kappa}{\rho}]$ we construct a function $f$ for which $\delta(f) = \Theta(\frac{\rho}{\kappa''\log (\kappa''/\rho)})$ (Claim~\ref{claim: setting parameters for tight curve for k''}). From the above discussion one may verify that for every $\rho, \kappa$ that we consider and for every
$\kappa'' \geq \rho\cdot \kappa^{\Omega(1)}$,
the function that we construct witnesses tightness of the lower bound in Theorem~\ref{thm:delta lower bound in terms of rk and also k''}.

In general, for all settings of $\rho, \kappa$ and $\kappa''$ that we consider, the upper bound on $\delta(f)$ from~Theorem~\ref{thm:delta upper bound in terms of rk and also k''} is off by a factor of at most $O(\log \kappa)$ from the lower bound in Theorem~\ref{thm:delta lower bound in terms of rk and also k''}.

See Figure~\ref{fig: k''} for a graphical visualization of the bounds in Theorem~\ref{thm:delta lower bound in terms of rk and also k''} for any fixed values of $\rho$ and $\kappa$.


\begin{figure}[h]
\centering
\begin{tikzpicture}[scale = 1, xscale = 1.8, yscale =.7]
\def\dummyvar{2};
\def\rcval{2.5}; 
\def\kval{100}; 
\def\rval{(\rcval)^2};
\def\logkval{log2 \kval}; 
\def\rkval{\rval^2/\kval};

    \begin{axis}[
    axis y line = left,
    axis x line = bottom,
    x label style={at={(axis description cs:.25,-0.02)},anchor=north},
    y label style={at={(axis description cs:-0.01,.5)},rotate=0,anchor=south},
    xlabel={max-Entropy ($\kappa''$) $\longrightarrow$},
    ylabel={weight ($\delta$) $\longrightarrow$},
    xtick       = {(5/.66), 12.21, 100},
    xticklabels = {$\sqrt{\rho\kappa^{\frac{1}{\sqrt{\rho}}}}\ \ \ \ \ \ \ \ \ $, $\ \ \ \ \ \ \frac{\kappa}{\rho}\log \kappa$,$\kappa$},
    ytick       = {1},
    yticklabels = {1},
    xmin = 0, xmax = 102,
    ymin = 0, ymax = 1,
  ]
  
        \addplot[name path=testa, domain=0:100, color = ao, dashed, thick, samples=100]{\rkval} node[above,pos=.8] {\textcolor{black}{\textbf{\color{ao}$k$-line:\color{black}} $\ \delta = \rho^2/(\kappa\log^2 \kappa)$}}; 
  
      \addplot[name path=testa1, domain=12.21:100, color = green, smooth, ultra thick, samples=100]{\rkval} node[above,pos=.88]{}; 

        \addplot[name path=testb, domain=0:100, color = blue, dashed, thick, samples=100]{.22 + \rval/(3*(x))} node[above, pos=.78] {\textcolor{black}{\textbf{\color{blue}$k''$-curve:\color{black}} $\ \delta = \rho/(\kappa''\log \kappa)$}}; 
 
   \addplot[name path=testb1, domain=(5/.66):12.21, color = blue, ultra thick, samples=100]{.22 + \rval/(3*(x))} node[above, pos=.9]{}; 
 
       \addplot[name path=chang, domain=3.5:100, color = red, dashed, thick, samples=100]{\rcval/(x/1.5)} node[above, pos=0.7] {\textcolor{black}{\textbf{\color{red}CL-$k''$-curve:\color{black}} $\ \delta = \sqrt{\rho}/\left(\kappa''\log\left(\kappa''^2/\rho\right)\right)$}}; 
       
       \addplot[name path=chang1, domain=3.5:(5/.66), color = red, ultra thick, samples=100]{\rcval/(x/1.5)} node[above, pos=0.88] {\textcolor{black}{}}; 
  
  
         \addplot [name path=vertline, dotted, thick, samples=100] coordinates {(12.21,0) (12.21, 1)};
       
            \addplot [name path=vertline, dotted, thick, samples=100] coordinates {((5/.66),0) ((5/.66), 1)};

       \addplot[name path=upper, domain=0:100, color = white, samples=100]{1};

        \addplot fill between[ 
    of = chang and upper, 
    split, 
    every odd segment/.style  = {gray!70}
  ];
        
        \addplot fill between[ 
    of = chang and testb, 
    split, 
    every even segment/.style = {gray!20},
    every odd segment/.style  = {gray!20}
  ];
  
          \addplot fill between[ 
    of = testa and chang, 
    split, 
   every even segment/.style = {white!60},
    every odd segment/.style  = {gray!20}
  ];
  
  \node[anchor=west] (source) at (axis cs:25,.9){Chang's lemma Bound};
       \node (destination) at (axis cs:5,.75){};
       \draw[->](source)--(destination);
  
  \node[anchor=west] (source2) at (axis cs:45,.7){Our Bounds};
       \node (destination2) at (axis cs:50,.39){};
       \node (destination3) at (axis cs:10,.458){};
       \draw[->](source2)--(destination2);
       \draw[->](source2)--(destination3);

\end{axis}
\end{tikzpicture}
\caption{This plot is constructed for any fixed values of $\rho, \kappa$ for which $\log \kappa \leq \rho \leq \sqrt{\kappa}$, and depicts the relationship between $\delta(f)$ and $k''(f)$ for functions $f$ with $r(f) = \Theta(\rho)$ and $k(f) = \Theta(\kappa)$. For any fixed values of $\rho, \kappa$, we will refer to this plot as $(\rho,\kappa)$-$k''$-plot. Chang's lemma implies that Boolean functions lie above the CL-$k''$-curve. Theorem~\ref{thm:delta lower bound in terms of rk and also k''} improves upon Chang's lemma and shows that Boolean functions lie above both the $k$-line and the $k''$-curve, highlighted by the dark grey region in the figure.   Although the picture indicates that the CL-$k''$-curve is better than the $k''$-curve for certain ranges of $\kappa''$, this is actually only possible for certain values of $\rho$ and $\kappa$. This is because the CL-$k''$-curve and the $k''$-curve intersect at $\sqrt{\rho \kappa^{1/\sqrt{\rho}}}$, which is less than $\sqrt{\rho}$ if $\rho\geq (\log \kappa)^2$.
By Lemma~\ref{lem:relationships between rk and k'} we know that for any function $f$ on this plot, the range of $k''(f)$ is between $\max\{\sqrt{\rho}, \rho/\log \kappa\}$ and $\kappa$. Thus our bounds in Theorem~\ref{thm:delta lower bound in terms of rk and also k''} dominate those given by the CL-$k''$-curve in all $(\rho, \kappa)$-$k''$ plots where $\rho \geq \log^2 \kappa$.}
\label{fig: k''}
\end{figure}
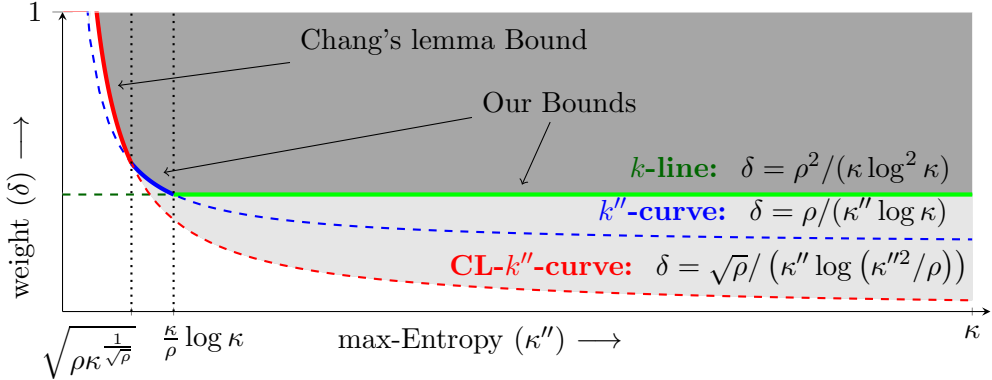


\end{enumerate}

\paragraph{Dominating Chang's lemma for all thresholds.} Our final contribution is to show that there exists a function for which: our lower bounds (Theorem~\ref{thm:delta lower bound in terms of rk and also k'} and \ref{thm:delta lower bound in terms of rk and also k''}) asymptotically match its weight, but for any choice of the threshold the lower bound obtained from Chang's lemma (Lemma~\ref{lem:chang}) is asymptotically smaller than its weight (Claim~\ref{claim: beating changs lemma for all thresholds for ADtt'}).

\subsection{Applications of our results}
\label{sec: Applications of our results}

An application of our result is an enhanced understanding of the bound $r(f)=O(\sqrt{ k(f)} \log k(f))$ proven by Sanyal \cite{San19}. This bound is a special case of Theorem~\ref{thm:delta lower bound in terms of rk only} for $\delta(f)=\Theta(1)$. It is not known whether the $\log k(f)$ term is required in Sanyal's upper bound on $r(f)$ (when $f$ equals the Addressing function, $r(f) = \Omega(\sqrt{k(f)})$, see Definition~\ref{defi:Addressing} and Observation~\ref{obs:properties of AND, Bent and Addressing}). For all the functions we construct witnessing the tightness of the bound in Theorem~\ref{thm:delta lower bound in terms of rk only}, $\delta(f)=o(1)$. We prove Theorem~\ref{thm:delta lower bound in terms of rk only} by generalizing Sanyal's proof. As stated before, our bound is tight in this generality, i.e.~the logarithmic factor is required in the upper bound on $r(f)$. This sheds light on the presence of the logarithmic term in the bound $r(f)=O(\sqrt {k(f)} \log k(f))$.

Also, Fourier sparsity and Fourier rank of $f$ have intimate connections with the communication complexity of functions of the form $F:=f \circ \XOR$. The Fourier sparsity of $f$ equals the real rank ($\mathsf{rank}(M_F)$) of the communication matrix $M_F$ of $F$, and the Fourier rank of $f$ equals the deterministic (and even exact quantum) one-way communication complexity of $F$ \cite{MO09}. Theorem~\ref{thm:delta lower bound in terms of rk only} thus implies an improved upper bound of $O(\sqrt{k(f) \delta(f)} \log k(f))$ on the one-way communication complexity of $F$ in these models, which asymptotically beats the best known upper bound of $O(\sqrt{\mathsf{rank}(M_F)})$ even for two-way protocols~\cite{TWXZ13, Lovett16}, for the special case of functions of this form (when $\delta(f)=o(1/\log k)$).

Given the wide-ranging application of Chang's lemma to areas like additive combinatorics, learning theory and communication complexity, we strongly feel that our refinements  of Chang's lemma will find many more applications.


\section{Proof techniques for lower bound results}\label{sec: lb overview}

Our lower bound results on $\delta(f)$ can be divided into two parts: lower bounds in terms of $r(f)$, $k(f)$, and $k'(f)$ (Theorem~\ref{thm:delta lower bound in terms of rk and also k'}), and lower bounds in terms of $r(f)$, $k(f)$, and $k''(f)$ (Theorem~\ref{thm:delta lower bound in terms of rk and also k''}).

Theorem~\ref{thm:delta lower bound in terms of rk and also k'} consists of two lower bounds. The second bound, $\delta({f}) = \Omega\left(\frac{k(f)}{k'(f)^2}\right)$, is a direct application of Parseval's identity (Claim~\ref{claim:delta at least k/k'^2}). The first bound follows from Theorem~\ref{thm:delta lower bound in terms of rk only}, one of the main technical contributions of this paper. Similarly, Theorem~\ref{thm:delta lower bound in terms of rk and also k''} consists of two lower bounds: the first bound is Theorem~\ref{thm:delta lower bound in terms of rk only} and the second one is Lemma~\ref{lem:delta lower bound in terms of rk and also k''}.

The formal proofs of Theorems \ref{thm:delta lower bound in terms of rk and also k'} and \ref{thm:delta lower bound in terms of rk and also k''} are given in Section~\ref{sec:lower bound}. We discuss the outline of the proofs of Theorem~\ref{thm:delta lower bound in terms of rk only} and Lemma~\ref{lem:delta lower bound in terms of rk and also k''} in Sections \ref{sec:overview of delta lower bound in terms of rk only} and \ref{sec:overview of delta lower bound in terms of rk and also k''}, respectively.

\subsection{Overview of the proof of Theorem~\ref{thm:delta lower bound in terms of rk only}}
\label{sec:overview of delta lower bound in terms of rk only}

The lower bound on $\delta(f)$ in Theorem~\ref{thm:delta lower bound in terms of rk only} can also be viewed as an upper on $r(f)$ in terms of $\delta(f)$ and $k(f)$. The best known upper bound on the Fourier rank of a Boolean function in terms of the sparsity of the function was given by Sanyal~\cite{San19}. They showed that for any Boolean function $f$, $r({f}) = O(\sqrt{k(f)} \log{k(f)})$. Theorem~\ref{thm:delta lower bound in terms of rk only} improves upon this upper bound on the Fourier rank by adding a dependence on the weight of the function: $r({f}) = O(\sqrt{\delta(f) k(f)} \log{k(f)})$.

The outline of the proof is similar to the proof by Sanyal (\cite[Theorem 1.2]{San19}).
We give an algorithm which takes a Boolean function as an input and outputs $O(\sqrt{\delta(f) \sparsity(f)} \log{\sparsity(f)})$ parities, such that, any assignment of these parities makes the function constant. This gives an upper bound on Fourier rank of the function since the Fourier support of the function must be contained in the span of this set of parities (Observation~\ref{obs: napdt_implies_rank_ub}). 
The central ingredient in the algorithm is a lemma in~\cite[Lemma 28]{TWXZ13}.
\begin{lemma}[\cite{TWXZ13}]
\label{lem:TWXZ13}    
    Let $f: \pmone^n \to \pmone$ a function. There is an affine subspace $V \subseteq \pmone^n$ of co-dimension at most $3\sqrt{\delta(f) k(f)}$ such that $f$ is constant on $V$.
\end{lemma}
The above lemma is stated slightly differently in \cite{TWXZ13}; they use $\|\wh{f}\|_1$ to bound the co-dimension instead of $3 \sqrt{\delta(f) k(f)}$ (Claim~\ref{claim:lone_le_sqrt_kdelta}).
Lemma~\ref{lem:TWXZ13} allows us to fix small number of parities, such that, the sparsity of every possible restriction (for all possible assignments to these parities) is halved. The algorithm is formally stated in Section~\ref{sec:lower bound using k'} (Algorithm~\ref{alg:NAPDT}); we give an outline here.

\paragraph{Outline of the algorithm:}
Our iterative algorithm incrementally constructs a set of parities such that, finally, the function becomes constant for every assignment of these set of parities. Every iteration, implemented as a \textnormal{\textbf{while}} loop in Algorithm~\ref{alg:NAPDT}, is essentially an application of Lemma~\ref{lem:TWXZ13}. 

Let $\Gamma$ be the set of parities fixed after a certain number of iterations of the \textnormal{\textbf{while}} loop. For the next iteration of the loop, we ``greedily'' pick a function, out of all possible restrictions corresponding to $2^{|\Gamma|}$ possible assignments, of $\Gamma$. We then find a set of parities such that the greedily picked function becomes constant under some assignment of these parities; a small set of such parities exist (Lemma~\ref{lem:TWXZ13}) and we include these parities in $\Gamma$.
The algorithm finishes once all possible restrictions of $f$, corresponding to $\Gamma$, become constant. 
The termination condition implies that the algorithm outputs a set of parities satisfying the required condition. 

\paragraph{Completing the proof of Theorem~\ref{thm:delta lower bound in terms of rk only}:}
It remains to show is that the number of parities fixed in Algorithm~\ref{alg:NAPDT} is small.
Given a Boolean function $f$ and a set of parities $\Gamma$ over the set of the variables of $f$, following equivalence relation over $\supp(f)$ arises naturally:
        $$
        \forall \gamma_1, \gamma_2 \in \supp(f),
        \gamma_1 \equiv \gamma_2\ \text{iff}\ \gamma_1+\gamma_2 \in \spann(\Gamma).
        $$

Let us denote $\Gamma$ after the $i$-th iteration of the \textnormal{\textbf{while}} loop by $\Gamma^{(i)}$ ($\Gamma^{(0)} = \emptyset$). Let $\fmin^{(i)}$ be the selected function $\fmin$ after the $i$-th iteration ($\fmin^{(0)} = f$).

To bound the total number of parties fixed in Algorithm~\ref{alg:NAPDT}, we would like to bound the number of parities included
in the $i$-th iteration of the \textbf{while} loop. In Step~\ref{item: step_a}, the algorithm chooses the minimum 
number, say $q_i$, of parities such that
$\fmin^{(i-1)}$ becomes constant 
after fixing these parities to some assignment.
$\Gamma$ is updated with these parities to obtain $\Gamma^{(i)}$. Let $\ell_i$ be the number of partitions of the Fourier support of $f$ with respect to the equivalence relation corresponding to $\Gamma^{(i)}$.

In Step~\ref{item: step_b} the algorithm considers all possible assignments of parities in $\Gamma^{(i)}$ and the corresponding restrictions of $f$. A non-constant restriction with the smallest weight-to-sparsity ratio is chosen to be $\fmin$. 

The main idea for the analysis of Algorithm~\ref{alg:NAPDT} is to upper bound the ratio of $q_i$ and $(\ell_{i-1} - \ell_i)$ for every iteration $i$. On one hand
$\frac{q_i}{(\ell_{i-1} - \ell_{i})}$ is at most the square root of weight-to-sparsity ratio of the chosen $\fmin$ (Lemma~\ref{lem:q_i bound} and its proof).
On the other hand, Lemma~\ref{lem:main_lemma} (main technical lemma in this proof, outline of the proof in the next paragraph) ensures that the weight-to-sparsity ratio of $\fmin$ can be upper bounded by $O\left(\frac{\delta(\fmin)k(\fmin)}{\ell_{i-1}^2}\right)$. Thus, we show that for every iteration $i$, $q_i$ is upper bounded by 
\[
    O\left( \frac{\sqrt{\delta(\fmin) \sparsity(\fmin)}}{\ell_{i-1}}(\ell_{i-1} - \ell_{i}) \right).
\]
Using standard arguments, summing over the iterations, we get a bound of $O(\sqrt{\delta(f) k(f)}\log \ell_0)$ on $|\Gamma|$.
Since $\ell_0 = k$, 
the desired upper bound on $r(f)$ in Theorem~\ref{thm:delta lower bound in terms of rk only} follows from Observation~\ref{obs: napdt_implies_rank_ub}.

\paragraph{Outline of the proof of Lemma~\ref{lem:main_lemma}:} Given a Boolean function $f$ and a set of parities $\Gamma$, let $\ell$ be the number of equivalence classes for the equivalence relation corresponding to $\Gamma$.
Define $f|_{{(\Gamma, b)}} := f|_{\{x \in \pmone^n: \forall \gamma \in \Gamma, \chi_{\gamma}(x) = b_{\gamma}\}}$ to be the restricted function when parities of $\Gamma$ are set to assignment $b$ in function $f$.
Lemma~\ref{lem:main_lemma} states that there exists an assignment $b$ of $\Gamma$ such that the restricted function $g_b = f|_{{(\Gamma, b)}}$ satisfies
$$
\frac{\delta(g_{{(b)}})}{
        k(g_{{(b)}})
        } \leq \frac{4 k(f) \delta(f)}{\ell^2} .$$

In contrast to the proof in \cite{San19}, where we only need to find a restriction with large sparsity, we need to balance both $\delta(g_b)$ and $\sparsity(g_b)$ here.\footnote{For technical reasons, we consider sparsity without the empty Fourier coefficient in this proof.}

We show that $\Ex{b}{\delta(g_{b})}= \delta(f)$ and $\Ex{b}{\sparsity(g_{b})} \geq \ell^2/4k(f)$, where $b$'s are picked uniformly from $\pmone^\Gamma$. A careful manipulation of these expected values gives the required $b$. 

The equality $\Ex{b}{\delta(g_{b})}= \delta(f)$ follows by the observation that the set of inputs of $f$ is partitioned by the set of inputs of $g_b$.
For the expectation of the sparsity of $g_b$, observe that the Fourier coefficients of $g_b$ are non-zero polynomials over the parities of $\Gamma$. By the uncertainty principle (Lemma~\ref{lem:uncertainity principle}), any Fourier coefficient is non-zero for a large number of $g_b$'s. Summing up these lower bounds for all Fourier coefficients, we get that the total number of non-zero Fourier coefficients (for all $g_b$) is large. This shows the required lower bound on the expectation of the sparsity of $g_b$, finishing the proof of Lemma~\ref{lem:main_lemma}.

\subsection{Overview of the proof of Lemma~\ref{lem:delta lower bound in terms of rk and also k''} (for Theorem~\ref{thm:delta lower bound in terms of rk and also k''})}
\label{sec:overview of delta lower bound in terms of rk and also k''}
The crucial ingredient to prove the lower bound in Lemma~\ref{lem:delta lower bound in terms of rk and also k''} is the following lemma.

\begin{lemma}
\label{lem:CHLT_improvement}
    For any Boolean function $f$, 
    $\sum_{i=1}^n |\widehat{f}(i)| = O(\delta(f) \degtwo(f))$.
\end{lemma}
 
This lemma is a refinement of a similar theorem proved in \cite{CHLT19} (Theorem~\ref{thm:CHLT}) which does not contain the factor of $\delta(f)$. The proof of Lemma~\ref{lem:CHLT_improvement} for a Boolean function $f$ essentially applies Theorem~\ref{thm:CHLT} on the XOR of disjoint copies of $f$. 

Lemma~\ref{lem:CHLT_improvement} shows a bound on the sum of absolute values of level-$1$ Fourier coefficients for the standard basis of the Fourier support of $f$; we extend this bound for any basis of span of the Fourier support of $f$ (Corollary~\ref{cor:improved chlt-implication}). The proof essentially constructs another function $h$ by doing a basis change on parities, and then applies Lemma~\ref{lem:CHLT_improvement} on the function $h$.

Lemma~\ref{lem:delta lower bound in terms of rk and also k''} is a direct implication of Corollary~\ref{cor:improved chlt-implication}; observe that every Fourier coefficient on the left hand side of Corollary~\ref{cor:improved chlt-implication} is bigger than $1/k''(f)$ (from the definition of $k''(f)$).


\section{Proof techniques for upper bound results}\label{sec: ub overview}

In this section we give the overview of our two upper bound results, Theorems~\ref{thm:delta upper bound in terms of rk and also k'} and \ref{thm:delta upper bound in terms of rk and also k''}. 
For presenting the overview of the proofs of these theorems we
will use $(\rho, \kappa)$-$k'$-plots (Figure~\ref{fig: k'}) and $(\rho, \kappa)$-$k''$-plots (Figure~\ref{fig: k''}), respectively. In a $(\rho, \kappa)$-$k'$-plot  ($(\rho, \kappa)$-$k''$-plot, respectively) we will refer to the ``intersection point'' as the point of intersection between the $k$-line and $k'$-curve (the point of intersection between the $k$-line and $k''$-curve, respectively). Which intersection point we are referring to will be clear from context. 

\subsection{Proof techniques for Theorem~\ref{thm:delta upper bound in terms of rk and also k'}}
\label{sec: proof of upper bound on delta for k'}
   
     To prove Theorem~\ref{thm:delta upper bound in terms of rk and also k'}, we split our goal into two natural parts: constructing functions on the $k$-line and constructing functions on the $k'$-curve. 
   Both the classes of functions 
   are modifications of the Addressing function (Definition~\ref{defi:Addressing}). In these modifications, all or some of the target variables of the Addressing function are replaced with an AND function or a Bent function or a combination of them. 
    We first provide a description of some functions that lie on the intersection point.
    While we do not require this, we choose to describe these functions in order to provide more intuition. 
    \paragraph{Construction of functions at the intersection point in any $(\rho, \kappa)$-$k'$-plot:}  
    Note that a function lies at the intersection point when
    \begin{equation}\label{eq:intersection}
        k'(f) = \frac{k(f) \log(k(f))}{r(f)}.
    \end{equation} Thus, we want to construct a function $f$ with $k(f)=\Theta(\kappa)$, $r(f)=\Theta(\rho)$, $k'(f) = \Theta\left(\frac{\kappa\log \kappa}{\rho}\right)$ and 
    $\delta(f) = \rho^2/\kappa(\log^2 \kappa)$. In particular, we want to construct such functions for all $\rho, \kappa$ satisfying $\log \kappa \leq \rho \leq \kappa^{\frac{1}{2}}$. Note that, 
    the Addressing function $\AD_t : \pmone^{\log t + t} \to \pmone$ has sparsity $t^2$, rank $(t + \log t)$, max-supp-entropy $t$ and  weight $1/2$ (Observation~\ref{obs:properties of AND, Bent and Addressing}) and thus,
    $\AD_t$ satisfies     
    Equation~\eqref{eq:intersection}. This only gives functions on the intersection point on all $(\rho, \kappa)$-$k'$-plots where $\rho = \Theta(\sqrt{\kappa})$, while we have to exhibit such functions for all $(\rho, \kappa)$-$k'$-plots where $\log \kappa \leq \rho = O\bra{\sqrt{\kappa}}$. 
    
        Our next step is to tweak $\AD_t$ in such a way that the rank of the new function $f$ does not change significantly while the sparsity and max-supp-entropy both increase by the same multiplicative factor. This would ensure that the resulting function satisfies Equation~\eqref{eq:intersection}. If the resulting function's weight decreases to the required value,
        we would have a function at the intersection point.
         
        In order to tweak $\AD_t$, we consider a special kind of composed function $f:=  \AD_t \circt g$,\footnote{see Definition~\ref{defi:composedaddressing}
        for a precise definition} obtained by replacing each target variable in the addressing function with a function $g$ where each copy of $g$ acts on a set of new variables. Lemma~\ref{lem:properties of composition of addressing and g} gives the properties of such composed functions. Due to the structure of the Fourier spectrum of the Addressing function, Lemma~\ref{lem:properties of composition of addressing and g} gives us $r(f)  \approx t \cdot r(g)$, $k(f)  \approx t^2 \cdot k(g)$, $k'(f)  = t \cdot k'(g)$ and $\delta(f)  = \delta(g)$.
                
        So, if $g$ is a function on a small number of variables (say $\log t'$) with near-maximal sparsity and max-supp-entropy ($\Theta(t')$), then the resulting function satisfies Equation~\eqref{eq:intersection}. 
        The $\AND$ function is a natural choice for $g$. We denote the resulting function by $\AD_{t, t'}$ 
        (Definition~\ref{defi:ADtt'}),
      and this is a function at the intersection point for all plots by suitably varying $t$ and $t'$.
      
    \paragraph{Constructing functions on the $k$-line:} 
    We start with $\AD_{t,t'}$, the function at the intersection point in $(\rho, \kappa)$-$k'$-plots.  
    We modify $\AD_{t, t'}$ in such a way that its sparsity, rank and weight do not change much, while the max-supp-entropy increases. We replace a single $\AND_{\log t'}$ in $\AD_{t, t'}$ by $\AND_{\log a}$ for some suitable $a > t$,
    denote the new function by $\AD_{t,t',a}$ (Definition~\ref{defi:ADtt'a}). 
    A suitable setting of the parameters $t, t'$ and $a$ yields functions on the $k$-line for all plots (Claim~\ref{claim:setting parameters for tightstraightline}).

    \paragraph{Constructing functions on the $k'$-curve of the $(\rho, \kappa)$-$k'$-plot:} 
    We start with $\AD_{t, t'}$ at the intersection point
    on $(\rho, \kappa/\ell)$-$k'$-plot (for some parameter $\ell > 0$). We modify $\AD_{t, t'}$ in such a way that its rank and weight do not change, the sparsity increases by a multiplicative factor of $\ell$ and the max-supp-entropy increases by a factor of $\sqrt{\ell}$. The new function $f$ will be on the $k'$-curve in the $(\rho, \kappa)$-$k'$-plot because $\frac{k(f)}{k'(f)^2} = \frac{k(\AD_{t,t'})}{k'(\AD_{t,t'})^2} = \delta(\AD_{t,t'}) = \delta(f)$. Note that $k'(f) \approx \frac{\kappa\log(\kappa)}{\rho \sqrt{\ell}}$, thus making $\ell$ suitably large yields functions on the $k'$-curve for all $\rho\leq \kappa' \leq \frac{\kappa\log(\kappa)}{\rho}$ for all plots.
    
    We now change $\AD_{t,t'}$ to have the properties mentioned above.
    We modify each $\AND_{\log t'}$ in $\AD_{t, t'}$ as follows: replace a single variable $x$ by $x \cdot B$, where $B$ is a bent function on $\log \ell$ new variables. We  
    denote this new inner function by $\AB$ (Definition~\ref{defi:AB}), and  $\AD_{t} \circt \AB$ by
    $\AAB$ (Definition~\ref{defi:AAB}). The effect of changing $\AND_{\log t'}$ to $\AB$ keeps its rank and weight roughly the same, while increasing its sparsity by a factor of $\ell$ and increasing its max-supp-entropy by a factor of $\sqrt{\ell}$ (Claim~\ref{claim:properties of AND of Bent}). In Claim~\ref{claim:properties of AAB} we show, using our composition lemma (Lemma~\ref{lem:properties of composition of addressing and g}), that the properties of $\AD_t \circt \AND_{\log t'}$ and $\AD_t \circt \AB$ change in a similar fashion. Thus, a suitable setting of the parameters $t, t', \ell$ yields functions on the $k'$-curve for all plots (Claim~\ref{claim:setting parameters for tight curve}).

\subsection{Proof techniques for Theorem~\ref{thm:delta upper bound in terms of rk and also k''}} 
\label{sec: proof of upper bound on delta for k''}

We split our goal into two parts: constructing functions on the $k$-line when $\frac{\kappa}{\rho}\log\kappa \leq \kappa'' \leq \kappa$, and constructing functions on the $k''$-curve when $\kappa \leq \frac{\kappa}{\rho}\log\kappa$.
To construct functions on the $k$-line, we use the functions $\AD_{t,t',a}$ constructed for the proof of Theorem~\ref{thm:delta upper bound in terms of rk and also k'}, since $k'(\AD_{t,t',a}) = k''(\AD_{t,t',a})$.

For constructing functions on the $k''$-curve, we need to construct functions $f$ such that 
\begin{equation}\label{eq:ub k'' curve}
\delta(f) = \Theta\left(\frac{r(f)}{k''(f)\log\left(k''(f)/r(f)\right)}\right). 
\end{equation}
We will use a similar technique as in our construction of functions on the $k'$-curve in Theorem~\ref{thm:delta upper bound in terms of rk and also k'}. 
We start from the function $\AD_{t,t'}$ at the intersection point. Note that $\AD_{t,t'}$ satisfies Equation~\eqref{eq:ub k'' curve}. We modify $\AD_{t,t'}$ such that
the rank, weight and max-rank-entropy changes very little but the sparsity increases by a multiplicative parameter $2^p$. We achieve this by replacing a variable (say $x$) in $\AD_{t,t'}$ with $x\cdot \AND(y_1, \dots, y_p)$, where $x$ and $y_i$s are all variables in $\AD_{t,t'}$, but for any $i$, $x$ and $y_i$ do not appear in the same monomial (Claim~\ref{claim: setting parameters for tight curve for k''}). The new function $f$ still satisfies Equation~\eqref{eq:ub k'' curve}. This places $f$ on the $k''$-curve in a plot corresponding to the same rank as that of $\AD_{t, t'}$, but where the sparsity increases by a factor of $2^p$. By suitably setting $p$, $t$ and $t'$, we obtain functions on the $k''$-curve for all plots. This proves the second bound in Theorem~\ref{thm:delta upper bound in terms of rk and also k''}.


\section{Preliminaries}\label{sec:prelims}
    All logarithms in this paper are taken to be base 2. We use the notation $[n]$ to denote the set $\cbra{1, 2, \dots, n}$. When necessary, we assume $t$ is a power of $2$. We use the notation $1^n$ (respectively, $(-1)^n$) to denote the $n$-bit string $(1, 1, \dots, 1)$ (respectively, $(-1, -1, \dots, -1)$).  

    For a function $f : \pmone^n \to \pmone$, its $\ftwo$-degree, denoted by $\degtwo(f)$, is the degree of its unique $\ftwo$-polynomial representation.
    Throughout this paper, we often identify subsets of $[n]$ with 
    their corresponding characteristic vectors in $\ftwo^n$.  
    Thus when we refer to linear algebraic measures of a collection of subsets of $[n]$, we mean the measure on the corresponding subset of $\ftwo^n$ (where $\ftwo^n$ is viewed as an $\ftwo$-vector space). 
    
    Throughout this paper, we assume that $f$ is not a constant function or a parity or a negative parity, unless mentioned otherwise.
    
    \subsection{Fourier analysis of Boolean functions}
    Consider the vector space of functions from $\fcube{n}$ to $\R$ equipped with the following inner product.
    \[
    \abra{f, g} := \frac{1}{2^n}\sum_{x \in \fcube{n}}f(x)g(x).
    \]
    For a set $S \subseteq [n]$, define a \emph{parity} function (which we also refer to as \emph{characters}) $\chi_S : \pmone^{n} \to \pmone$ by $\chi_S(x) = \prod_{i \in S}x_i$. The set of parity functions $\cbra{\chi_S : S \subseteq [n]}$ forms an orthonormal basis for this vector space.  Hence, every function $f : \fcube{n} \to \R$ has a unique representation as
    \[
    f = \sum_{S \subseteq [n]}\wh{f}(S)\chi_S,
    \]
    where $\wh{f}(S) = \langle f, \chi_S \rangle$ for all $S \subseteq [n]$. The coefficients $\cbra{\wh{f}(S) : S \subseteq [n]}$ are called the \emph{Fourier coefficients} of $f$. Define the Fourier $\ell_1$-norm of a function $f : \pmone^n \to \mathbb{R}$ by $\lone{\wh{f}} := \sum_{S \subseteq [n]} |\wh{f}(S)|$.
    \begin{defi}[Weight of a Boolean function]
    \label{defi:weight}
        Let $f: \pmone^{n} \to \pmone$ be any function. The \emph{weight} of $f$, denoted by $\delta(f)$, is defined as 
            \[
                \delta(f) = \Pr_{x \in \pmone^n}[f(x) = -1].
            \]
    \end{defi}
    The following observation follows from the fact that $\wh{f}(\emptyset) = \frac{1}{2^n}\sum_{x \in \pmone^n}f(x)$.
    \begin{observation}
    \label{obs:weight, empty Fourier}
        Let $f : \pmone^n \to \pmone$ be any function. Then,
        \[
        \wh{f}(\emptyset) = 1 - 2 \delta(f).
        \]
    \end{observation}
     \begin{defi}[Fourier Support]
     \label{defi: Fouriersupport}
        Let $f:\pmone^n \to \R$ be any function. The Fourier support of $f$, denoted by $\supp(f)$, is defined as
        \[
            \supp(f) = \cbra{S\subseteq[n] : \wh{f}(S) \neq 0}.
        \]
    \end{defi}
    \begin{remark}
    In the literature, Fourier support is generally denoted by $\supp(\wh{f})$. For ease of notation we drop the hat symbol above $f$. A similar convention has been adopted in Definitions~\ref{defi:sparsity},~\ref{defi:rank}, and~\ref{defi:max entropy, max rank entropy}.
    \end{remark}
    
    For ease of notation, we sometimes abuse notation and say that the elements of the Fourier support of $f$ are the characters $\cbra{\chi_S : S \subseteq [n], \wh{f}(S) \neq 0}$, rather than the corresponding sets as given in Definition~\ref{defi: Fouriersupport}.
    \begin{defi}[Fourier sparsity]
    \label{defi:sparsity}
        Let $f:\pmone^n \to \R$ be any function. The
        Fourier sparsity of $f$, denoted by $k(f)$, is defined as
        \[
            k(f) = |\supp(f)|.
        \] 
    \end{defi}
    For simplicity we assume that $k(f) \geq 2$ for all Boolean functions $f$ considered in this paper (unless explicitly mentioned otherwise). We often simply refer to the Fourier sparsity as \emph{sparsity}.
    
    \begin{theorem}[Parseval's identity]
    \label{thm:Parseval}
        Let $f:\pmone^n \to \pmone$ be any function. Then,
        \[
            \sum_{S \subseteq[n]}\wh{f}(S)^2 = 1.
        \]
    \end{theorem}

        We require the following lemma (see, for example,~\cite{GT13}).
    \begin{lemma}[Uncertainty Principle]
    \label{lem:uncertainity principle}
    Let $f:\pmone^n \to \R$ be a polynomial and let $U_n$ denote the uniform distribution on $\pmone^n$. Then,
    \[
        \Pr_{x \sim U_n}[f(x) \neq 0] \geq \frac{1}{k(f)}.
    \]
    \end{lemma}
    \begin{lemma}[{\cite[Theorem 8.1]{GOS+}}]
    \label{lem:granularity}
        Let $f: \pmone^n \to \pmone$ be any function. Then, for all $S \subseteq [n]$, $|\wh{f}(S)|$ is an integral multiple of $2^{1 -\lfloor \log k(f) \rfloor}$.  
    \end{lemma}
    We also require the following lemma relating the $\ftwo{}$-degree of a Boolean function and its Fourier sparsity (see, for example,~\cite{BC99}).
    \begin{lemma}\label{lem:ftwodeg and Fourier sparsity}
        Let $f:\fcube{n} \rightarrow \fcube{}$ be any function with $k(f) > 1$. Then, 
        \[
        \degtwo(f) \leq \log k(f).
        \]
    \end{lemma}
    
The next claim shows that $\degtwo(f)$ does not change under a change of basis over the Fourier domain.
\begin{claim}
\label{thm:ftwo_deg_does_not_change}
Let $f:\pmone^n \to \pmone$ be any function and let $B\in \ftwo^{n \times n}$ be an invertible matrix.. Define the function $f_B:\pmone^n \to \R$ as
\[
\widehat{f_B}(\alpha) = \widehat{f}(B\alpha) \ \ \textnormal{for all~} \alpha \in \ftwo^n,
\]
 Then $f_B$ is Boolean valued and $\degtwo(f_B) = \degtwo(f)$.
\end{claim}
\begin{proof}
Viewing $f_B$ and $f$ as functions over the domain $\zone^n$ instead of $\pmone^n$, we get that this basis change over the Fourier domain amounts to applying $(B^{-1})^T$ on the input space (see~\cite[Lemma 4]{ACL+19}). In other words $f_B$ is Boolean valued, and if $p_{f_B}$ and $p_f$ are the $\ftwo$-polynomials representing $f_B$ and $f$, respectively, then $p_{f_B}(x) = p_f((B^{-1})^T x)$.

For all $x \in \ftwo^n$, let $p_f(x) = \sum_{\gamma \in \ftwo^n} \widehat{p_f}(\gamma) \prod_{i: \gamma_i = 1} x_i$.
If $(B^{-1})^T_{j}$ denotes the $j$-th 
row of $(B^{-1})^T$ (for $j \in [n]$), then $p_{f_B}$ has the unique representation 
\begin{align*}
    p_{f_B}(x) = \widehat{p_f}(\gamma) \prod_{i: \gamma_i = 1} \langle (B^{-1})^T_{i}, x \rangle .
\end{align*}
So, every variable appearing in the polynomial representation of $p_f$ is replaced by a linear combination (over $\ftwo$) of $x_i$'s in $p_{f_B}$. In particular, the degree of any monomial in the polynomial representation of $p_f$ is at least as large as the degree of its expansion in $p_{f_B}$, and hence $\deg(p_{f_B}) \leq \deg(p_f)$.

Since $B$ is invertible, the same argument shows $\deg(p_f) \leq \deg(p_{f_B})$.  Thus $\deg(p_{f_B}) = \deg(p_f)$, which implies $\degtwo(f_B) = \degtwo(f)$.
\end{proof}

The following corollary follows from \cite[Theorem 13]{CHLT19} and Lemma~\ref{lem:ftwodeg and Fourier sparsity}.

\begin{corollary}\label{cor:chlt-implication}
Let $f:\pmone^n \to \pmone$ be any function, and let $\cS \subseteq \supp(f)$ be a basis of $\spann(\supp(f))$. Then,
\[
\sum_{S \in \cS} |\wh{f}(S)| \leq 4\log k(f).
\]
\end{corollary}

    We now define notions of restriction of a function $f : \pmone^n \to \pmone$ to a subset $A \subseteq \pmone^n$.
    \begin{defi}[Restriction]
    Let $f: \pmone^n \to \pmone$ and $A \subseteq \pmone^n$. The restriction of $f$ to $A$ is the function $f|_A: A \to \pmone$ defined as $f|_A(x) = f(x)$ for all $x \in A$.
    \end{defi}

    \begin{defi}[Affine Restriction]
    Let $f:\pmone^n \rightarrow \pmone$, let $\Gamma$ be a set of parities and $b \in \pmone^{\Gamma}$ be an assignment to these parities. Define the function $f|_{\Vb}$ to be the restriction of $f$ to the affine subspace obtained by fixing parities in $\Gamma$ according to $b$. That is,
    \[ 
        f|_{\Vb} := f|_{\{x \in \pmone^n: \chi_{\gamma}(x) = b_{\gamma} \textnormal{~for all~} \gamma \in \Gamma\}}. 
    \]
    \end{defi}

    \subsection{Fourier expansions and properties of some standard functions}
    For any integer $n > 0$, define the function $\AND_n : \pmone^n \to \pmone$ by $\AND_n(x) = -1$ if $x = (-1)^n$, and 1 otherwise. We drop the subscript $n$ when it is clear from the context.
    We state the Fourier expansion of $\AND$ below without proof.
    \begin{fact}[Fourier expansion of $\AND$]
    \label{fact:Fourier AND}
        Let $n \geq 1$ be any positive integer. Then
        \begin{align*}
        \wh{\AND_n}(S) = 
            \begin{cases}
                1 - \frac{2}{2^n} & S = \emptyset,\\
                \frac{2 \cdot (-1)^{|S| + 1}}{2^n} & \text{otherwise}.
            \end{cases}
        \end{align*}
    \end{fact}
    \begin{defi}[Bent functions]
    \label{defi:bent}
    A function $f : \pmone^n \to \pmone$ is said to be a bent function if $|\wh{f}(S)| = |\wh{f}(T)|$ for all $S, T \subseteq [n]$.
    \end{defi}
    Using Parseval's identity (Theorem~\ref{thm:Parseval}) we get the following observation.
    \begin{observation}
    \label{obs: bentfcoeffs}
    Let $f : \pmone^n \to \pmone$ be a bent function. Then, $|\wh{f}(S)| = \frac{1}{\sqrt{2^n}}$ for all $S \subseteq [n]$.
    \end{observation}
    
    \begin{defi}[Indicator function]
    \label{defi:ind}
        For any integer $n \geq 1$ and $b \in \pmone^n$, define the function $\ind_b : \pmone^n \to \zone$ by
        \[
        \ind_b (x) = 
            \begin{cases}
                1 & x = b,\\
                0 & \text{otherwise}.
            \end{cases}
        \]
    \end{defi}

    We require the following observation about the Fourier expansion of Indicator functions, which we state without proof.
    \begin{observation}[Fourier expansion of Indicator functions]
    \label{obs:indexpansion}
    For any integer $n \geq 1$ and $b \in \pmone^n$, let $\ind_b$ be as in Definition~\ref{defi:ind}. Then,
    \[
        \wh{\ind}_b(S) = \frac{\prod_{i \in S}b_i}{2^n} \quad\text{for all } S \subseteq [n].
    \]
    \end{observation}
    \begin{defi}[Addressing function]\label{defi:Addressing}
        For any integer $t \geq 2$, define the Addressing function $\AD_{t} : \fcube{\log t} \times \fcube{t} \rightarrow \fcube{}$ by
        \[
            \AD_{t}(x,y) = y_{\bin(x)},
        \]
    where $x \in \pmone^{\log t}$ and $y \in \pmone^t$, and $\bin(x)$ denotes the integer in $[t]$ whose binary representation is given by $x$ (where $-1$'s are viewed as 1 in the string $x$, and $1$'s are viewed as 0). 
    We refer to the $x$-variables as \emph{addressing variables}, and the $y$-variables as \emph{target variables}.
    \end{defi}
    The following combinatorial observation is useful to us.
    \begin{observation}
    \label{obs:sum_b char = 0 }
    For any integer $n \geq 1$ and non-empty subset $S \subseteq [n]$,
        $$\sum_{b \in \pmone^{n}} \prod_{i \in S} b_i = 0.$$
\end{observation}
    We require the following representation of Addressing functions.
    \begin{observation}
    \label{obs:addexpansion}
    For any integer $t \geq 2$, $x \in \pmone^{\log t}$ and $y \in \pmone^t$, we have
    \[
    \AD_t(x, y) = \sum_{b \in \pmone^{\log t}} y_b \ind_b(x).
    \]
    \end{observation}
    
    We next define a way of modifying the Addressing function that is of use to us. In this modification, we replace target variables by functions, each acting on disjoint variables.
    \begin{defi}[Composed addressing functions]
    \label{defi:composedaddressing}
    Let $t \geq 2$, $\ell_1, \dots, \ell_t \geq 1$ be any integers. Let $g_i : \pmone^{\ell_i} \to \pmone$ be any functions for $i \in [t]$. Define the function $\AD_t \circt (g_1, \dots, g_t) : \pmone^{\log t} \times \pmone^{\ell_1 + \cdots + \ell_t} \to \pmone$ by
    \[
    \AD_t \circt (g_1, \dots, g_t) (x, y_1, \dots, y_t) = \AD_t(x, g_1(y_1), \dots, g_t(y_t)),
    \]
    where $x \in \pmone^{\log t}$ and $y_i \in \pmone^{\ell_i}$ for all $i \in [t]$.
    \end{defi}
    For any function $g : \pmone^s \to \pmone$, we use the notation $\AD_t \circt g$ to denote the function $\AD_t \circt (g, g, \dots, g) : \pmone^{\log t} \times \pmone^{ts} \to \pmone$.

    \subsection{Fourier-analytic measures of Boolean functions}
            We now introduce a few Fourier-analytic measures on Boolean functions that we use throughout the rest of the paper, and state some important relationships between them. Recall that we use the notation $\dim(S)$ to denote the dimension of the span of the set $S$. 
        \begin{defi}[Fourier rank]
        \label{defi:rank}
            Let $f: \pmone^n \to \pmone$ be any function. 
            Define the Fourier rank of $f$, denoted $r(f)$, by
            \[
                r(f) = \dim(\supp(f)).
            \]
        \end{defi}
    We often refer to Fourier rank as simply \emph{rank}.
    Sanyal~\cite{San19} showed the following upper bound on the rank of Boolean functions in terms of their sparsity.
    \begin{theorem}[{\cite[Theorem 1.2]{San19}}]
    \label{thm:Sanyal-original}
        Let $f:\pmone^n \to \pmone$ be any function. Then
        \[
            r(f) = O(\sqrt{k(f)}\log k(f)).
        \]
    \end{theorem}

We require the following observation which gives a simple upper bound on the rank of a Boolean function.
\begin{observation}
\label{obs: napdt_implies_rank_ub}
Let $f: \pmone^n \to \pmone$ be any function and $\Gamma$ be a set of parities. If for all $b \in \pmone^{\Gamma}$ the restricted function $f|_{{(\Gamma, b)}}$ is constant then $r(f) \leq |\Gamma|$.
\end{observation}
        Recall that for any function $f : \pmone^n \to \pmone$ and any real $t > 0$, we define $\calS_t := \{S \subseteq [n]: |\widehat{f}(S)| \geq 1/\threshold\}$ (we suppress the dependence of $\calS_t$ on $f$ as the underlying function will be clear from context).

        \begin{defi}
        \label{defi:max entropy, max rank entropy}
            Let $f:\fcube{n} \to \fcube{}$ be any function. Define the Fourier max-supp-entropy of $f$, denoted $k'(f)$, by
             \[
             k'(f) := \argmin_\threshold\cbra{\cS_{\threshold} = \supp(f)}.
            \]
            Equivalently,
            \[
             k'(f) := \max_{S \in \supp(f)}\cbra{\frac{1}{|\wh{f}(S)|}}.
            \]
            Define the Fourier max-rank-entropy of $f$, denoted $k''(f)$, by 
                \[
                k''(f): = \argmin_\threshold\cbra{\dim(\cS_{\threshold}) = r(f)}.
                \]
        \end{defi}
        We often refer to the Fourier max-supp-entropy and Fourier max-rank-entropy as simply \emph{max-supp-entropy} and \emph{max-rank-entropy}, respectively.
        \begin{lemma}[Relationships between parameters]
        \label{lem:relationships between rk and k'}
            Let $f:\fcube{n} \to \fcube{}$ be any function.
            Then the following inequalities hold.
        \begin{enumerate}
            \item $\log k(f) \leq r(f) = O(\sqrt{k(f)} \log{k(f)})$.
            \item
            $\sqrt{k(f)} \leq k'(f) \leq k(f)/2$.
            \item $\max\cbra{\sqrt{r(f)}, r(f)/(4\log k(f))} \leq k''(f) \leq k'(f)$.
        \end{enumerate}
        \end{lemma}
        \begin{proof}~
        \begin{enumerate}
            \item
            The first inequality holds since $k(f) \leq 2^{r(f)}$, and the second inequality follows from Theorem~\ref{thm:Sanyal-original}.
            \item 
            Recall from Definition~\ref{defi:max entropy, max rank entropy} that $k'(f) = \argmin_\threshold\cbra{\cS_{\threshold} = \supp(f)}$.
            This means for all $S \in \supp(f)$, $|\wh{f}(S)| \geq \frac{1}{k'(f)}$. 
            We have from Parseval's identity (Theorem~\ref{thm:Parseval}) that 
            \[
                \sum_{S \subseteq [n]} \wh{f}(S)^2 = 1
                \implies k(f)/k'(f)^2 \leq 1 \implies \sqrt{k(f)} \leq k'(f).
            \]
            By Lemma~\ref{lem:granularity}, 
            \[
            \frac{1}{k'(f)} \geq 2^{1 - \lfloor \log k(f) \rfloor}
            = \frac{2}{2^{\lfloor \log k(f) \rfloor}}
            \geq \frac{2}{k(f)}.
            \]
            \item      
            Recall from Definition~\ref{defi:max entropy, max rank entropy} that 
            $k''(f) = \argmin_\threshold\cbra{\dim(\cS_{\threshold}) = r(f)}$.
            Observe that for $t = k'(f)$, we have $\dim(\calS_t) = \dim(\supp(f)) =  r(f)$. Hence $k''(f) \leq k'(f)$. 
        
            Since $\rank(\calS_{k''(f)}) = r(f)$, there exists $\calB \subseteq \calS_{k''(f)}$
            such that $|\calB| = r(f)$ and $\calB$
            is a basis of $\spann(\supp(f))$.
            Moreover $|\wh{f}(S)| \geq 1/k''(f)$ for all $S \in \calB$.  
            Choose such a set $\calB$.
            
                By Theorem~\ref{thm:Parseval},
                \begin{align*}
                    & 1 \geq \sum_{S \in \calB} \wh{f}(S)^2 \geq \frac{r(f)}{(k''(f))^2} \\
                    & \implies k''(f) \geq \sqrt{r(f)}.
                \end{align*}
            
            By Corollary \ref{cor:chlt-implication},
            \begin{align*}
                & \sum_{S \in \calB}|\wh{f}(S)| \leq 4\log k(f)\\
                & \implies \frac{r(f)}{k''(f)} \leq 4\log k(f) \\
                & \implies k''(f) \geq r(f)/(4\log k(f)).
            \end{align*}
            Therefore $k''(f) \geq \max\cbra{\sqrt{r(f)}, r(f)/(4\log k(f))}$.
        \end{enumerate}
        \end{proof}

\begin{claim}
\label{claim:delta at least k/k'^2}
    Let $f: \pmone^n \to \pmone$ a function with $k(f) \geq 2$. Then
        \[
        \delta(f) = \Omega\bra{\frac{k(f)}{k'(f)^2}}.
        \]
\end{claim}
\begin{proof}
                Recall from Definition~\ref{defi:max entropy, max rank entropy} that $k'(f) = \argmin_\threshold\cbra{\cS_{\threshold} = \supp(f)}$.
                This means for all $S \in \supp(f)$, $|\wh{f}(S)| \geq \frac{1}{k'(f)}$. Therefore
                \begin{align*}
                    (k(f)-1)/(k'(f))^2 & \leq \sum_{S \subseteq [n], S \neq \emptyset} \wh{f}(S)^2\\
                    &= 1 - \wh{f}(\emptyset)^2 \tag*{by Theorem~\ref{thm:Parseval}}\\
                    &= 1 - (1-2\delta(f))^2 \tag*{by Observation~\ref{obs:weight, empty Fourier}}\\
                    &= 4\delta(f) - 4\delta(f)^2
                    \leq 4\delta(f)\\
                    \implies \delta(f) &\geq \frac{(k(f) - 1)}{4(k'(f))^2} \geq \frac{k(f)}{8(k'(f))^2}. \tag*{since $k(f) \geq 2$}
                \end{align*}
\end{proof}

\begin{claim}
\label{claim:lone_le_sqrt_kdelta}
Let $f:\pmone^n \to \pmone$ be any function. Then
    \[
    \lone{\wh{f}} \leq 3\sqrt{k(f) \delta(f)}.
    \]
\end{claim}
\begin{proof}
    Since $\widehat{f}(\emptyset) = 1 - 2\delta(f)$ we have
    $\lone{\wh{f}} = |1- 2 \delta(f)| + \sum_{S \neq \emptyset} |\wh{f}(S)|$. The term $\sum_{S \neq \emptyset} |\wh{f}(S)|$ can be bounded as follows:
    \begin{align*}
        \bra{\sum_{S \neq \emptyset } |\wh{f}(S)| }^2 \tag*{by Cauchy-Schwarz inequality}
            &\leq k(f) \bra{\sum_{S \neq \emptyset} \wh{f}(S)^2}\\
            &= k(f) (1 - (1-2\delta(f))^2) 
            \leq 4 k(f)\delta(f).
    \end{align*}
    Thus we have,
    \begin{align*}
        \lone{\wh{f}} &= |1-2\delta(f)| + \sum_{S \neq \emptyset} |\wh{f}(S)|\\
            &\leq |1-2\delta(f)| + 2\sqrt{k(f) \delta(f)} \\
            &\leq 1 + 2\sqrt{k(f) \delta(f)} \leq 3 \sqrt{k(f)\delta(f)}. \tag*{since $k(f)\delta(f) \geq 1$ by Lemma~\ref{lem:uncertainity principle}}
    \end{align*}
\end{proof}

    We require the following observation about the rank, sparsity, max-supp-entropy, max-rank-entropy and weight of $\AND, \AD_{t}$ and bent functions, which follows immediately from definitions and first principles. We omit its proof.

    \begin{observation}
    \label{obs:properties of AND, Bent and Addressing}
    Let $t \geq 2$ and $\ell, t' \geq 4$ be any positive integers, and let $B_\ell: \pmone^{\log \ell} \to \pmone$ be any bent function. Then the rank, sparsity, max-supp-entropy, max-rank-entropy and weight of $\AND_{\log t'}$, $B_\ell$ and $\AD_t$ are as in the following table.
        \begin{center}
    \begin{tabular}{ |c|c|c|c|c|c| } 
         \hline
         $f$ & $r(f)$ & $k(f)$ & $k'(f)$ & $k''(f)$ & $\delta(f)$\\ 
        \hline
        \hline
         $\AND_{\log t'}$ & $\log t'$ & $t'$ & $t'/2$ & $t'/2$ & $\frac{1}{t'}$\\ 
        \hline 
         $B_\ell$ & $\log \ell$ & $\ell$ & $\sqrt{\ell}$ & $\sqrt{\ell}$ & $\frac{1}{2} \pm \frac{1}{2 \sqrt{\ell}}$ \\
         \hline
         $\AD_{t}$ & $t + \log t$ & $t^2$ & $t$ & $t$ & $\frac{1}{2}$\\
         \hline
    \end{tabular}
\end{center}
    \end{observation}


\section{Lower bound proofs}
\label{sec:lower bound}

For lower bounds on $\delta(f)$ of a Boolean function $f$, we need to prove two theorems: Theorems~\ref{thm:delta lower bound in terms of rk and also k'} and~\ref{thm:delta lower bound in terms of rk and also k''}. The proof of Theorem~\ref{thm:delta lower bound in terms of rk and also k'} is given in Section~\ref{sec:lower bound using k'} and the proof of Theorem~\ref{thm:delta lower bound in terms of rk and also k''} is given in Section~\ref{sec:lower bound using k''}.

\subsection{Proof of Theorem~\ref{thm:delta lower bound in terms of rk and also k'} (and Theorem~\ref{thm:delta lower bound in terms of rk only})}
\label{sec:lower bound using k'}

Remember that we defined the \emph{Fourier max-supp-entropy} of a Boolean function $f$, denoted by $\seespectrum(f)$, to be 
$$
\max_{S \in \supp(f)} \frac{1}{|\wh{f}(S)|}.
$$ 
The main aim of this section is to give a lower bound on $\delta(f)$ with respect to $\seespectrum(f)$ for a Boolean function $f$ (Theorem~\ref{thm:delta lower bound in terms of rk and also k'}).

We first prove Theorem~\ref{thm:delta lower bound in terms of rk only} which implies Theorem~\ref{thm:delta lower bound in terms of rk and also k'} (together with Claim~\ref{claim:delta at least k/k'^2}). See Section~\ref{sec:overview of delta lower bound in terms of rk only} for an overview of the proof of Theorem~\ref{thm:delta lower bound in terms of rk only}.

Theorem~\ref{thm:delta lower bound in terms of rk only} can be viewed as an upper bound of $O(\sqrt{k(f)\delta(f)}\log k(f))$ on the Fourier rank of $f$.
In order to prove Theorem~\ref{thm:delta lower bound in terms of rk only}, we give an algorithm (Algorithm~\ref{alg:NAPDT}) which takes a Boolean function $f$ as input and outputs a set of $O(\sqrt{\delta(f) k(f)}\log k(f))$ parities such that any assignment of these parities makes the function constant. From Observation~\ref{obs: napdt_implies_rank_ub}, this implies an upper bound of $O(\sqrt{\delta(f) k(f)}\log k(f))$ on Fourier rank of the function. We start by formally describing this algorithm.

 Recall that for a function $f:\pmone^n \rightarrow \pmone$, a set of parities $\Gamma$ and an assignment $b \in \pmone^\Gamma$, we define the restriction
    \[ 
        f|_{\Vb} := f|_{\{x \in \pmone^n: \chi_{\gamma}(x) = b_{\gamma} \textnormal{~for all~} \gamma \in \Gamma\}}. 
    \]
Also let $\mathcal{B}_{\Gamma} := \{b \in \pmone^{\Gamma}: f|_{\Vb} \textit{ is not constant} \}$.

\begin{algorithm}[H]
\label{alg:NAPDT}
\SetAlgoLined
{\bf Input:} A function $f: \pmone^n \to \pmone$.\\

{\bf Output: }
A set $\Gamma$ of parities whose evaluation determines $f$.\\
{\bf Initialization:} $\fmin \leftarrow f$, $\Gamma \leftarrow \emptyset$.\ \\
 \While{$\mathcal{B}_{\Gamma}$ is non-empty}{
  \begin{enumerate}[label = (\alph*)]
        \item \label{item: step_a} 
        \textbf{Update $\Gamma$:}
        Let $\Gamma'$ be the smallest set of parities, such that, there exists $b\in \{-1,1\}^{\Gamma'}$ for which $\fmin|_{{(\Gamma',b)}}$ is constant,
        \[ \Gamma \leftarrow \Gamma \cup \Gamma'. \]
        
        \item \label{item: step_b}
        \textbf{Update $\fmin$:}
        Define $b^* := \arg\!\min_{b \in \mathcal{B}_{\Gamma}} 
            \left\{ \frac{\delta(f|_{\Vb})}{k(f|_{\Vb})} \right\}$, and update
           \[ \fmin \leftarrow f|_{({\Gamma,b^*})}. \]
    \end{enumerate}
 }
 Return $\Gamma$.
\caption{}
\end{algorithm}

Since number of parities are finite and we fix at least one parity at each iteration of Step~\ref{item: step_a} of the \textnormal{\textbf{while}} loop, the algorithm terminates. The termination condition implies that the algorithm outputs a set of parities $\Gamma$ such that for any assignment $b \in \pmone^{\Gamma}$ of $\Gamma$, the restricted function $f_{(\Gamma,b)}$ becomes constant. 

The only remaining step is to show that the number of parities fixed in Algorithm~\ref{alg:NAPDT} is $O(\sqrt{\delta(f) k(f)}\log k(f))$. For this we first need to recall the notion of equivalence relation defined in Section~\ref{sec:overview of delta lower bound in terms of rk only} and few properties of restricted functions (restricted according to an assignment of a set of parities).

\paragraph{Equivalence relation for a set of parities} Let $f$ be the input to Algorithm~\ref{alg:NAPDT}, first we define an equivalence relation given a set of parities over the variables of $f$. Given a set of parities $\Gamma$, define the following equivalence relation among parities in $\supp(f)$.
\begin{equation}
\label{eq: equivalence relation for Gamma}
    \forall \gamma_1, \gamma_2 \in 
    \supp(f),
    \gamma_1 \equiv \gamma_2\ \text{iff}\ \gamma_1+\gamma_2 \in \spann(\Gamma).
\end{equation}
    
Let $\ell$ be the number of equivalence classes according to the equivalence relation for $\Gamma$. For $j \in [\ell]$, let $k_j$ be the size of the $j$-th equivalence class. Since the equivalence classes form a partition of 
$\supp(f)$, we have
\begin{observation}
\label{obs:coset_sparsity_sum} 
    Following the notation of the paragraph above, $\sum_{j=1}^{\ell} k_j = \sparsity(f)$.
\end{observation}

Let $\beta_1, \dots , \beta_{\ell} \in \supp(f)$
be some representatives of the equivalence classes. For $j\in [\ell]$, let $\beta_j + \alpha_{j,1}, \dots ,\beta_j + \alpha_{j,k_j}$ be the elements of the $j$-th equivalence class. This notation gives a compact representation of $f$ in terms of these equivalence classes. For all $x \in \pmone^n$,
    \begin{align}
        f(x) = \sum_{j =1}^{\ell} P_j (x) \chi_{\beta_j}(x) \label{eq:coset1},
    \end{align}
    where
    \begin{align}
        P_j(x) = \sum_{r=1}^{k_j} \wh{f}(\beta_j + \alpha_{j,r}) \cdot \chi_{\alpha_{j,r}}(x). \label{eq:coset2}
    \end{align}
Note that $P_j$ are non-zero multilinear polynomials and depend only on the parities in $\Gamma$. So, fixing parities in $\Gamma$ collapses all the parities in an equivalence class to their representative, thereby making $P_j$'s constant. 

We will denote $\Gamma$ after the $i$-th iteration of the \textnormal{\textbf{while}} loop by $\Gamma^{(i)}$ (so $\Gamma^{(0)} = \emptyset$). Let $\fmin^{(i)}$ be the selected function $\fmin$ after the $i$-th iteration (thus $\fmin^{(0)} = f$). 

With the above properties of restricted functions we are ready to prove the main technical lemma needed to show Theorem~\ref{thm:delta lower bound in terms of rk only}.

\begin{lemma}
\label{lem:main_lemma}
    Let $f: \pmone^n \to \pmone$ a function.
    Suppose $\Gamma$ be a set of parities and $\ell$ be the number of equivalence classes of 
    $\supp(f)$ under the equivalence relation defined by  in Equation~\eqref{eq: equivalence relation for Gamma},
    Then, there exists a $b \in \pmone^{\Gamma}$ such that $f|_{{(\Gamma, b)}}$ is non-constant and 
    \[
        \frac{\delta(f|_{{(\Gamma, b)}})}{
        k(f|_{{(\Gamma, b)}})
        } \leq \frac{4 k(f) \delta(f)}{\ell^2}. 
    \]
\end{lemma}
\begin{proof}
    For the sake of succinctness, when $\Gamma$ is clear from the context, let  $V_b = \{x \in \pmone^n: \forall \gamma \in \Gamma, x_{\gamma} = b_{\gamma}\}$, for all $b \in \pmone^{\Gamma}$, and $f|_b = f|_{\{x : x \in V_b\}}$.
    
    Since we are interested in a non-constant $f|_b$, define 
    $\sparp(f)$ 
    to be the number of non-zero non-empty monomials in Fourier representation of $f$. We first need to prove the following two bounds on the expected values of $\delta(f|_{b})$ and 
    $\sparp(f|_{b})$.
    \begin{itemize}
        \item $\Ex{b}{\delta(f|_{b})}= \delta(f)$,
        \item $\Ex{b}{\sparp(f|_{b})} \geq \frac{\ell^2}{4 k(f)}$.
    \end{itemize}

    \paragraph{Expected value of $\delta(f|_{b})$:} Since $\cbra{V_b: b \in \pmone^{\Gamma^{(i)}}}$
    form a partition on $\pmone^{n}$ and all partitions are of the same size, we get the expected value of $\delta(f|_{b})$.
    \begin{equation}
    \label{weight expectation}
        \Ex{b}{\delta(f|_{b})}= \delta(f).
    \end{equation}
    
    \paragraph{Expected value of 
    $\sparp(f|_{b})$:} From Equation~\eqref{eq:coset1},
    for all $b \in \pmone^{\Gamma}$ and for all $x \in \pmone^n$,
    \begin{align}
        f|_{b}(x) = \sum_{j=1}^{\ell} P_{j}(b) \chi_{\beta_j}(x). \label{eq: rep_f_sum_of_cosets}
    \end{align}
    For each $j \in [\ell]$ and $b \in \pmone^{\Gamma}$,
    let $I_{j}(b)$ be the indicator function for $P_j(b) \neq 0$,
        $$
            I_{j}(b) = \begin{cases}
                        1 & \text{if}\  P_j(b) \neq 0\\
                        0 & \text{otherwise}.
                      \end{cases}
        $$
    From Equation~\eqref{eq:coset2}, each $P_j$ is a polynomial having monomials $\{\chi_{\alpha_{j,r}}: r \in [k_j] \}$ 
    with 
    Fourier sparsity of $P_j$ being equal to $k_j$.
    Since each $P_j$ is a non-zero polynomial, by Lemma~\ref{lem:uncertainity principle}
    \begin{align}
        \Ex{b}{I_{j}(b)} = \prob{b \sim \pmone^{\Gamma}}{P_j(b) \neq 0} \geq \frac{1}{k_j}. \label{eq:coset_poly_nonzer_largefrac}
    \end{align}
    
    We calculate the expectation of
    $\sparp(f|_b)$.
    \begin{align}    \label{spar expectation}
        \Ex{b}{\sparp(f|_{b})}
            &= \Ex{b}{\sum_{j=1}^{\ell-1} I_{j}(b)} \tag*{by Equation~(\eqref{eq: rep_f_sum_of_cosets})} \nonumber \\
            &= \sum_{j=1}^{\ell-1} \Ex{b}{I_{j}(b)} \tag*{by linearity of expectation} \nonumber \\
            &\geq \sum_{j=1}^{\ell-1} \frac{1}{k_j} \tag*{by Equation~\eqref{eq:coset_poly_nonzer_largefrac}} \nonumber \\
            &\geq \frac{(\ell-1)^2}{\sum_{j=1}^{\ell-1} k_j} \nonumber \tag*{by Cauchy-Schwarz inequality}\\
            &\geq \frac{\ell^2}{4k(f)}. \tag*{by Observation~\ref{obs:coset_sparsity_sum}}
    \end{align}
     To finish the proof of the theorem, we use bounds on the two expected values,\footnote{this part of our proof is inspired by a proof of the Cheeger's inequality in spectral graph theory. See, for example, the proof of Fact 2 in \url{https://people.eecs.berkeley.edu/~luca/expanders2016/lecture04.pdf}.}
    \begin{align*}
        &\frac{\Ex{b}{\delta(f|_{b})}}{
        \Ex{b}{\sparp(f|_{b})}}
        \leq \frac{4 k(f)\delta(f)}{\ell^2} \\
        \iff &\Ex{b}{\delta(f|_{V_b}) - \frac{4\sparsity(f)\delta(f)}{\ell^2} \sparp(f|_{V_b})}
        \leq 0. \tag*{by linearity of expectation}\\
    \end{align*}
    If $\delta(f|_{V_b}) - \frac{4\sparsity(f)\delta(f)}{\ell^2} 
    \sparp(f|_{V_b}) = 0$ for all $b$, then pick any non-constant $f|_b$. Otherwise, there exists a $b_0$ such that 
    \[
        \delta(f|_{V_{b_0}}) - \frac{4\sparsity(f)\delta(f)}{\ell^2} 
        \sparp(f|_{V_{b_0}}) <
        0 .
    \]
    Since this equation can only be satisfied when $\sparp(f|_{V_{b_0}}) > 0$, $f|_{V_{b_0}}$ is not constant. Dividing by $\sparp(f|_{V_{b_0}})$,
    $$
    \frac{\delta(f|_{b_0})}{k(f|_{b_0})} \leq \frac{\delta(f|_{b_0})}{\sparp(f|_{b_0})} \leq \frac{4 k(f) \delta(f)}{\ell^2},
    $$
    and $f|_{b_0}$ is non-constant.
\end{proof} 

Lemma~\ref{lem:main_lemma} allows us to bound the number of parities fixed in the $i$-th iteration (in terms of the decrease in number of equivalence classes).
    \begin{lemma}
    \label{lem:q_i bound}
    Suppose $f$ is given as input to Algorithm~\ref{alg:NAPDT}.
    Consider the $i$-th iteration of Algorithm~\ref{alg:NAPDT}. Let $q_i$ be the be number of parities fixed in Step~\ref{item: step_a} of the $i$-th iteration of the \textnormal{\textbf{while}} loop, and $\ell_i$ be the number of equivalence classes after Step~\ref{item: step_a} of the $i$-th iteration. Then
            $$
            \frac{q_i}{(\ell_{i-1} - \ell_{i})} \leq \frac{6\sqrt{\delta(f) \sparsity(f)}}{\ell_{i-1}}.
            $$
    \end{lemma}
    \begin{proof}
        Recall that $\Gamma = \Gamma^{(i)}$ after the $i$-th of Step~\ref{item: step_a} of Algorithm~\ref{alg:NAPDT}.
        Again, for the sake of succinctness, let  $V_b = \{x \in \pmone^n: \forall \gamma \in \Gamma^{(i)}, x_{\gamma} = b_{\gamma}\}$, for all $b \in \pmone^{\Gamma^{(i)}}$, and $f|_b = f|_{\{x : x \in V_b\}}$.
        
        Let $\fmin$ be the function chosen after the $i$-th iteration of Step~\ref{item: step_b} of Algorithm~\ref{alg:NAPDT}. 
        Since Step~\ref{item: step_b} of Algorithm~\ref{alg:NAPDT} chooses $\fmin$ to be a non-constant function 
         such that weight-to-sparsity ratio is minimized, from Lemma~\ref{lem:main_lemma} we have,
        \begin{align}
          \frac{\delta(\fmin)}{k(\fmin)} 
            &\leq \frac{4 k(f) \delta(f)}{\ell_{i-1}^2}. \label{eq: fmin_i_satisfies_main_lemma}
        \end{align}   
        
        Write every $f|_b$ as in Equation~\eqref{eq:coset1}, and define $\mathcal{S}^{(i)} := \bigcup_{b \in \pmone^{\Gamma^{(i)}}} 
        \supp(f|_{b})$. We now prove that $|\mathcal{S}^{(i)}| = \ell_{i}$.
        \begin{itemize}
            \item $|\mathcal{S}^{(i)}| \leq \ell_{i}$:
            Follows from the representation in Equation~\eqref{eq:coset1}, since each 
            $\supp(f|_b)$ is a subset of $\{\chi_{\beta^{(i)}_j} \mid j \in [\ell_i]\}$.
            \item $|\mathcal{S}^{(i)}| \geq \ell_{i}$: Since $P^{(i)}_j$ is a non-zero polynomial, there exists an assignment to parities in $\Gamma^{(i)}$,
        such that, $P^{(i)}_j$ is non-zero. Thus, for all $j \in [\ell_i]$, we have $\chi_{\beta_j^{(i)}} \in \mathcal{S}^{(i)}$.
        \end{itemize}

        Since $|\mathcal{S}^{(i)}| = \ell_{i}$,
        Lemma~\ref{lem:TWXZ13} guarantees that $q_i \leq 3 \sqrt{k(\fmin) \delta(\fmin)}$. Since $\fmin$ becomes constant after fixing these $q_i$ parities, every parity in 
        $\supp(\fmin)$ is paired with at least one other parity in 
        $\supp(\fmin)$ for the equivalence class with respect to $\Gamma^{(i)}$.\footnote{There is a boundary case ($k(f) = 1$) which can be dealt with separately, as in~\cite[Lemma 3.4]{San19}. For readability, we assume $k(f) \geq 2$.} This implies that $\ell_{i-1} - \ell_{i} \geq \frac{k(\fmin)}{2}$
        Combining the two inequalities in the last paragraph we have,
        \begin{align*}
            \frac{q_i}{(\ell_{i-1} - \ell_{i})}
            &\leq 6\sqrt{\frac{\delta(\fmin)}{k(\fmin)}}. \label{eq: tsang_plus_delta_coset_size}
        \end{align*}
        From Equation~\eqref{eq: fmin_i_satisfies_main_lemma},
        \begin{align}
          \frac{q_i}{(\ell_{i-1} - \ell_{i})} 
            &\leq \frac{6\sqrt{\delta(f) \sparsity(f)}}{\ell_{i-1}}.
        \end{align}
    \end{proof}

We are now ready to prove Theorem~\ref{thm:delta lower bound in terms of rk only}.
\begin{proof}[Proof of Theorem \ref{thm:delta lower bound in terms of rk only}]
    We only need to show that the parities fixed in Algorithm~\ref{alg:NAPDT} is $O(\sqrt{\delta(f) k(f)}\log k(f))$ (Observation~\ref{obs: napdt_implies_rank_ub}).
    Suppose the while loop runs for $t$ iterations. 
    Let $q_i$ be the number of queries made in Step~\ref{item: step_a}
    of Algorithm~\ref{alg:NAPDT}
    in the $i$-th iteration.
    From Lemma~\ref{lem:main_lemma}, we have
    \begin{align*}
      q_i 
            \leq \frac{6\sqrt{\delta(f) \sparsity(f)}}{\ell_{i-1}}(\ell_{i-1} - \ell_{i})
    \end{align*}
    
    Thus when Algorithm~\ref{alg:NAPDT} is run of $f$, the total number of queries made by the algorithm is
    \begin{align*}
    \sum_{i=1}^t q_i
            &\leq 6 \sqrt{\delta(f) \sparsity(f)}\sum_{i=1}^t \frac{(\ell_{i-1} - \ell_{i})}{\ell_{i-1}}\\ 
            &\leq 6\sqrt{\delta(f) \sparsity(f)}\sum_{i=1}^t \bra{\frac{1}{\ell_{i-1}} + \frac{1}{\ell_{i-1} - 1} \ldots + \frac{1}{ \ell_{i} + 1}} \\
            &\leq 6 \sqrt{\delta(f) \sparsity(f)}\sum_{i=1}^{\ell_0} \frac{1}{i} \\
            &\leq 6 \sqrt{\delta(f) \sparsity(f)}\log{\ell_0}\\ 
            &= 6\sqrt{\delta(f) \sparsity(f)} \log {\sparsity(f)}.    
    \end{align*}
    Observation~\ref{obs: napdt_implies_rank_ub} implies $r(f) = O(\sqrt{\delta(f) k(f)} \log k(f))$.
\end{proof}

Along with Theorem~\ref{thm:delta lower bound in terms of rk only}, this proves Theorem~\ref{thm:delta lower bound in terms of rk and also k'}.

    \begin{proof}[Proof of Theorem~\ref{thm:delta lower bound in terms of rk and also k'}] The bound $\delta(f) = \Omega\left(\frac{1}{k(f)}\bra{\frac{r(f)}{\log k(f)}}^2\right)$ follows from Theorem~\ref{thm:delta lower bound in terms of rk only} and the bound $\delta(f) = \Omega\left(\frac{\sparsity(f)}{(\seespectrum(f))^2}\right)$ from Claim~\ref{claim:delta at least k/k'^2}.
    \end{proof}
    

\subsection{Proof of Theorem~\ref{thm:delta lower bound in terms of rk and also k''}}
\label{sec:lower bound using k''}

Recall from Definition~\ref{defi:max entropy, max rank entropy} that we defined \emph{max-rank-entropy} of a Boolean function $f$, denoted by $\seerank(f)$, to be 
$$
\argmin_\threshold\{\dim(\cS_{\threshold})\} = r(f).
$$ 
The main aim of this section is to give a lower bound on $\delta(f)$ with respect to $\seerank(f)$ for a Boolean function $f$ (Theorem~\ref{thm:delta lower bound in terms of rk and also k''}).
The second bound of Theorem~\ref{thm:delta lower bound in terms of rk and also k''} is given by the following lemma.

\begin{lemma}
\label{lem:delta lower bound in terms of rk and also k''}
Let $f: \pmone^n \to \pmone$ be any function such that $k(f) > 1$. Then,
\[
\delta(f) = \Omega\left( \frac{r(f)}{\seerank(f) \log \sparsity(f)}\right).
\]
\end{lemma}

Together with Theorem~\ref{thm:delta lower bound in terms of rk only} proved in Section~\ref{sec:lower bound using k'}, Lemma~\ref{lem:delta lower bound in terms of rk and also k''} implies Theorem~\ref{thm:delta lower bound in terms of rk and also k''}. We now give the proof of Lemma~\ref{lem:delta lower bound in terms of rk and also k''}. See Section~\ref{sec:overview of delta lower bound in terms of rk and also k''} for an overview of the proof of Lemma~\ref{lem:delta lower bound in terms of rk and also k''}.

Lemma~\ref{lem:delta lower bound in terms of rk and also k''} gives a lower bound of $\Omega\left( \frac{r(f)}{\seerank(f) \log \sparsity(f)}\right)$ on $\delta(f)$. The crucial ingredient for this lower bound is Lemma~\ref{lem:CHLT_improvement}, which is a refinement of the following theorem.
\begin{theorem}[{\cite[Theorem 13]{CHLT19}}]
    \label{thm:CHLT}
        Let $f:\pmone^n \to \pmone$ be any function such that $\deg_{\ftwo}(f) = d$. Then,
    \[
        \sum_{i \in [n]} |\wh{f}(\cbra{i})| \leq 4d.
    \]
    \end{theorem}

The only difference in the statement of Lemma~\ref{lem:CHLT_improvement} and Theorem~\ref{thm:CHLT} is that the right hand side becomes $O(\delta(f) \cdot \deg_{\ftwo}(f))$ instead of $4\deg_{\ftwo}(f)$.

\begin{proof}[Proof of Lemma~\ref{lem:CHLT_improvement}]   
    Assume $\delta(f) \leq 1/4$ (otherwise Theorem~\ref{thm:CHLT} implies $\sum_{i=1}^n |\widehat{f}(i)| = O(\delta(f) d)$).
    
    Define $F : \pmone^{nt} \to \pmone$ to be  \[ F(x^{(1)}, \ldots , x^{(t)}) = f(x^{(1)}) \times \ldots \times f(x^{(t)}),  \] 
    where $t$ is a parameter to be fixed later, and $x^{(i)} \in \pmone^n$  for all $i \in [t]$. Since $\degtwo(F) = \degtwo(f)$, Theorem~\ref{thm:CHLT} implies 
    \begin{equation}
    \label{eq:CHLT-First-L1}
    \sum_{\substack{S \subseteq [nt]\\ |S| = 1}} |\widehat{F}(S)| = O(d). 
    \end{equation}
    Since $(1-x)^{1/x}$ is a decreasing function in $x$ for $x \in (0, 1/2]$, we have
    \begin{equation}
        \label{eq:1-xex decreasing}
        (1-x)^{1/x} \geq 1/4 \quad \text{for all}~x \in (0, 1/2].
    \end{equation}
    Expressing the Fourier coefficients of $F$ in terms of the Fourier coefficients of $f$,
    \begin{align*}
        \sum_{\substack{S \subseteq [nt]\\ |S| = 1}} |\widehat{F}(S)| &= t \cdot \widehat{f}(\emptyset)^{t-1} \sum_{i=1}^n |\widehat{f}(i)| \\
        &= \left(1+\frac{1}{2\delta(f)}\right) \cdot (1- 2\delta(f))^{\frac{1}{2\delta(f)}} \sum_{i=1}^n |\widehat{f}(i)|  \tag*{Choosing $t = 1+ \frac{1}{2\delta(f)}$, and by Observation~\ref{obs:weight, empty Fourier}}\\
        &\geq \left(1+\frac{1}{2\delta(f)}\right) \cdot \bra{\frac{1}{4}} \sum_{i=1}^n |\widehat{f}(i)| \tag*{by Equation~\eqref{eq:1-xex decreasing}}\\
        &\geq \frac{1}{8\delta(f)} \cdot \sum_{i=1}^n |\widehat{f}(i)|. 
    \end{align*}
      Now, Equation~\eqref{eq:CHLT-First-L1} implies the desired bound, $\sum_{i=1}^n |\widehat{f}(i)| = O(\delta(f) d)$.
\end{proof}

We would like to extend the upper bound of Lemma~\ref{lem:CHLT_improvement} to any basis of $\spann(\supp(f))$ instead of just the standard basis of the set of parities.

\begin{corollary}\label{cor:improved chlt-implication}
        Let $f:\pmone^n \to \pmone$ be any function with $\degtwo(f) = d$. Suppose $\cS \subseteq \supp(f)$ is a basis of $\spann(\supp(f))$, then
    \[
        \sum_{S \in \cS} |\wh{f}(S)| = O(\delta(f) d) = O(\delta(f) \log k(f)).
    \]
\end{corollary}
\begin{proof}
The main idea of the proof is to do a basis change on parities and construct another function $h$, the corollary will follow by applying Lemma~\ref{lem:CHLT_improvement} on $h$.

Recall that we denote both a subset of $[n]$ and the corresponding indicator vector in $\ftwo^n$, by the same notation.

Let $\mathcal{S} = \{S_1, \dots, S_{r(f)}\}$, extend $\mathcal{S}$ to $\mathcal{S'} = \{S_1, \dots, S_{r(f)}, S_{r(f)+1}, \dots, S_n\}$, a complete basis of $\ftwo^n$. Observe that $\wh{f}(S_i) = 0$, for $i \in \{r(f)+1, \dots, n\}$ (since $\mathcal{S}$ spans $\supp(f)$). Fix the change of basis matrix $B \in \ftwo^{n \times n}$ with $i$-th column as $S_i$, $i \in [n]$.

    Consider the function $h:\pmone^n \to \R$ satisfying $\widehat{h}(\alpha) = \widehat{f}(B\alpha)$, for all $\alpha \in \ftwo^n$.  
    By Claim~\ref{thm:ftwo_deg_does_not_change}, $h$ is Boolean and $\degtwo(h) = \degtwo(f)$.
    Using Lemma~\ref{lem:CHLT_improvement},
    \begin{equation*}
        \sum_{i \in [n]} |\wh{h}(\{i\})| = O(\delta(f) d) .
    \end{equation*}
    From the definition of $h$, $\widehat{h}(e_i) = \widehat{f}(S_i)$ for $i \in [r(f)]$ and $\widehat{h}(e_i) = 0$ for $i \in \{r(f)+1, \dots, n\}$.
    \begin{equation*}
        \sum_{S \in \mathcal{S}} |\wh{f}(S)| = O(\delta(f) d) .
    \end{equation*}    
    The second equality in the statement of the lemma follows from Lemma~\ref{lem:ftwodeg and Fourier sparsity}.
\end{proof}

\begin{proof}[Proof of Lemma~\ref{lem:delta lower bound in terms of rk and also k''}] Observe that every summand on the left hand side of Corollary~\ref{cor:improved chlt-implication} is at least $1/\seerank(f)$, giving the following lower bound on $\delta(f)$ and finishing the proof of Lemma~\ref{lem:delta lower bound in terms of rk and also k''}.
\end{proof}

\begin{proof}[Proof of Theorem~\ref{thm:delta lower bound in terms of rk and also k''}]
From Lemma~\ref{lem:delta lower bound in terms of rk and also k''} we have $\delta(f) =\Omega\left(\frac{r(f)}{\seerank(f) \log k(f)}\right)$, and from Theorem~\ref{thm:delta lower bound in terms of rk only} we have $\delta(f) = \Omega\left( \frac{r(f)^2}{\sparsity(f) \log^2 \sparsity(f)} \right)$.
\end{proof}

The following corollary combines the lower bounds on $\delta(f)$ from Theorem~\ref{thm:delta lower bound in terms of rk and also k''} and Lemma~\ref{lem:chang} by setting $k''(f)$ as the threshold.

\begin{corollary}
\label{cor:all delta lower bound in terms of rk and also k''}
Let $f: \pmone^n \to \pmone$ be any function such that $k(f) > 1$.
Then,
\[
\delta(f) = \Omega\left(\max\left\{\frac{r(f)^2}{\sparsity(f) \log^2 \sparsity(f)}, \frac{r(f)}{\seerank(f) \log \sparsity(f)}, \frac{\sqrt{r(f)}}{\seerank(f) \log(\seerank(f)^2/r(f)}\right\}\right).
\]
\end{corollary}


\section{Upper bound proofs}

In this section we prove Theorems~\ref{thm:delta upper bound in terms of rk and also k'} and~\ref{thm:delta upper bound in terms of rk and also k''}. Recall that these theorems require us to exhibit functions $f$ witnessing certain upper bounds on $\delta(f)$. The descriptions of these functions are given in Section~\ref{subsec:defining some functions}. In Section~\ref{subsec:proofofclaimfnproperties} we compute certain properties of interest of these functions. Finally in Section~\ref{subsec:tightness} we instantiate these functions with suitable parameters to yield the proofs of Theorems~\ref{thm:delta upper bound in terms of rk and also k'} and~\ref{thm:delta upper bound in terms of rk and also k''}. 

\subsection{Defining some functions}
\label{subsec:defining some functions}
    The functions we consider are all modifications of the Addressing function defined in Definition~\ref{defi:Addressing}. 
    The main technique we use to define our functions is given in Definition~\ref{defi:composedaddressing}. 
    That is, we first consider an Addressing function on $t + \log t$ input bits. Next, we replace each target bit by a suitable function. Different choices of the various functions substituted yield our upper bounds.

    In our first modification, we replace the target bits by $\AND$ functions on disjoint variables, each having the same arity.

    \begin{defi}[AND-Target-Addressing Function]
    \label{defi:ADtt'}
    For any integers $t, t' \geq 2$, define the function $\AD_{t, t'} : \pmone^{\log t} \times \pmone^{t \log t'} \to \pmone$ by
    \[
    \AD_{t, t'} = \AD_t \circt \AND_{\log t'}.
    \]
    \end{defi}
    
    Our next modification is similar to the previous one, except for the fact that one of the $\AND$ functions used in the replacement above has larger arity than the others.

    \begin{defi}[AND-Target-Addressing Function with a Huge AND]
    \label{defi:ADtt'a}
    For any integers $t \geq 2$ and $a \geq t' \geq 2$, define the function $\AD_{t, t', a} : \pmone^{\log t} \times \pmone^{\log a} \times \pmone^{(t-1) \log t'} \to \pmone$ by
    \[
    \AD_{t, t', a} = \AD_t \circt (\AND_{\log a}, \AND_{\log t'}, \dots, \AND_{\log t'}).
    \]
    That is, for all $(x, z) \in \pmone^{\log t} \times \pmone^{\log a} \times \pmone^{(t-1) \log t'}$, where $z_{1^{\log t}} \in \pmone^{\log a}$ and $z_{b} \in \pmone^{\log t'}$ for all $1^{\log t} \neq b \in \pmone^{\log t}$, 
    \begin{align*}
        \AD_{t, t', a}(x,z) =
        \begin{cases}
            \AND(z_{1^{\log t}, 1}, \dots, z_{1^{\log t}, \log a}) & \textnormal{ if } x = 1^{\log t} \\
            \AND(z_{x,1}, \dots, z_{x,\log t'}) & \textnormal{ otherwise.}
        \end{cases}
    \end{align*}
    \end{defi}

    We require the following function to define our next modification.
        
    \begin{defi}[AND-of-Bent]
    \label{defi:AB}
        For any integers $t',\ell \geq 2$, let $B: \pmone^{\log \ell} \to \pmone$ be a bent function on $\log \ell$ input bits. Define the function $\AB : \pmone^{\log t'} \times \pmone^{\log \ell} \to \pmone$ by 
            $$\AB(y,z) = \AND(y_1B(z),y_2, \ldots , y_{\log t'}),$$
        where $y \in \pmone^{\log t'}$, $z \in \pmone^{\log \ell}$. 
    \end{defi}

    In the next modification, we replace each target bit by the function $\AB$ on $\log t' + \log \ell$ input bits, as in Definition~\ref{defi:AB}.
    \begin{defi}[(AND-of-Bent)-Target-Addressing Function]
    \label{defi:AAB}    
    For any integers $t,t',\ell \geq 2$, define the function $\AAB : \pmone^{\log t} \times \pmone^{t(\log \ell + \log t')} \to \pmone$ by
    \[
    \AAB = \AD_t \circt \AB.
    \]
    \end{defi}


We define an auxiliary function, which is a modification of the $\AND$ function where the first variable is replaced by that variable times another $\AND$ on a disjoint set of variables.
\begin{defi}[Modified AND]
    \label{defi:mAND}
    For any integers $t'\geq 2, p \geq 1$, define the function $\mAND_{t', p} : \pmone^{\log t' + p} \to \pmone$ by 
    \begin{equation}
        \mAND_{t', p}(y,u) = \AND_{\log t'}(y_1 \AND_p(u), y_2,y_3, \ldots ,y_{\log t'}),
    \end{equation}
    where $y \in \pmone^{\log t'}$ and $u \in \pmone^{p}$.
\end{defi}

In the next modification we replace one of the variables in the first block of $\AD_{t, t'}$ (where the variables in the first block refer to those variables on which $\AND_{\log t'}$ is evaluated when the addressing variables equal $1^{\log t}$) with that variable times the AND of some $p$ variables from the the other blocks.

\begin{defi}[Modified $\AD_{t, t'}$ with Modified AND]
    \label{defi:mAD}
    Let $t, t' \geq 2$ be any integers and let $p$ be an integer such that $(t-1)(\log t') \geq p \geq 1$.
    Let $x \in \pmone^{\log t}$, for each $b \in \pmone^{\log t}$, let $y_b \in \pmone^{\log t'}$.
    Let $u = \cbra{y_{b,i} \vert b \in \pmone^{\log t} \setminus \cbra{1^{\log t}}, i \in [\log t'] }$.
    Fix an arbitrary ordering on the variables in $u$ and let $u_{\leq p}$ be the the first $p$ variables in $u$ according to that order. Define the function $\mAD_{t,t',p}:\pmone^{\log t + t(\log t')} \to \pmone$ by
    \begin{equation}
        \mAD_{t,t',p}(x,y) = \begin{cases}
                       \mAND_{t',p}(y_{1^{\log t}},u_{\leq p}) & \text{ if } x = 1^{\log t}\\
                        \AND_{\log t'}(y_x) & \text{ otherwise }. 
                    \end{cases}
    \end{equation}    
\end{defi}

The table in the following claim summarizes various properties of interest of the functions defined above: rank, sparsity, max-supp-entropy, max-rank-entropy and weight. 
The first row is used to show that our lower bounds on weight beat those obtained from Chang's lemma (Lemma~\ref{lem:chang}) for the function $\AD_{t, t'}$, no matter what threshold is chosen (Claim~\ref{claim: beating changs lemma for all thresholds for ADtt'}).

The second and third rows of the table will be crucial to prove Theorem~\ref{thm:delta upper bound in terms of rk and also k'} (Claims~\ref{claim:setting parameters for tightstraightline} and~\ref{claim:setting parameters for tight curve}). The second and the last row of the table are required to prove Theorem~\ref{thm:delta upper bound in terms of rk and also k''} (Claims~\ref{claim: setting parameters for tight curve for k''} and~\ref{claim:setting parameters for tightstraightline for k''}).

\begin{claim}\label{claim:fnproperties}
    The rank, sparsity, max-supp-entropy, max-rank-entropy and weight of the functions $\AD_{t,t'}$, $\AD_{t,t',a}$, $\AAB$ and $\mAD_{t, t', p}$ are as follows.\footnote{Precise statements along with quantifications on $t, t', \ell, a, p$ are stated in Section~\ref{subsec:tightness} (Claims~\ref{claim:properties of ADtt'}, \ref{claim:properties of ADtt'a},  \ref{claim:properties of AAB} and \ref{claim:properties of mAD}). We do not formally prove the claimed bound on the max-supp-entropy of $\mAD_{t, t', p}$ since we do not require it; this bound can be observed from the proof of Claim~\ref{claim:properties of mAD}.}
    \begin{center}
        \label{table:properties2}
        \begin{tabular}{ |c|c|c|c|c|c| } 
             \hline
            $f$ & $r(f)$ & $k(f)$ & $k'(f)$ & $k''(f)$ & $\delta(f)$\\ 
            \hline
            \hline
             $\AD_{t,t'}$ & $\Theta(t\log t')$ & $\Theta(t^2t')$ & $\Theta(tt')$ & $\Theta(tt')$ & $\frac{1}{t'}$\\
            \hline    
             $\AD_{t,t',a}$ & $\Theta(t\log t' + \log a)$ & $\Theta(t^2t' + ta)$ & $\Theta(at)$ & $\Theta(at)$ & $\frac{1}{t'} + \frac{1}{at} - \frac{1}{tt'}$\\
            \hline 
             $\AAB$ & $\Theta(t(\log t' + \log \ell))$ & $ \Theta(t^2t'\ell)$ & $\Theta(tt'\sqrt{\ell})$ & $\Theta(tt'\sqrt{\ell})$ & $\frac{1}{t'}$\\
            \hline
            $\mAD_{t,t',p}$ & $\Theta(t\log t')$ &  $\Theta(2^p tt' + t^2t')$ & $\Theta(2^ptt')$ & $\Theta(tt')$ & $\frac{1}{t'}$\\
            \hline
        \end{tabular}
    \end{center}
\end{claim}

\subsection{Proof of properties of our constructed functions}\label{subsec:proofofclaimfnproperties}

In this section, we prove Claim~\ref{claim:fnproperties} by computing the properties of interest i.e., rank, sparsity, max-supp-entropy, max-rank-entropy and weight for each of the functions $\AD_{t,t'}, \AD_{t,t',a}, \AAB$ and $\mAD_{t,t',p}$.
We prove a composition lemma (Lemma~\ref{lem:properties of composition of addressing and g}) that relates the rank, sparsity, max-supp-entropy, max-rank-entropy and weight of $\AD_t \circt g$ to those of $g$. 

Since $\AD_{t,t'} = \AD_t \circt \AND_{\log t'}$ and $\AAB = \AD_t \circt \AB$, we are able to use the composition lemma to prove the properties of interest of $\AD_{t,t'}$ and $\AAB$ (Claims~\ref{claim:properties of ADtt'} and~\ref{claim:properties of AAB}, respectively). This proves the bounds corresponding to two of the rows in Table~\ref{table:properties2}. To conclude the proof of Claim~\ref{claim:fnproperties} we prove bounds on the rank, sparsity, max-supp-entropy, max-rank-entropy and weight of $\AD_{t, t', a}$ and $\mAD_{t, t', p}$ from first principles (Claims~\ref{claim:properties of ADtt'a} and~\ref{claim:properties of mAD}, respectively).

We begin by stating a composition lemma.


    \begin{restatable}[Composition lemma]{lemma}{composition}
    \label{lem:properties of composition of addressing and g}
         Let $t \geq 2, m \geq 1$ be any positive integers, and let $g: \pmone^{m} \to \pmone$ be a non-constant function such that  
         there exists a non-empty set $S \subseteq [m]$ with $0 \neq |\wh{g}(S)| \leq |\wh{g}(\emptyset)|$.
        Let $f:\pmone^{\log t + m t} \to \pmone$ be defined as 
        \[
        f = \AD_t \circt g.
        \] 
        Then
        \begin{align}
            r(f) &= t \cdot r(g) + \log t, \label{eq:rank Adressing composition}\\        
            k(f) &= 1 + t^2(k(g) -1), \label{eq:sparsity Adressing composition}\\        
            k'(f) &= t \cdot k'(g), \label{eq:maxEnt Adressing composition}\\
            k''(f) &= t \cdot k''(g),\\
            \delta(f) &= \delta(g). \label{eq:weight Adressing composition}
        \end{align}
    \end{restatable}
We defer the proof of the composition lemma (Lemma~\ref{lem:properties of composition of addressing and g}) to Section~\ref{subsubsec:proof of composition lemma} and proceed to compute the properties of $\AD_{t,t'}$ using the composition lemma.

The other functions we consider in this section can be viewed as modifications of $\AD_{t, t'}$.

We prove the properties of $\AD_{t, t'}$ below to provide insight into our other proofs. Moreover, we use this function to show that Theorem~\ref{thm:delta lower bound in terms of rk and also k'} gives a tight lower bound on $\delta(\AD_{t, t'})$ as opposed to Chang's lemma (Lemma~\ref{lem:chang}) applied with any threshold parameter, which only gives a weaker bound (Claim~\ref{claim: beating changs lemma for all thresholds for ADtt'}).

\begin{restatable}[Properties of $\AD_{t,t'}$]{claim}{ADtt}
    \label{claim:properties of ADtt'}
Fix any integers $t \geq 2$, $t' > 4$  
        let $\AD_{t, t'} : \pmone^{\log t} \times \pmone^{t \log t'}  \to \pmone$ be as in Definition~\ref{defi:ADtt'}. 
        Then,
        \begin{itemize}
        \item $r(\AD_{t,t'}) = t \log t' + \log t$,
        \item $k(\AD_{t,t'}) = 1 + t^2(t' -1)$,
        \item $k'(\AD_{t,t'}) = k''(\AD_{t,t'}) = \frac{tt'}{2}$, and 
        \item $\delta(\AD_{t,t'}) = \frac{1}{t'}$.
        \end{itemize}
\end{restatable}
\begin{proof}
    Recall from Definition~\ref{defi:ADtt'} that $\AD_{t,t'} = \AD \circt \AND$ where $\AND$ is on $\log t'$ bits. Since $t' > 4$,  
    by Observation \ref{obs:properties of AND, Bent and Addressing},
    $|\wh{\AND}(\emptyset)| = 1- \frac{2}{t'} > \frac{2}{t'} = |\wh{\AND}(S)|$ for all $S \ne \emptyset$. Therefore the claim follows by Lemma~\ref{lem:properties of composition of addressing and g} and Observation~\ref{obs:properties of AND, Bent and Addressing}.
\end{proof}
We next compute the properties of $\AAB$.
Since $\AAB = \AD \circt \AB$, we first state the properties of $\AB$ in Claim~\ref{claim:properties of AND of Bent} and then deduce the properties of $\AAB$ using composition lemma (Lemma~\ref{lem:properties of composition of addressing and g}) and Claim~\ref{claim:properties of AND of Bent}. 

    \begin{restatable}[Properties of $\AB$]{claim}{propertiesofAB}
    \label{claim:properties of AND of Bent}
        For any integers $t' > 3, \ell \geq 2$, let $\AB: \pmone^{\log t'} \times \pmone^{\log \ell} \to \pmone$ be as in Definition \ref{defi:AB}. Then,
        \begin{itemize}
            \item $r(\AB) = \log t' + \log \ell$,
            \item $k(\AB) = 1 + \frac{t'}{2} + \frac{\ell t'}{2}$,
            \item $k'(\AB) = k''(\AB) = \frac{t'\sqrt{\ell}}{2}$ and
            \item $\delta(\AB) = \frac{1}{t'}$.
        \end{itemize}
    \end{restatable}
    We defer the proof of Claim~\ref{claim:properties of AND of Bent} to Section~\ref{subsubsec:properties of AND of Bent}.
    The following claim gives the properties of $\AAB$.

\begin{claim}[Properties of $\AAB$]
    \label{claim:properties of AAB}
        For any integers $t\geq 2, t' > 3,\ell \geq 2$, let $\AAB : \pmone^{\log t} \times \pmone^{t (\log \ell+ \log t')} \to \pmone$ be as in Definition~\ref{defi:AAB}. Then,
        \begin{itemize}
        \item $r(\AAB) = t (\log t' + \log \ell) + \log t$,
        \item $k(\AAB) = 1 + \frac{1}{2}t^2(\ell + 1)t'$,
        \item $k'(\AAB) = k''(\AAB) = \frac{tt'\sqrt{\ell}}{2}$, and        
        \item $\delta(\AAB) = \frac{1}{t'}$.         
        \end{itemize}
\end{claim}
\begin{proof}
    Recall from Definition~\ref{defi:AAB} that $\AAB = \AD \circt \AB$ where $\AB$ is on $\log \ell + \log t'$ bits. Since $t' > 3$, $\ell \geq 2$,
    by Claim \ref{claim:properties of AND of Bent}, $|\wh{\AAB}(\emptyset)| = 1- \frac{2}{t'} \geq \frac{2}{t'\sqrt{\ell}} = \frac{1}{k'(\AAB)}$.  
    Therefore the claim follows by Lemma~\ref{lem:properties of composition of addressing and g} and Claim~\ref{claim:properties of AND of Bent}. 
\end{proof}

In the following claim, we deduce the properties of $\AD_{t,t',a}$ from its Fourier expansion.  

\begin{restatable}[Properties of $\AD_{t,t',a}$]{claim}{ADtta}
    \label{claim:properties of ADtt'a}
        Fix any integers $t \geq 2$, $t' \geq 2$ and $a \geq 2t'$. 
        Let $\AD_{t, t',a} : \pmone^{\log t} \times \pmone^{\log a} \times \pmone^{(t-1) \log t'}  \to \pmone$ be as in Definition~\ref{defi:ADtt'a}.        
        Then,
        \begin{itemize}
        \item $r(\AD_{t, t', a}) = (t - 1)\log t' + \log a + \log t$,
        \item $k(\AD_{t, t', a}) = (t-1)(t'-1)t + ta$,
        \item $k'(\AD_{t, t', a}) = k''(\AD_{t,t',a}) =  \frac{ta}{2}$, and
        \item $\delta(\AD_{t, t', a}) = \frac{1}{t'} + \frac{1}{at} - \frac{1}{tt'}$.
        \end{itemize}
\end{restatable}    

We prove Claim~\ref{claim:properties of ADtt'a} in Section~\ref{subsubsec:properties of ADtt'a}. In the following claim, we deduce the properties of $\mAD_{t,t',p}$ from first principles.


\begin{claim}[Properties of $\mAD_{t,t',p}$]
\label{claim:properties of mAD}~
      Let $t, t' \geq 2$ be any integers and let
      be an integer such that $(t-1)(\log t') \geq p \geq 2$. 
      Let $\mAD_{t,t',p}:\pmone^{\log t + t(\log t')} \to \pmone$ be as in Definition~\ref{defi:mAD}. Then, 
      \begin{itemize}
        \item $r(\mAD_{t,t',p}) = t \cdot \log t' + \log t$,
        \item $k(\mAD_{t,t',p}) = \Theta(2^p tt' + t^2t')$,
        \item $k''(\mAD_{t,t',p}) = \Theta(tt')$, and \item $\delta(\mAD_{t,t',p}) = \frac{1}{t'}$.
        \end{itemize}
\end{claim}

We prove Claim~\ref{claim:properties of mAD} in Section~\ref{subsubsec:properties of mAD}.
Finally, the proof of Claim~\ref{claim:fnproperties} follows from some of the claims above.
\begin{proof}[Proof of Claim~\ref{claim:fnproperties}]
    The proof follows from Claims~\ref{claim:properties of ADtt'},~\ref{claim:properties of AAB},~\ref{claim:properties of ADtt'a} and \ref{claim:properties of mAD}.
\end{proof}

\subsubsection{Proof of Composition lemma (Lemma~\ref{lem:properties of composition of addressing and g})}
\label{subsubsec:proof of composition lemma} 
    Let $t \geq 2$ and $m \geq 1$ be any integers.
    For the purpose of the following proof, we introduce the following notation. For any $b \in \pmone^{\log t}$, $T \subseteq [\log t]$ and non-empty $S_b \subseteq [m]$ , define characters $\chi_{b, S_b, T} : \pmone^{\log t} \times \pmone^{t m} \to \pmone$ by $\chi_{b, S_b, T}(x, z) =  \prod_{j \in S_b}z_{b, j} \prod_{i \in T}x_i$. Here $z = (\dots, z_{b}, \dots)$, where $b \in \pmone^{\log t}$ and $z_b \in \pmone^{m}$ for all $b \in \pmone^{\log t}$.
    \begin{proof}[Proof of Lemma~\ref{lem:properties of composition of addressing and g}]
        Let $x \in \pmone^{\log t}$ and 
        $z \in \pmone^{t m}$. By Definition~\ref{defi:composedaddressing} and Observation~\ref{obs:addexpansion},           
        \begin{align}
            f(x,z)
                &= \sum_{b \in \pmone^{\log t}} g(z_b) \cdot \ind_b(x)\nonumber\\
                &= \sum_{b} 
                \left(\sum_{S_b \subseteq [m]} \wh{g}(S_b)\chi_{S_b}(z_b)\right)
                \left(\sum_{T \subseteq [\log t]}  \wh{\ind}_b(T) \chi_T(x)\right)
                 \nonumber\\
                &= \sum_{b} 
                    \bra{\wh{g}(\emptyset)\wh{\ind}_b(\emptyset)
                +  
                    \wh{g}(\emptyset)
                    \sum_{T \neq \emptyset}
                        \wh{\ind}_b(T) \chi_T(x)
                + 
                    \sum_{{S_b} \neq \emptyset}
                        \sum_{T}      
                             \wh{g}({S_b})\wh{\ind}_b(T)  \chi_{S_b}(z_b) \chi_T(x)} \nonumber\\
                &= \underbrace{\sum_{b} 
                    \wh{g}(\emptyset)\wh{\ind}_b(\emptyset)}_{A_1}
                + \underbrace{\sum_{b} 
                    \wh{g}(\emptyset)
                    \sum_{T \neq \emptyset}
                        \wh{\ind}_b(T) \chi_T(x)}_{A_2}
                + \underbrace{\sum_{b,{S_b} \neq \emptyset,T} 
                             \wh{g}({S_b}) \wh{\ind}_b(T)
                             \chi_{S_b}(z_b)\chi_T(x)}_{A_3}. \label{eq:T1T2T3}
        \end{align}
        By Observation~\ref{obs:indexpansion},        
        \begin{align}
                A_1 &= \sum_b \wh{g}(\emptyset) \frac1t = \wh{g}(\emptyset), \nonumber\\
                A_2 &= \wh{g}(\emptyset)
                    \sum_b
                        \sum_{T \neq \emptyset}                
                      \frac{\prod_{i \in T}b_i }{t}\chi_T(x)  \nonumber\\
                    &= \wh{g}(\emptyset) 
                    \sum_{T \neq \emptyset}
                        \chi_T(x)
                        \sum_b
                      \frac{\prod_{i \in T}b_i }{t}
                     = 0, \tag*{by Observation \ref{obs:sum_b char = 0 }} \nonumber\\
                A_3 &= \sum_{b}
                    \sum_{{S_b} \neq \emptyset}
                        \sum_T
                            \frac{\wh{g}(S_b)\cdot \prod_{i \in T} b_i}{t} \chi_{S_b}(z_b) \chi_T(x) \nonumber\\
                &= \sum_{b,{S_b} \neq \emptyset,T} c_{b,{S_b},T} \cdot \chi_{b, S_b, T}(x, z),\nonumber        
        \end{align}
        where $|c_{b,{S_b},T}| = \frac{|\wh{g}(S_b)|}{t}$ for all $b \in \pmone^{\log t}, T \subseteq [\log t]$ and non-empty ${S_b} \subseteq [m]$.
        From Equation~\eqref{eq:T1T2T3} and the above expressions for $A_1, A_2$ and $A_3$, we obtain the following Fourier expansion for $f$.
        \begin{equation}
        \label{eq:composition Fourier}
            f = \wh{g}(\emptyset) + \sum_{b, \emptyset \neq {S_b} \in \supp(g), T}  c_{b,S_b, T} \cdot \chi_{b, S_b, T},
        \end{equation}         
        since $|c_{b,S_b,T}| = \frac{|\wh{g}(S_b)|}{t}$, $c_{b,S_b,t}$ is non-zero iff $\wh{g}(S_b)$ is non-zero.
        Therefore
        \begin{equation}
        \label{eq:supp f}
            \supp(f) = \cbra{\chi_{\emptyset}} \cup \cbra{\chi_{b,S_b,T}| b \in \pmone^{\log t}, \emptyset \neq S_b  \in \supp(g), T \subseteq [\log t]}.
        \end{equation}

        \begin{itemize}
            \item Rank: Fix a Fourier basis $\calB_g$ of $g$ such that $\calB_g \subseteq \supp(g)$ and a character $\chi_U \in \calB_g$. Consider the set of characters
                $$\calB_f = \cbra{\chi_{b,S_b,\emptyset}| b \in \pmone^{\log t}, S_b \in \calB_g} \cup \{\chi_{\mathbf{1}, U, \{i\}}| i \in [\log t] \}.$$
            By Equation~\eqref{eq:supp f}, $\calB_f \subseteq \supp(f)$ and $\supp(f) \subseteq \spann(\calB_f)$. Therefore, \[
            r(f) = |\calB_f| = t |\calB_g| + \log t = t\cdot r(g) + \log t.
            \]
            \item Sparsity:  By Equation~\eqref{eq:supp f}, 
            \[
            k(f) = |\supp(f)| = 1 + t^2(k(g)-1).
            \]
            \item 
            Max-supp-entropy: Recall from Definition~\ref{defi:max entropy, max rank entropy} that $k'(f)$ equals the smallest non-zero Fourier coefficient of $f$ in absolute value. From the Fourier expansion of $f$ given in Equation~\eqref{eq:composition Fourier},
            \begin{align*}
            k'(f) & = \max\cbra{\frac{1}{|\wh{g}(\emptyset)|}, \max\cbra{\frac{t}{|\wh{g}(S)| } : \emptyset \neq S \in \supp(g)}}\\
            & = \max\cbra{\frac{t}{|\wh{g}(S)| } : \emptyset \neq S \in \supp(g)} \tag*{since $|\wh{g}(\emptyset)| \geq \frac1{k'(g)}$ and $t \geq 1$ by assumption}\\
            & = t \cdot k'(g).
            \end{align*}
            \item Max-rank-entropy: Recall from Definition~\ref{defi:max entropy, max rank entropy} that            
                $k''(f) = \argmin_\theta \{\dim(\cS_{\theta}) = r(f)\},$
            where $\cS_{\theta} = \cbra{S: |\wh{f}(S)| \geq \frac{1}{\theta}}$.
            
            From the Fourier expansion of $f$ given in Equation~\eqref{eq:composition Fourier}, the following set $\calB_f$ is a spanning set for the Fourier support of $f$. Let $\calB_g$ be a Fourier basis for $g$ such that $|\wh{g}(S)| \geq \frac{1}{k''(g)}$ for all $S \in \calB_g$. Define
            \begin{align}
                \calB_f = \cbra{\chi_{b, S_b, T} : b \in \pmone^{\log t}, S_b \in \calB_g, T \subseteq [\log t]}
            \end{align}
            One may verify that $\calB_f$ indeed is a spanning set for $\supp(f)$. By Equation~\eqref{eq:composition Fourier}, $|c_{b,{S_b},T}| = \frac{|\wh{g}(S_b)|}{t}$ for all $b \in \pmone^{\log t}, T \subseteq [\log t]$ and non-empty ${S_b} \subseteq [m]$. Hence $k''(f) \leq t \cdot k''(g)$.
            
            It now remains to show that $k''(f) \geq t \cdot k''(g)$. 
            Towards a contradiction, consider a basis $T_f \subseteq \cbra{\chi_{b,S_b,T} : b \in \pmone^{\log t}, S_b \in \supp(g)}$ for $\supp(f)$, with $|\wh{f}(S)| > \frac1{t \cdot k''(g)}$ for all $S \in T_f$. Fix any $b \in \pmone^{\log t}$. Observe that the set $\cbra{\chi_{S_b} : \chi_{b, S_b, T} \in  T_f}$ forms a spanning set for $\supp(g)$. Moreover, since $|c_{b,{S_b},T}| = \frac{|\wh{g}(S_b)|}{t}$ for all $b \in \pmone^{\log t}, T \subseteq [\log t]$ and non-empty ${S_b} \subseteq [m]$ by Equation~\eqref{eq:composition Fourier}, the set $\cbra{\chi_{S_b} : \chi_{b, S_b, T} \in  T_f}$ is such that each of its Fourier coefficients (i.e.~$\wh{g}(S_b)$) has absolute value strictly larger than $\frac1{k''(g)}$, which is a contradiction by the definition of $k''(g)$.
                
            \item Weight: By Observation~\ref{obs:weight, empty Fourier},
            \begin{align*}
                \delta(f) = \frac{1 - \wh{f}(\emptyset)}{2} = \frac{1 - \wh{g}(\emptyset)}{2} = \delta(g),
            \end{align*}
            where the second equality follows by Equation~\eqref{eq:composition Fourier}.
        \end{itemize}    
\end{proof}

\subsubsection{Properties of AND of Bent (Claim~\ref{claim:properties of AND of Bent})}
\label{subsubsec:properties of AND of Bent}
For the purpose of this proof, define the characters $\chi_{S, T} : \pmone^{\log t'} \times \pmone^{\log \ell} \to \pmone$ by $\chi_{S,T}(y, z) = \prod_{i \in S} y_i \cdot \prod_{j \in T} z_j$ for all $S \subseteq [\log t'], T \subseteq [\log \ell]$.

    \begin{proof}[Proof of Claim~\ref{claim:properties of AND of Bent}]
    Recall that from Definition~\ref{defi:AB}, for all $y \in \pmone^{\log t'}$ and $z \in \pmone^{\log \ell}$,
    \begin{equation*}
        \AB(y,z) = \AND(y_1\B(z), y_2, \dots , y_{\log t'}),
    \end{equation*}
    where  
    $\B:\pmone^{\log\ell} \to \pmone$ is a bent function. Since 
    $\AND(y_1,\ldots ,y_{\log t'}) = 1 - 2 \prod_{i=1}^{\log t'} \frac{(1 - y_i)}{2}$,
        \begin{align*}
            \AB(y,z) &= \AND(y_1\B(z), y_2, \dots , y_{\log t'})\\
                &= 1 - (1 - y_1\B(z)) \cdot \prod_{i=2}^{\log t'} \frac{(1 - y_i)}{2}\\
                &= 1 - \bra{ 1 - y_1\sum_{T \subseteq [\log \ell]} \wh{B}(T) \chi_T(z) } \cdot \bra{\frac{2}{t'}\sum_{1 \not\in S \subseteq [\log t']} (-1)^{|S|} \chi_S(y)}\\
                &= 1 
                - \frac{2}{t'} \bra{\sum_{1 \not \in S \subseteq [\log t']} (-1)^{|S|} \chi_S(y)}
                + \frac{2}{t'} y_1 \bra{\sum_{T \subseteq [\log \ell]} \sum_{1 \in S \subseteq [\log t']} (-1)^{|S|} \wh{B}(T) \chi_S(y)\chi_T(z)}
        \end{align*}

    Hence, the Fourier expansion of $\AB$ is given by
    \begin{equation}
    \label{eq: abFourier}
        \AB = \bra{1 - \frac{2}{t'}} \chi_{\emptyset, \emptyset} - \frac{2}{t'}\bra{\sum_{\substack{1 \notin S \subseteq [\log t']\\ S \neq \emptyset}} (-1)^{|S|} \chi_{S, \emptyset}} + \frac{2}{t'} y_1
                \bra{\sum_{
                    \substack{
                        1 \in S \subseteq [\log t']\\
                        T \subseteq [\log \ell]
                        }
                    } (-1)^{|S|} \wh{B}(T) \chi_{S,T}}.
    \end{equation}
    Since $\wh{B}(T) \neq 0$ for all $T \subseteq [\log \ell]$ by Definition~\ref{defi:bent}, Equation~\eqref{eq: abFourier} implies
    \begin{equation}
    \label{eq:supp AB}
        \supp(\AB) = \cbra{(\emptyset, \emptyset)} 
                \cup \cbra{(S, \emptyset) : S \neq \emptyset,1 \notin S \subseteq [\log t']}
        \cup \cbra{(S, T) : 1 \in S \subseteq [\log t'], T \subseteq [\log \ell]}.
    \end{equation}
    \begin{itemize}
        \item Rank: Consider  
        $\calB_B = \cbra{ (\cbra{i}, \emptyset): i \in [\log t']} \cup \cbra{(\cbra{1}, \cbra{j}): j \in [\log \ell]}$. 
         
        By Equation~\eqref{eq:supp AB}, $\calB_B \subseteq \supp(\AB)$. Moreover $\calB_B$ is a linearly independent set and generates all the characters.
        Therefore $\calB_B$ forms a Fourier basis of $\AB$. Hence,
        \[
        r(\AB) = |\calB_B| = \log t' + \log \ell.
        \]
        \item Sparsity:
        From Equation~\eqref{eq:supp AB},
        \[
        k(\AB) = |\supp(\AB)| = \frac{t'}{2} + \frac{\ell t'}{2}.
        \]
        \item Max-supp-entropy: Recall from Definition \ref{defi:max entropy, max rank entropy} that $k'(\AB)$ equals the smallest non-zero Fourier coefficient of $\AB$ in absolute value. From the Fourier expansion of $\AB$ given in Equation~\eqref{eq: abFourier},
            \begin{align*}
            k'(\AB) & = \max\cbra{\frac{t'}{t' -2}, \frac{t'}{2}, \max\cbra{\frac{t'}{2|\wh{B}(T)| } : T \subseteq [\log \ell]}}\\
            &= \frac{t'\sqrt{\ell}}{2}. \tag*{by Observation~\ref{obs: bentfcoeffs}}
            \end{align*}
        
        \item Max-rank-entropy: Recall from Definition \ref{defi:max entropy, max rank entropy}, 
        $k''(f) = \argmin_\theta \{\dim(\cS_{\theta}) = r(f)\},$
            where $\cS_{\theta} = \cbra{S: |\wh{f}(S)| \geq \frac{1}{\theta}}$.
        Observe from the Fourier expansion of $\AB$ in Equation~\eqref{eq: abFourier} that the only monomials which containing $z$-variables are $\chi_{S,T}$ such that $T \neq \emptyset$.        
        Any such monomial has coefficient whose absolute value is $\frac{2}{t'\sqrt{\ell}}$.        
        So, $k''(\AB) \geq \frac{t'\sqrt{\ell}}{2}$. 
        Furthermore by Lemma~\ref{lem:relationships between rk and k'} $k''(\AB) \leq k'(\AB) = \frac{t'\sqrt{\ell}}{2}$. Therefore 
        \[
        k''(\AB) = \frac{t'\sqrt{\ell}}{2}.
        \]
        \item Weight: Observation~\ref{obs:weight, empty Fourier} and
        Equation~\eqref{eq: abFourier} imply
        \[
        \delta(\AB) = \frac{1 -\wh{\AB}(\emptyset, \emptyset)}{2} = \frac{1}{t'}.
        \]
    \end{itemize}
    
\end{proof}

\subsubsection{Properties of $\AD_{t,t',a}$ (Claim~\ref{claim:properties of ADtt'a})}
\label{subsubsec:properties of ADtt'a}

    Let $t \geq 2$, $t' \geq 2$ and $a \geq 2t'$ be any integers. Consider the function $\AD_{t, t', a} : \pmone^{\log t} \times \pmone^{\log a} \times \pmone^{(t-1)\log t'}  \to \pmone$ as in Definition~\ref{defi:ADtt'a}. 
    For the purpose of the following proof, we introduce the following notation.
    Let $\mathbf{1} := 1^{\log t}$. For any $\mathbf{1} \neq b \in \pmone^{\log t}$, $T \subseteq [\log t]$ and non-empty $S_b \subseteq [\log t']$ , define characters $\chi_{b, S_b, T} : \pmone^{\log t} \times \pmone^{\log a} \times \pmone^{(t-1)\log t'}  \to \pmone$ by 
    \begin{align*}
        \chi_{b, S_b, T}(x, z) =  \prod_{j \in S_b}z_{b, j} \prod_{i \in T}x_i.
    \end{align*}
     Here $z = (\dots, z_{b}, \dots)$, where $b \in \pmone^{\log t}$, $z_{\bone} \in \pmone^{\log a}$ and $z_b \in \pmone^{\log t'}$ for all $b \in \pmone^{\log t} \setminus \cbra{\bone}$.
    Also for $b = \mathbf{1}$, $T \subseteq [\log t]$ and non-empty $S_{\mathbf{1}} \subseteq [\log a]$, define characters $\chi_{\mathbf{1}, S_{\mathbf{1}}, T} : \pmone^{\log t}  \times \pmone^{\log a} \times \pmone^{(t-1)\log t'} \to \pmone$ by 
    \begin{align*}
        \chi_{\mathbf{1}, S_{\mathbf{1}}, T}(x, z) =  \prod_{j \in S_{\mathbf{1}}}z_{\mathbf{1}, j} \prod_{i \in T}x_i.
    \end{align*}
    For any set $U \subseteq [\log t]$, define characters $\chi_{\emptyset, U} : \pmone^{\log t}  \times \pmone^{\log a} \times \pmone^{(t-1)\log t'} \to \pmone$ by 
    \[
    \chi_{\emptyset, U} (x, z) = \prod_{i \in U}x_i.
    \]
    \begin{proof}[Proof of Claim~\ref{claim:properties of ADtt'a}]
        Let $x \in \pmone^{\log t}$ and $z \in \pmone^{\log a} \times \pmone^{(t-1) \log t'} $ be such that $z_{\mathbf{1}} \in \pmone^{\log a}$ and $z_{b} \in \pmone^{\log t'}$ for $\mathbf{1} \neq b \in \pmone^{\log t}$. By Definition~\ref{defi:ADtt'a} and Observation~\ref{obs:addexpansion},
        \begin{equation}
        \label{eq:topofhugeandproof}
            \AD_{t,t',a}(x,z)
                = \underbrace{\AND(z_{\mathbf{1}})\cdot \ind_{\mathbf{1}}(x)}_{A} 
                + \underbrace{\sum_{\mathbf{1} \neq b \in \pmone^{\log t}} \AND(z_b) \cdot \ind_b(x)}_{B}.
        \end{equation}
        We first analyze $A$ from Equation~\eqref{eq:topofhugeandproof}.
        In this analysis, $\AND$ is on $\log a$ variables.
        \begin{align}
            A &= 
                \left(\sum_{{S_\mathbf{1}} \subseteq [\log a]} \wh{\AND}({S_\mathbf{1}})\chi_{S_\mathbf{1}}(z_{\mathbf{1}})\right)
                \left(\sum_{T \subseteq [\log t]}  \wh{\ind}_{\mathbf{1}}(T) \chi_T(x)\right)
                 \nonumber\\
                &= 
                \bra{\wh\AND(\emptyset) + \sum_{{S_\mathbf{1}} \neq \emptyset}\wh\AND({S_\mathbf{1}}) \chi_{S_\mathbf{1}}(z_{\mathbf{1}})}
                \bra{\wh{\ind}_{\mathbf{1}}(\emptyset) + \sum_{T \neq \emptyset}\wh{\ind}_{\mathbf{1}}(T) \chi_T(x)} \nonumber\\
                &=  
                    \bra{\wh{\AND}(\emptyset)\wh{\ind}_{\mathbf{1}}(\emptyset)
                +                    \wh{\AND}(\emptyset)
                    \sum_{T \neq \emptyset}
                        \wh{\ind}_{\mathbf{1}}(T) \chi_T(x)
                +                    \sum_{{S_\mathbf{1}} \neq \emptyset}
                        \sum_{T}      
                             \wh{\AND}({S_\mathbf{1}})\wh{\ind}_\mathbf{1}(T)  \chi_{S_\mathbf{1}}(z_\mathbf{1}) \chi_T(x)} \nonumber\\
                &= \underbrace{ 
                    \wh{\AND}(\emptyset)\wh{\ind}_{\mathbf{1}}(\emptyset)}_{A_1}
                + \underbrace{ 
                    \wh{\AND}(\emptyset)
                    \sum_{T \neq \emptyset}
                        \wh{\ind}_{\mathbf{1}}(T) \chi_T(x)}_{A_2}
                + \underbrace{\sum_{{S_\mathbf{1}} \neq \emptyset,T} 
                             \wh{\AND}({S_\mathbf{1}}) \wh{\ind}_{\mathbf{1}}(T)
                             \chi_{S_\mathbf{1}}(z_{\mathbf{1}})\chi_T(x)}_{A_3}.\label{eq:A1A2A3}   
        \end{align}
        By Fact~\ref{fact:Fourier AND} and Observation~\ref{obs:indexpansion},
        \begin{align}
                A_1 &= \left(1 - \frac{2}{a}\right)\frac1t , \label{eq:a1}\\
                A_2 &=
                    \left(1- \frac{2}{a}\right) 
                    \sum_{T \neq \emptyset}
                    \frac{\chi_{T}(x)}{t}
                    =\frac{1}{t}\left(1- \frac{2}{a}\right) 
                    \sum_{T \neq \emptyset}
                    \chi_{T}(x), \label{eq:a2}\\
                A_3 &=
                    \sum_{{S_\mathbf{1}} \neq \emptyset}
                        \sum_T
                            \frac{2(-1)^{|{S_\mathbf{1}}|+1}}{at} \chi_{S_\mathbf{1}}(z_\mathbf{1}) \chi_T(x). \label{eq:a3}     
        \end{align}
        We next analyze $B$ from Equation~\eqref{eq:topofhugeandproof}. In this analysis, $\AND$ is on $\log t'$ variables.
        \begin{align}
            B &= \sum_{b \neq \mathbf{1}} 
                \left(\sum_{{S_b} \subseteq [\log t']} \wh{\AND}({S_b})\chi_{S_b}(z_b)\right)
                \left(\sum_{T \subseteq [\log t]}  \wh{\ind}_b(T) \chi_T(x)\right)
                 \nonumber\\
                &= \sum_{b \neq \mathbf{1}} 
                \bra{\wh\AND(\emptyset) + \sum_{{S_b} \neq \emptyset}\wh\AND({S_b}) \chi_{S_b}(z_b)}
                \bra{\wh{\ind}_b(\emptyset) + \sum_{T \neq \emptyset}\wh{\ind}_b(T) \chi_T(x)} \nonumber\\
                &= \sum_{b \neq \mathbf{1}} 
                    \bra{\wh{\AND}(\emptyset)\wh{\ind}_b(\emptyset)
                +                    \wh{\AND}(\emptyset)
                    \sum_{T \neq \emptyset}
                        \wh{\ind}_b(T) \chi_T(x)
                + \sum_{{S_b} \neq \emptyset}
                        \sum_{T}      
                             \wh{\AND}({S_b})\wh{\ind}_b(T)  \chi_{S_b}(z_b) \chi_T(x)} \nonumber\\
                &= \underbrace{\sum_{b \neq \mathbf{1}} 
                    \wh{\AND}(\emptyset)\wh{\ind}_b(\emptyset)}_{B_1}
                + \underbrace{\sum_{b \neq \mathbf{1}} 
                    \wh{\AND}(\emptyset)
                    \sum_{T \neq \emptyset}
                        \wh{\ind}_b(T) \chi_T(x)}_{B_2}
                + \underbrace{\sum_{b \neq \mathbf{1},{S_b} \neq \emptyset,T} 
                             \wh{\AND}({S_b}) \wh{\ind}_b(T)
                             \chi_{S_b}(z_b)\chi_T(x)}_{B_3}.
        \end{align}
        
        By Fact~\ref{fact:Fourier AND} and Observation~\ref{obs:indexpansion},
        \begin{align}
                B_1 &= \sum_{\mathbf{1} \neq b \in \pmone^{\log t}} \left(1 - \frac{2}{t'}\right)\left(\frac{1}{t}\right) 
                = \bra{1- \frac{2}{t'}} \bra{1- \frac1t},\label{eq:b1}\\
                B_2 &=  \left(1- \frac{2}{t'}\right) 
                    \sum_{T \neq \emptyset}
                    \sum_{b \neq \mathbf{1}}
                       \frac{\prod_{i \in T}b_i}{t}  \chi_T(x)\nonumber
                       \\
                    &= \bra{1-\frac{2}{t'}}\sum_{T \neq \emptyset} \frac{(-1)}{t} \chi_T(x)
                    =\bra{\frac{-1}{t}}\bra{1-\frac{2}{t'}}\sum_{T \neq \emptyset} 
                    \chi_T(x),\label{eq:b2}\\          
                B_3 &= \sum_{b \neq \mathbf{1}}
                    \sum_{{S_b} \neq \emptyset}
                        \sum_T
                            \frac{2(-1)^{|{S_b}|+1}\cdot \prod_{i \in T} b_i}{tt'} \chi_{S_b}(z_b) \chi_T(x), \label{eq:b3}  
        \end{align}
        where Equation~\eqref{eq:b2} follows since $\sum_{b \neq \mathbf{1}} \prod_{i \in T}b_i = -1$ for any non-empty $T \subseteq [\log t]$ by Observation \ref{obs:sum_b char = 0 }.
        From Equation~\eqref{eq:topofhugeandproof}, 
        \begin{align}
        \label{eq:a123b123}
            \AD_{t,t',a}(x, z) &= A+B = A_1+B_1 + A_2+B_2 + A_3+B_3.
        \end{align}    
        Observe that the only terms from Equation~\eqref{eq:a123b123} that contribute to $c_0$ are $A_1$ and $B_1$. Moreover, we have from Equations~\eqref{eq:a1} and~\eqref{eq:b1} that
        \begin{equation}\label{eq:c0}
        c_0 = A_1 + B_1 = \left(1 - \frac{2}{a}\right)\frac1t + \bra{1- \frac{2}{t'}} \bra{1- \frac1t} = 1 + \frac{2}{tt'} - \frac{2}{at} - \frac{2}{t'}.
        \end{equation}
        Next observe that the only terms contributing to $c_U$ for $\emptyset \neq U \subseteq [\log t]$ appear in $A_2$ and $B_2$. Matching coefficients we obtain from Equations~\eqref{eq:a2} and~\eqref{eq:b2} that for any non-empty $U \subseteq [\log t]$,
        \begin{equation}\label{eq:cU}
        c_U = \frac{1}{t}\left(1 - \frac{2}{a}\right) - \frac{1}{t}\left(1 - \frac{2}{t'}\right) = \frac{2}{tt'}-\frac{2}{at}.
        \end{equation}
        Next, the only terms contributing to $c_{b, S_b, T}$ for $\mathbf{1} \neq b \in \pmone^{\log t}$, non-empty $S_b \subseteq [\log t']$ and $T \subseteq [\log t]$ arise from $B_3$. By comparing coefficients we obtain from Equation~\eqref{eq:b3} that for any $\mathbf{1} \neq b \in \pmone^{\log t}$, non-empty $S_b \subseteq [\log t']$ and $T \subseteq [\log t]$,
        \begin{equation}\label{eq:cbsbt}
        |c_{b, S_b, T}| = \frac{2}{tt'}.
        \end{equation}
        Finally the only term that contributes to $c_{\mathbf{1}, S_\mathbf{1}, T}$ for non-empty $S_{\mathbf{1}} \subseteq [\log a]$ and $T \subseteq [\log t]$ is $A_3$. Matching coefficients, we obtain from Equation~\eqref{eq:a3} that for any non-empty $S_{\mathbf{1}} \subseteq [\log a]$ and $T \subseteq [\log t]$,
        \begin{equation}\label{eq:c1s1t}
        |c_{\mathbf{1}, S_{\mathbf{1}}, T}| = \frac{2}{at}.
        \end{equation}
        Moreover, all terms that appear in Equation~\eqref{eq:a123b123} appear in the cases covered above. This proves the claim.
        
        Thus the Fourier expansion of $\AD_{t, t', a}$ is given by
        \begin{equation}\label{eq: hugeandFourier}
            \AD_{t,t',a} 
            = c_0 
            + \sum_{\emptyset \neq U \subseteq [\log t]} c_U\chi_{\emptyset, U} 
            + \sum_{
                \substack{
                    \mathbf{1} \neq b  \in \pmone^{\log t},\\
                    \emptyset \neq S_b \subseteq [\log t'],\\ 
                    T \subseteq [\log t]}
                }   c_{b,S_b,T} \cdot \chi_{b,S_b,T} 
            + \sum_{
                \substack{
                    \emptyset \neq S_{\mathbf{1}} \subseteq [\log a],\\
                    T \subseteq [\log t]
                }
                } c_{\mathbf{1},S_{\mathbf{1}},T} \cdot \chi_{\mathbf{1},S_{\mathbf{1}},T},
        \end{equation}
        where $c_0$ is as in Equation~\eqref{eq:c0}, $c_U$ is as in Equation~\eqref{eq:cU} for all $\emptyset \neq U \subseteq [\log t]$, $c_{b,S_b,T}$ is as in Equation~\eqref{eq:cbsbt} for all $\mathbf{1} \neq b  \in \pmone^{\log t}, \emptyset \neq S_b \subseteq [\log t']$ and $T \subseteq [\log t]$, and $c_{\mathbf{1},S_{\mathbf{1}},T}$ is as in Equation~\eqref{eq:c1s1t} for all $\emptyset \neq S_{\mathbf{1}} \subseteq [\log a]$ and $T \subseteq [\log t]$.
        
        In the Fourier expansion of $\AD_{t,t',a}$ from Equation~\eqref{eq: hugeandFourier}, coefficients in the third and fourth summands are non-zero from  Equations~\eqref{eq:c1s1t}~\eqref{eq:cbsbt},
        coefficients in the second summand are non-zero by Equation~\eqref{eq:cU} since $a \geq 2t'$. Finally, by Equation~\eqref{eq:c0}, $c_0 \neq 0$ since $1 + \frac{2}{tt'} - \frac{2}{at} - \frac{2}{t'} \geq 1 + \frac{1}{tt'} - \frac{2}{t'} \geq \frac{1}{tt'} \geq 0$ as $a \geq 2t' \geq 4$.
        From Equation~\eqref{eq: hugeandFourier},  
        \begin{align}
            \supp(\AD_{t,t',a}) = 
                \cbra{\chi_{\emptyset, \emptyset}} 
                &\cup\cbra{\chi_{\emptyset, U}| \emptyset \neq U \subseteq [\log t]}
                \cup  
                \cbra{\chi_{\mathbf{1},S_{\mathbf{1}},T}|  \emptyset \neq S_{\mathbf{1}}  \subseteq [\log a], T \subseteq [\log t]} \nonumber\\        
                &\cup \cbra{\chi_{b,S_b,T}| b \in \pmone^{\log t} \setminus \cbra{\mathbf{1}}, \emptyset \neq S_b  \subseteq [\log t'], T \subseteq [\log t]}.\label{eq:supp ADtt'a}
        \end{align}
        
   \begin{itemize}
        \item Rank:
        Consider the set of characters 
        \[
        \calB = \cbra{\chi_{b, \cbra{i}, \emptyset} : b \in \pmone^{\log t} \setminus \cbra{ \mathbf{1}}, i \in [\log t']} \cup \cbra{\chi_{\mathbf{1}, \cbra{i}, \emptyset} : i \in [\log a]}  
        \cup \cbra{\chi_{\mathbf{1}, \cbra{1}, j} : j \in [\log t]}.
        \]
    These characters can be seen to be linearly independent and span all monomials.
    Moreover, by Equation~\eqref{eq:supp ADtt'a}, $\calB \subseteq \supp(\AD_{t,t',a})$.
    Therefore,
    \begin{align*}
        r(\AD_{t, t', a}) & = |\calB| = (t -1)\log{t'} + \log a + \log t.               
    \end{align*}
        \item Sparsity: 
            By Equation~\eqref{eq:supp ADtt'a},
            \begin{align*}
                k(\AD_{t,t',a}) &= |\supp(\AD_{t,t',a})|\\
                            &= 1 + t-1 + (a-1)t + (t-1)(t'-1)t\\
                            &= (t-1)(t'-1)t + ta.
            \end{align*}
        \item Max-supp-entropy:
        Recall from Definition~\ref{defi:max entropy, max rank entropy} that $k'(\AD_{t, t', a})$ equals the inverse of the smallest non zero Fourier coefficient in absolute value. 
        From the Fourier expansion of $\AD_{t, t', a}$ given in Equation~\eqref{eq: hugeandFourier}, the candidates for smallest nonzero coefficient in absolute value are $\cbra{1+\frac{2}{tt'}-\frac{2}{at} -\frac{2}{t'}, \frac{2}{tt'} - \frac{2}{at}, \frac{2}{tt'}, \frac{2}{at}}$. Since $t' \geq 3$, $1+\frac{2}{tt'}-\frac{2}{at} -\frac{2}{t'} \geq \frac{2}{tt'} - \frac{2}{at}$. Since, $a \geq 2t'$, $\frac{2}{tt'} - \frac{2}{at} \geq \frac{2}{at}$ and $\frac{2}{tt'} \geq \frac{2}{at}$. Therefore
    \begin{align}
    \label{eq: maxent adtt'}
            k'(\AD_{t,t',a}) &=\frac{ta}{2}. 
        \end{align}
        
        \item Max-rank-entropy:
            Recall from Definition~\ref{defi:max entropy, max rank entropy} that
                $k''(f) = \argmin_\theta \{\dim(\cS_{\theta}) = r(f)\},$
            where $\cS_{\theta} = \cbra{S: |\wh{f}(S)| \geq \frac{1}{\theta}}$.
            From the Fourier expansion of $\AD_{t,t',a}$ given in Equation~\eqref{eq: hugeandFourier}, observe that every monomial which involves a variable from $z_{\mathbf{1}}$ has coefficient whose absolute value equals $\frac{2}{at}$.
            Thus, if $\theta < \frac{at}{2}$, $S_\theta$ does not include any momonial containing a variable from $z_{\mathbf{1}}$. Therefore $k''(\AD_{t,t',a}) \geq \frac{at}{2}$. By Lemma~\ref{lem:relationships between rk and k'}, $k''(\AD_{t,t',a}) \leq k'(\AD_{t,t',a}) = \frac{at}{2}$. Therefore, by Equation~\eqref{eq: maxent adtt'} 
            \[
            k''(\AD_{t,t'a}) = \frac{at}{2}.
            \]
            \item Weight: From Observation~\ref{obs:weight, empty Fourier} and Equation~\eqref{eq:c0},
        \begin{align*}
            \delta(\AD_{t, t', a}) 
            &= \frac{1 - \wh{\AD_{t,t',a}}(\emptyset, \emptyset)}{2}\\
            &= \frac{1}{t'} + \frac{1}{at} - \frac{1}{tt'}.
        \end{align*}
    \end{itemize}
 \end{proof}       

\subsubsection{Properties of $\mAD_{t,t',p}$ (Claim~\ref{claim:properties of mAD})}
\label{subsubsec:properties of mAD}
    Recall that we constructed the Modified AND function (Definition~\ref{defi:mAND}) by replacing one variable by that variable times the product of an AND function of other variables. The next claim computes the Fourier coefficients of $\mAND_{t', p}$.
    
    For the purpose of the following claim, for any $S \subseteq [\log t']$ and $T \subseteq [p]$, define characters $\chi_{S, T} : \pmone^{\log t' + p} \to \pmone$ by $\chi_{S, T}(y, u) = \prod_{i \in S}y_i \prod_{j \in T} u_j$.
    \begin{claim}
    \label{claim:mAND Fourier expansion}
        Let $t' \geq 2, p \geq 2$ be any integers and let $f = \mAND_{t',p} : \pmone^{\log t' + p} \to \pmone$ be as in Definition~\ref{defi:mAND}. Then $f = \sum_{S \subseteq [\log t'], T \subseteq [p]} \wh{f}(S, T)\chi_{S, T}$, where
        \begin{align*}
            \wh{f}(S, T) = 
                \begin{cases}
                    1 - \frac{2}{t'} & S = T = \emptyset\\
                    \frac{2 \cdot (-1)^{|S|}}{t'} & T = \emptyset, 1 \notin S \subseteq [\log t'], S \neq \emptyset \\
                    \frac{2 \cdot (-1)^{|S|}}{t'}\bra{1 - \frac{2}{2^p}} & T = \emptyset, 1 \in S \subseteq [\log t']\\
                    \frac{4 \cdot (-1)^{|S| + |T| + 1}}{2^p t'} & \emptyset \neq T \subseteq [p], 1 \in S \subseteq [\log t'].
                \end{cases}
        \end{align*}
    \end{claim}

For the purpose of the proof, recall that we view inputs to $\mAND_{t',p}$ as $(y, u)$, where $y \in \pmone^{\log t'}$ and $u \in \pmone^{p}$.
    
\begin{proof}[Proof of Claim~\ref{claim:mAND Fourier expansion}]
    We have $\AND_{\log t'}(y_1,\ldots ,y_{\log t'}) = 1 - 2 \prod_{i=1}^{\log t'} \frac{(1 - y_i)}{2}$. Thus, by Definition~\ref{defi:mAND},
    \begin{align*}
        \mAND_{t',p}(y,u) &= \AND_{\log t'}(y_1 \AND_p(u),y_2,\dots,y_{\log t'})\\        
        &= 1 - (1 - y_1 \AND_p(u)) \prod_{i=2}^{\log t'}\frac{(1-y_i)}{2}\\
        &= 1 - \bra{1 - y_1 \sum_{T \subseteq [p]}\wh{\AND_p}(T) \chi_T(u)} \bra{\frac{2}{t'}\sum_{1 \notin S \subseteq [\log t']} (-1)^{|S|} \chi_S(y)}\\
        &= 1 - \frac{2}{t'}  \bra{\sum_{1 \notin S \subseteq [\log t']} (-1)^{|S|} \chi_S(y)}
        + \frac{2}{t'} y_1\bra{\sum_{T \subseteq [p]} \sum_{1 \in S \subseteq [\log t']} (-1)^{|S|} \wh{\AND_p}(T) \chi_S(y)\chi_T(u)}\\
        &= 1 - \frac{2}{t'}  \bra{\sum_{1 \not \in S \subseteq [\log t']} (-1)^{|S|} \chi_S(y)}
        + \frac{2}{t'}\bra{1 - \frac{2}{2^p}}y_1
        \bra{\sum_{1 \in S \subseteq [\log t']} (-1)^{|S|} \chi_S(y)}\\
        &+ \frac{4}{2^pt'} 
        \bra{\sum_{\emptyset \neq T \subseteq [p]} \sum_{1 \in S \subseteq [\log t']} (-1)^{|S|+|T|+1} \chi_S(y)\chi_T(u)}
    \end{align*}
This proves the claim.
\end{proof}

We now prove the required properties of $\mAD_{t,t',p}$.
\begin{proof}[Proof of Claim~\ref{claim:properties of mAD}]~
Define $\bone := 1^{\log t}$. Recall that on input $(x, y) \in \pmone^{\log t+ t \log t'}$, we define the set of variables $u = \cbra{y_{b,i}| b \in \pmone^{\log t} \setminus \cbra{\bone}, i \in [\log t'] }$. We also fix an arbitrary ordering on the variables in $u$ and let $u_{\leq p}$ be the the first $p$ variables in $u$ according to that order. By Definition~\ref{defi:mAD}, we have
\begin{equation}
    \label{eq:mAD expansion}
     \mAD_{t,t',p}(x,y) = \underbrace{\ind_\bone(x) \cdot \mAND_{t',p}(y_\bone, u_{\leq p})}_{T_\bone} + \sum_{\bone \neq b \in \pmone^{\log t}} \underbrace{\ind_b(x) \cdot \AND_{\log t'}(y_b)}_{T_b}
\end{equation} 
We use Claim~\ref{claim:mAND Fourier expansion} to expand $T_\bone$ as
\begin{align}    
    \label{eq:Fourier ADtt'p T1}
    T_\bone =& \ind_\bone(x) \left[1- \frac{2}{t'} + \sum_{\substack{1 \notin S\\ \emptyset \neq S \subseteq [\log t']}}\frac{2 \cdot (-1)^{|S|}}{t'}\prod_{j \in S}y_{\bone, j} + \sum_{1 \in S \subseteq [\log t']} \frac{2 \cdot (-1)^{|S|}}{t'}\bra{1 - \frac{2}{2^p}} \prod_{j \in S} y_{\bone, j} \right. \nonumber \\
    & \left. + \sum_{\substack{\emptyset \neq T \subseteq [p] \\ 1 \in S \subseteq [\log t']}} \frac{4 \cdot (-1)^{|S| + |T| + 1}}{2^p t'} \prod_{j \in S, \ell \in T} y_{\bone, j} u_\ell \right], 
\end{align}
and Fact~\ref{fact:Fourier AND} to expand $T_b$, for $b \neq \bone$, as
\begin{equation}
    \label{eq:Fourier ADtt'p Tb}
    T_b =
    \ind_b(x) \left[1 - \frac{2}{t'} + \frac{2}{t'} \sum_{S \neq \emptyset}  (-1)^{|S|}\prod_{j \in S} y_{b,j} \right].
\end{equation}

\begin{itemize}
    \item Rank:      
    For all $b \in \pmone^{\log t}$, define 
    \[
    \calB_b = \cbra{ y_{b,j}|j \in [\log t']}.
    \]
    
    From Equations~\eqref{eq:Fourier ADtt'p T1} and~\eqref{eq:Fourier ADtt'p Tb} the monomials from $\calB_\bone$ occur only in the term $T_\bone$. 
    Since $\wh{\ind_\bone}(\emptyset) = \frac1t$ by Observation~\ref{obs:indexpansion}, Equation~\eqref{eq:Fourier ADtt'p T1} yields that the absolute value of the coefficient of $y_{\bone,1}$ is $\bra{1 - \frac{2}{2^p}} \frac{2}{tt'}$, and the absolute value of the coefficient of $y_{\bone,j}$ is $\frac{2}{tt'}$ for all $j \in [\log t'] \setminus \cbra{1}$.
        
    Similarly, for $\bone \neq b \in \pmone^{\log t}$, from Equations~\eqref{eq:Fourier ADtt'p T1} and~\eqref{eq:Fourier ADtt'p Tb} the monomials from $\calB_b$ occur only in the term $T_b$. Since $\wh{\ind_b}(\emptyset) = \frac1t$ by Observation~\ref{obs:indexpansion}, Equation~\eqref{eq:Fourier ADtt'p Tb} yields that the absolute value of the coefficient of $y_{b,j}$ is $\frac{2}{tt'}$ for all $j \in [\log t']$ and $\bone \neq b \in \pmone^{\log t}$.        
    
    Fix $\bone \neq c \in \pmone^{\log t}$, and define \[
    \calB_{x} = \cbra{y_{c,1}\cdot  x_i| i \in [\log t]}.
    \]
    From Equations~\eqref{eq:Fourier ADtt'p T1} and~\eqref{eq:Fourier ADtt'p Tb} the monomials from $\calB_x$ occur only in the term $T_c$. 
    Since $\wh{\ind_b}(\emptyset) = \frac1t$ by Observation~\ref{obs:indexpansion}, Equation~\eqref{eq:Fourier ADtt'p Tb} yields that the absolute value of the coefficient of $y_{c,1} \cdot x_i$ is $\frac{2}{tt'}$ for all $i \in [\log t]$. 

    Therefore $\bigcup_{b \in \pmone^{\log t}} \calB_b \cup \calB_x \subseteq \supp(\mAD_{t,t',p})$.
    Since $\bigcup_{b \in \pmone^{\log t}} \calB_b \cup \calB_x$ generate all monomials, 
    \[
    r(\mAD_{t,t',p}) = t\log t' + \log t.
    \]
    
    \item Sparsity:
        By Equation~\eqref{eq:Fourier ADtt'p T1}, all monomials appearing in $T_\bone$, except for those purely in $x$-variables, contain at least one variable from $y_\bone$. Moreover, from Equation~\eqref{eq:Fourier ADtt'p Tb}, no monomial appearing in $T_b$ for $b \neq \bone$ contains a variable from $y_\bone$. Since all Fourier coefficients of $\ind_\bone$ are non-zero by Observation~\ref{obs:indexpansion}, Equation~\eqref{eq:Fourier ADtt'p T1} yields that these monomials contribute
        \begin{equation}
        \label{eq:adtt'p sparsity from tbone}
        t\bra{\frac{t'}{2} -1 + \frac{t'}{2} + (2^p-1)\frac{t'}{2}} = t\bra{\frac{t'}{2}-1+2^{p-1}t'}= \Theta(2^ptt').        
        \end{equation}
        to the sparsity of $\mAD_{t, t', p}$.

        By Equation~\eqref{eq:Fourier ADtt'p Tb}, all monomials all monomials appearing in $T_b$, for any $b \neq \bone$, except for those purely in $x$-variables, contain at least one variable from $y_b$. Moreover, from Equations~\eqref{eq:Fourier ADtt'p T1} and~\eqref{eq:Fourier ADtt'p Tb}, no monomial appearing in $T_b'$ for $b' \neq b$ contains a variable from $y_b$. Since all Fourier coefficients of $\ind_b$ are non-zero by Observation~\ref{obs:indexpansion}, Equation~\eqref{eq:Fourier ADtt'p Tb} yields that these monomials contribute at least (including contributions from each $T_b$ for $b \neq \bone$)
        \begin{equation}
        \label{eq:adtt'p sparsity from tb}
            (t-1)\cdot t \cdot (t' - 1) = \Omega(t^2t')
        \end{equation}
        to the sparsity of $\mAD_{t, t', p}$.
        By Equations~\eqref{eq:adtt'p sparsity from tbone} and~\eqref{eq:adtt'p sparsity from tb},
        \[
        k(\mAD_{t, t', p}) = \Omega(2^ptt' + t^2t').
        \]
        By Equation~\eqref{eq:mAD expansion},
        \begin{align*}
            k(\mAD_{t,t',p}) & \leq k(\ind_\bone)k(\mAND_{t', p}) + \sum_{\bone \neq b \in \pmone^{\log t}} k(\ind_b)k(\AND_{\log t'})\\
            & \leq t \cdot k(\mAND_{t',p}) + (t-1)t \cdot t' \tag*{by Observations~\ref{obs:indexpansion} and~\ref{obs:properties of AND, Bent and Addressing}}\\
            & = O(2^p t t' + t^2 t'). \tag*{since $\mAND_{t', p}$ is a function on $(\log t' + p)$ variables}
        \end{align*}
        Thus,
        \[
        k(\mAD_{t,t',p}) = \Theta(2^ptt' + t^2t').
        \]

    \item Max-rank-entropy:    
        Recall from the argument for rank that $\calB = \bigcup_{b \in \pmone^{\log t}}\calB_b \cup \calB_x \subseteq \supp(\mAD_{t,t',p})$ generate all the monomials of $\mAD_{t,t',p}$.
        Moreover the absolute values of the coefficients of these monomials take values in the set $\cbra{\frac{2}{tt'}, \bra{1 - \frac{2}{2^p}}\frac{2}{tt'}}$.
        Since $p \geq 2$, $ 1 \geq 1- \frac{2}{2^p} \geq \frac12$.
        Therefore 
        \begin{equation}
        \label{eq:maxrankent upper bound madtt'p}
        k''(\mAD_{tt'p}) = O(tt').
        \end{equation}

        Recall that no monomial arising from the terms $T_b$ for $\bone \neq b \in \pmone^{\log t}$ (see Equation~\eqref{eq:Fourier ADtt'p Tb}) contain variables from $y_\bone$.
         Thus, monomials which contain the variable $y_{\bone,1}$ only appear in Equation~\eqref{eq:Fourier ADtt'p T1}. Moreover, by Observation~\ref{obs:indexpansion} we conclude that the absolute value of coefficient of any such monomial takes values in the set $\cbra{\frac{2}{tt'}\bra{1 - \frac{2}{2^p}}, \frac{4}{2^p t t'}}$.        
         
        Recall from Definition~\ref{defi:max entropy, max rank entropy} that $k''(f) = \argmin_\theta \{\dim(\cS_{\theta}) = r(f)\}$, where $\cS_{\theta} = \cbra{S: |\wh{f}(S)| \geq \frac{1}{\theta}}$.
        Therefore if $\theta < \frac{2}{tt'}$, then $\cS_\theta$ does not include any monomial containing $y_{\bone,1}$. 
        Therefore 
        \[
        k''(\mAD_{t,t',p}) \geq \bra{1- \frac{2}{2^p}}^{-1}\frac{tt'}{2} \geq \frac{tt'}{2}.
        \]
        Hence by Equation~\eqref{eq:maxrankent upper bound madtt'p}, 
        \[
        k''(\mAD_{t,t',p}) = \Theta(tt').
        \]       
    
    \item Weight:
        By Equation~\eqref{eq:mAD expansion},
        \begin{align*}
            \wh{\mAD_{t,t',p}}(\emptyset) &= \wh{\ind_{\bone}}(\emptyset) \wh{\mAND_{t',p}}(\emptyset) + \sum_{\bone \neq b \in \pmone^{\log t}} \wh{\ind_b}(\emptyset) \wh{\mAND_{\log t'}}(\emptyset) \tag*{by Observation~\ref{obs:indexpansion} and Claim~\ref{claim:mAND Fourier expansion}}\\
            &= \frac{1}{t} 
            \bra{1-\frac{2}{t'}} + (t-1)\frac{1}{t} \bra{1 - \frac{2}{t'}}. \\
            &= 1 - \frac{2}{t'}
        \end{align*}        
            Thus by Observation~\ref{obs:weight, empty Fourier}, 
            \[
            \delta(\mAD_{t,t',p}) = \frac{1 - \wh{\mAD_{t,t',p}(\emptyset)}}{2} = \frac{1}{t'}.
            \]
\end{itemize}
\end{proof}


\subsection{Setting parameters in our constructed functions}
\label{subsec:tightness} 

In this section we prove Theorems~\ref{thm:delta upper bound in terms of rk and also k'} and~\ref{thm:delta upper bound in terms of rk and also k''}. Recall that these theorems require us to exhibit functions which achieve certain bounds. Claims~\ref{claim:setting parameters for tightstraightline} and~\ref{claim:setting parameters for tight curve} correspond to the bounds in Theorem~\ref{thm:delta upper bound in terms of rk and also k'}, and describe the required functions. Claims~\ref{claim: setting parameters for tight curve for k''} and~\ref{claim:setting parameters for tightstraightline for k''} correspond to the bounds in Theorem~\ref{thm:delta upper bound in terms of rk and also k''}, and describe the required functions.

    \begin{claim}
    \label{claim:setting parameters for tightstraightline}
    For all $\rho, \kappa, \kappa' \in \N$ such that $\kappa$ is sufficiently large, for all $\epsilon > 0$ such that $\log \kappa \leq \rho \leq \kappa^{\frac12 - \epsilon}$  and $\frac{\kappa \log \kappa}{\rho} \leq \kappa' \leq \kappa$, for $t=\frac{2\rho}{\log \kappa}$, $t'= \frac{\kappa \log^2 \kappa}{\rho^2}$ and $a= \frac{2\kappa'\log \kappa}{\rho}$,
    \begin{itemize}
        \item $\Omega(\epsilon \rho) = r(\AD_{t,t',a}) = O(\rho)$.
        \item $k(\AD_{t,t',a})  = \Theta(\kappa)$.
        \item $ k'(\AD_{t,t',a}) = \Theta(\kappa')$.
        \item $\delta(\AD_{t,t',a}) = \Theta\paren{\frac{1}{\kappa}\paren{\frac{\rho}{\log \kappa}}^2}$.
    \end{itemize}
    \end{claim}

    We prove Claim~\ref{claim:setting parameters for tightstraightline} in Section~\ref{subsubsec:functions-on-horizontal-k'}.
   

    \begin{claim}
    \label{claim:setting parameters for tight curve}
        For all $\rho, \kappa, \kappa' \in \N$ such that $\kappa$ is sufficiently large, for all constants $\epsilon > 0$, such that $\kappa^{1/2} \leq \kappa' \leq (\kappa\log \kappa)/\rho$ and $\log \kappa \leq \rho \leq \kappa^{\frac{1}{2} - \epsilon}$ 
        for $t= \frac{2\rho}{\log \kappa}$, $t'= \frac{4\kappa'^2}{\kappa}$ and $\ell = 2\bra{\frac{\kappa\log \kappa}{\kappa'\rho}}^2$,
         \begin{itemize}
            \item $\Omega(\epsilon \rho) = r(\AAB) = O(\rho)$.
            \item $k(\AAB)= \Theta(\kappa)$.
            \item $k'(\AAB) = \Theta(\kappa')$.
            \item $\delta(f) = O\bra{\frac{\kappa}{\kappa'^2}}$.
        \end{itemize}
    \end{claim}
    We prove Claim~\ref{claim:setting parameters for tight curve} in Section~\ref{subsubsec:funcs-on-curve-k'}.
    \begin{claim}
    \label{claim: setting parameters for tight curve for k''}
         For all $\rho, \kappa, \kappa'' \in \N$ such that $\kappa$ is sufficiently large, for all constants $ \epsilon > 0$ such that $\log \kappa \leq \rho \leq \kappa^{1/2 - \epsilon}$, $e\rho \leq \kappa'' \leq \frac{\kappa \log \kappa}{\rho}$, for $t = \frac{2\rho}{\log (\kappa''/\rho)}, t' = \frac{\kappa''}{\rho}\log \bra{\kappa''/\rho}$, $p = \log \bra{\frac{4\kappa}{\kappa''}}$,
         \begin{itemize}
         \item $r(\mAD_{t, t', p}) = \Theta(\rho)$.
         \item $\Omega(\kappa) = k(\mAD_{t, t', p}) = O(\kappa/\epsilon)$.
         \item $k''(\mAD_{t, t', p}) = \Theta(\kappa'')$.
         \item $\delta(\mAD_{t, t', p}) = \frac{\rho}{\kappa'' \log (\kappa''/\rho)}$.
         \end{itemize}
    \end{claim}
    
    We prove Claim~\ref{claim: setting parameters for tight curve for k''} in Section~\ref{sec:upper bound on wt in terms of k'' and logk''/r}.      

\begin{claim}
    \label{claim:setting parameters for tightstraightline for k''}
    For all $\rho, \kappa, \kappa' \in \N$ such that $\kappa$ is sufficiently large, for all $\epsilon > 0$ such that $\log \kappa \leq \rho \leq \kappa^{\frac12 - \epsilon}$  and $\frac{\kappa \log \kappa}{\rho} \leq \kappa'' \leq \kappa$, there exists a constant $c \geq 1$ such that the following holds for $t=\frac{2\rho}{\log \kappa}$, $t'= \frac{c\kappa \log^2 \kappa}{\rho^2}$ and $a= \frac{2c\kappa''\log \kappa}{\rho}$.
    \begin{itemize}
        \item $\Omega(\epsilon \rho) = r(\AD_{t,t',a}) = O(\rho)$.
        \item $k(\AD_{t,t',a})  = \Theta(\kappa)$.
        \item $ k''(\AD_{t,t',a}) = \Theta(\kappa'')$.
        \item $\delta(\AD_{t,t',a}) = \Theta\paren{\frac{1}{\kappa}\paren{\frac{\rho}{\log \kappa}}^2}$.
    \end{itemize}
    \end{claim}

\begin{proof}
It follows by Claim~\ref{claim:setting parameters for tightstraightline} and the fact that $k'(\AD_{t, t',a}) = k''(\AD_{t, t', a})$ (Claim~\ref{claim:properties of ADtt'a}).
\end{proof}

\subsubsection{Proof of Claim~\ref{claim:setting parameters for tightstraightline}}
\label{subsubsec:functions-on-horizontal-k'}
    In this section we prove Claim~\ref{claim:setting parameters for tightstraightline}, which gives us Fourier-analytic properties of $\AD_{t, t', a}$ for particular settings of $t, t', a$.

    \begin{proof}[Proof of Claim~\ref{claim:setting parameters for tightstraightline}]
    
    Given any $\rho, \kappa, \kappa'$ as in the assumptions of the claim, recall that we fix the following values. 
    \begin{align}
        t &= \frac{2\rho}{\log \kappa}, \label{eq:thugeand}\\
        t' &= \kappa\left(\frac{\log \kappa}{\rho}\right)^2, \label{eq:t'hugeand}\\
        a &= \frac{2\kappa'\log \kappa}{\rho}. \label{eq:ahugeand}
    \end{align}
    Since $\rho \geq \log \kappa$, 
    \[
        t = \frac{2\rho}{\log \kappa} \geq 2.
    \]
    Since $\rho < \sqrt{\kappa}$ by assumption, 
    \[
    t' = \kappa\left(\frac{\log \kappa}{\rho}\right)^2 \geq \log^2 \kappa \geq 2,
    \]
    for large enough $\kappa$.
    Since $\kappa' \geq \frac{\kappa\log \kappa}{\rho}$,
    \begin{equation}\label{eq:a2t'}
         a = \frac{2\kappa' \log \kappa}{\rho} \geq 2\kappa\bra{\frac{\log \kappa}{\rho}}^2 = 2t'.
    \end{equation}
    
    Hence the assumptions in Claim~\ref{claim:properties of ADtt'a} are satisfied with these values of $t, t', a$. By Equations~\eqref{eq:thugeand},~\eqref{eq:t'hugeand} and~\eqref{eq:ahugeand}, we obtain the following bound which is of use to us later.

    \begin{equation}
    \label{eq:at/t'}
    \frac{at}{t'} = \frac{2\kappa'\log \kappa}{\rho} \cdot \frac{2\rho}{\log \kappa} \cdot \frac{\rho^2}{\kappa \log^2\kappa} = \frac{4\kappa'}{\kappa} \bra{\frac{\rho}{\log \kappa}}^2 \leq \bra{\frac{2\rho}{\log \kappa}}^2 \leq t^2.
    \end{equation}

    We have the following properties of $\AD_{t, t', a}$.
    \begin{itemize}
        \item Rank:
        \begin{align*}
        r(\AD_{t,t',a}) & = (t -1)\log{t'} + \log a + \log t\tag*{by Claim~\ref{claim:properties of ADtt'a}}\\
                    &= t \log t' + \log \bra{\frac{at}{t'}} \leq 3 t\log t' \tag*{by Equation~\eqref{eq:at/t'}}\\
                    &= \frac{6\rho}{\log \kappa} \bra{\log (\kappa/\rho^2) + 2\log\log \kappa}  \tag*{by Equations~\eqref{eq:thugeand} and~\eqref{eq:t'hugeand}}\\
                    &\leq \frac{6\rho}{\log \kappa} \bra{\log \kappa + 2\log \kappa}\\
                    &=O(\rho).
    \end{align*}
    For our setting of parameters (Equations~\eqref{eq:thugeand}, \eqref{eq:t'hugeand} and \eqref{eq:ahugeand}), $at = 4\kappa'$. Since $\rho \geq \log \kappa$ and $\kappa' \geq \frac{\kappa \log \kappa}{\rho}$,
    \begin{equation}\label{eq:at'}
        \frac{at}{t'} = \frac{4\kappa'}{\kappa}\bra{\frac{\rho}{\log{\kappa}}}^2 \geq \frac{4\rho}{\log \kappa} > 1.
    \end{equation}
    Thus,
    \begin{align*}
        r(\AD_{t, t', a}) & = (t -1) \log{t'} + \log a + \log t \tag*{by Claim~\ref{claim:properties of ADtt'a}}\\
                &= t \log t' + \log {\bra{\frac{at}{t'}}} \geq t \log t' \tag*{by Equation~\eqref{eq:at'}}\\
                &= \frac{2\rho}{\log \kappa} \bra{\log (\kappa/\rho^2) + 2\log\log \kappa}  \tag*{by Equations~\eqref{eq:thugeand} and~\eqref{eq:t'hugeand}}\\
                &\geq \frac{2\rho}{\log \kappa} \bra{\log (\kappa/\rho^2)} \geq \frac{2\rho}{\log \kappa} \bra{\log \kappa^{2\epsilon}} \tag*{since $\rho \leq \kappa^{\frac12 - \epsilon}$}\\
                &= \frac{2\rho}{\log \kappa}(2\epsilon \log \kappa)\\
                &= \Omega(\epsilon \rho).
    \end{align*}
        \item Sparsity: 
            By our choice of parameters (Equations~\eqref{eq:thugeand},~\eqref{eq:t'hugeand} and~\eqref{eq:ahugeand}), we have 
            \begin{equation}
            \label{eq:tat^2t'}
            ta = 4\kappa' \leq 4\kappa = t^2t'.
            \end{equation}
            Thus,
            \begin{align*}
                k(f_{\rho,\kappa,\kappa'}) &= (t-1)(t'-1)t + ta \tag*{by Claim~\ref{claim:properties of ADtt'a}}\\
                &= \Theta(t^2t') \tag*{by Equation~\eqref{eq:tat^2t'} and since $t, t' \geq 2$}\\
                &= \Theta\bra{\bra{\frac{2\rho}{\log \kappa}}^2 \kappa \bra{\frac{\log \kappa }{\rho}}^2} \tag*{by Equations~\eqref{eq:thugeand} and~\eqref{eq:t'hugeand}}\\
                &= \Theta(\kappa).
            \end{align*}
        \item Max-supp-entropy:
        Since $ta = 4\kappa'$ from Equation~\eqref{eq:tat^2t'}, we have by Claim~\ref{claim:properties of ADtt'a} that
        \begin{align*}
            k'(f_{\rho,\kappa,\kappa'}) &= \frac{ta}{2} = \Theta(\kappa').
        \end{align*}
        \item Weight: 
        \begin{align*}
            \delta(f_{\rho,\kappa,\kappa'}) 
            &= \frac{1}{t'} + \frac{1}{at} - \frac{1}{tt'}\tag*{by Claim~\ref{claim:properties of ADtt'a}}\\
            &= \frac{1}{t'} + \frac{1}{t}\bra{\frac{1}{a} - \frac{1}{t'}}\\
            &\in \left[\frac{1}{t'} - \frac{1}{tt'}, \frac{1}{t'} - \frac{1}{2tt'}\right]  \tag*{since $2t' \leq a$ by Equation~\eqref{eq:a2t'}}\\
            &= \Theta \bra{\frac{1}{t'}} \tag*{since $t \geq 2$}\\
            &= \Theta\bra{\frac{1}{\kappa}\bra{\frac{\rho}{\log \kappa}}^2} \tag*{by Equation~\eqref{eq:t'hugeand}}.
        \end{align*}
    \end{itemize}
\end{proof}    


\subsubsection{Proof of Claim~\ref{claim:setting parameters for tight curve}}
\label{subsubsec:funcs-on-curve-k'}
In this section we prove Claim~\ref{claim:setting parameters for tight curve}, which gives us Fourier-analytic properties of $\AAB$ for particular settings of $t, t', \ell$.
    \begin{proof}[Proof of Claim~\ref{claim:setting parameters for tight curve}]
    Given any $\rho, \kappa, \kappa'$ as in the assumptions of the claim, recall that we fix the following values.
     \begin{align}
         t & = \frac{2\rho}{\log \kappa}, \label{eq:tAAB}\\
         t' & = \frac{4\kappa'^2}{\kappa}, \label{eq:t'AAB}\\
         \ell & = 2\bra{\frac{\kappa \log \kappa}{\kappa'\rho}}^2. \label{eq:ellAAB}
     \end{align} 

     Since $\rho \geq \log \kappa$, we have $t = \frac{2 \rho}{\log \kappa} \geq 2$.
     Since $\kappa' \geq \sqrt{\kappa}$, we have
     $t' = \frac{4\kappa'^2}{\kappa} \geq 4$.
     Finally since $\kappa' \leq (\kappa\log \kappa)/\rho$,
     $\ell = 2\bra{\frac{\kappa \log \kappa}{\kappa'\rho}}^2 \geq 2$.
     Hence the assumptions in Claim~\ref{claim:properties of AAB} are satisfied with these values of $t, t'$ and $\ell$.
     
    By Equations~\eqref{eq:t'AAB} and~\eqref{eq:ellAAB}, we obtain the following bound which is of use to us later.
    \begin{equation}\label{eq:lt'}
        \ell t' = 2\bra{\frac{\kappa \log \kappa}{\kappa'\rho}}^2 \cdot \frac{4\kappa'^2}{\kappa} = 8 \kappa\bra{\frac{\log \kappa}{\rho}}^2.
    \end{equation}
    We have the following properties of $\AAB$. \begin{itemize}
        \item Rank:
            \begin{align*}
                r(\AAB)  
                &= t (\log t' + \log \ell) + \log t 
                \tag*{by Claim~\ref{claim:properties of AAB}}\\
                &= t \log (\ell t') + \log t \leq 2t \log (\ell t') \tag*{since $\ell \geq 2$ and $t' \geq 4$}\\
                &= \frac{4\rho}{\log \kappa}\bra{ \log 8 + \log \bra{\kappa/\rho^2} + 2\log\log \kappa} \tag*{by Equation~\eqref{eq:lt'}}\\
                &\leq \frac{4\rho}{\log \kappa}\bra{ \log \kappa + \log \kappa + 2\log \kappa} \tag*{since $\kappa$ is sufficiently large}\\
                &= O(\rho).
            \end{align*}
            For the lower bound, we have
            \begin{align*}
                r(\AAB) &= t (\log  t' + \log \ell) + \log t \tag*{by Claim~\ref{claim:properties of AAB}}\\
                &\geq t \log(\ell t') \tag*{since $t \geq 1$}\\
                &= \frac{2 \rho}{\log \kappa} \cdot \bra{ \log 8 + \log \bra{\kappa/\rho^2} + 2\log\log \kappa} \tag*{by Equations~\eqref{eq:tAAB} and~\eqref{eq:lt'}}\\
                &\geq \frac{2 \rho}{\log \kappa} \cdot \bra{\log \bra{\kappa/\rho^2}} \geq \frac{2 \rho}{\log \kappa} \cdot \bra{\log \kappa^{2\epsilon}} \tag*{since $\rho \leq \kappa^{\frac12 - \epsilon}$}\\
                &= \frac{2 \rho}{\log \kappa}\bra{2\epsilon\log \kappa}\\
                & = \Omega(\epsilon \rho).
            \end{align*}
        
        \item Sparsity: 
            \begin{align*}
                k(\AAB) 
                &= 1 + \frac12 t^2(\ell+1)t' 
                \tag*{by Claim~\ref{claim:properties of AAB}}\\
                &= \Theta(t^2 \ell t')\\
                &= \Theta\bra{\bra{\frac{2\rho}{\log \kappa}}^2 \cdot 8\kappa \bra{\frac{\log \kappa}{\rho}}^2} \tag*{by Equations~\eqref{eq:tAAB} and~\eqref{eq:lt'}}\\
                & = \Theta(\kappa).
            \end{align*}
        \item Max-supp-entropy:
            \begin{align*}
                k'(\AAB) 
                &= \frac{tt' \sqrt{\ell}}{2} 
                \tag*{by Claim~\ref{claim:properties of AAB}}\\
                &= \frac12 \cdot \frac{2\rho}{\log \kappa} \cdot \frac{(2\kappa')^2}{\kappa} \cdot \frac{\sqrt{2}\kappa \log \kappa}{\kappa' \rho} \tag*{by Equations~\eqref{eq:tAAB}, \eqref{eq:t'AAB} and \eqref{eq:ellAAB}}\\
                &= \Theta(\kappa').
            \end{align*}
        \item Weight: 
            \begin{align*}
                \delta(\AAB)  
                &= \frac{1}{t'} 
                \tag*{by Claim~\ref{claim:properties of AAB}}\\
                &= \Theta\bra{\frac{\kappa}{\kappa'^2}}\tag*{by Equation~\eqref{eq:t'AAB}}.
            \end{align*}
    \end{itemize}  
   \end{proof}

\subsubsection{Proof of Claim~\ref{claim: setting parameters for tight curve for k''}}
\label{sec:upper bound on wt in terms of k'' and logk''/r}
    In this section we prove Claim~\ref{claim: setting parameters for tight curve for k''}, which gives us Fourier-analytic properties of $\mAD_{t, t', p}$ for particular settings of $t, t', p$.
    \begin{proof}[Proof of Claim~\ref{claim: setting parameters for tight curve for k''}]
            
    Given $\rho, \kappa$ and $\kappa''$ choose:
    \begin{align}
            p &= \log \left(\frac{4\kappa}{\kappa''}\right)
            \label{eq:p mAD}\\
            t' &= \frac{\kappa''}{\rho} \log \left( \frac{\kappa''}{\rho} \right)
            \label{eq:t' mAD}\\
            t &= \frac{2\rho}{\log(\kappa''/\rho)} \label{eq:t mAD}
    \end{align}
    
    Since $\rho \geq \log \kappa \geq \log (\kappa''/\rho)$, we have 
    $t = \frac{2 \rho}{\log (\kappa''/\rho)}  \geq 2$. Since $\kappa'' \geq 2\rho$, we have $t' = \frac{\kappa''}{\rho} \log \bra{\frac{\kappa''}{\rho}} \geq \frac{\kappa''}{\rho} \geq 2$. Note that $p \geq 2$ since $\log \left(\frac{4\kappa}{\kappa''}\right) \geq \log 4 = 2$.

Finally we show that $p \leq (t/2) \log t' \leq (t-1) \log t'$:
\begin{align*}
    (t-1) \log t' &\geq 
    (t/2) \log t'\\
    &= \frac{\rho}{\log (\kappa''/\rho)} \log\left(\frac{\kappa''}{\rho} \log\left( \frac{\kappa''}{\rho}\right) \right) \\
    &= \rho + \rho \cdot \frac{\log \log (\kappa''/\rho)}{\log (\kappa''/\rho)} \\
    &\geq \rho \tag*{since $\kappa'' \geq 2\rho$}\\
    & \geq \log {\kappa} \geq \log \bra{4\kappa/\kappa''} \tag*{since $\kappa'' \geq e \log \kappa \geq 4$ as $\kappa$ is sufficiently large}\\
    & = p.
\end{align*}

Hence the assumptions in Claim~\ref{claim:properties of mAD} are satisfied with these values of $t, t'$ and $p$.

We first state and prove some auxiliary claims which we require. The derivative of $\frac{\log (\kappa''/\rho)}{\kappa''}$ with respect to $\kappa''$ equals
\begin{align*}
    \frac{1 - \ln {(\kappa''/\rho)}}{\ln 2 \cdot (\kappa'')^2}.
\end{align*}
    This value is negative since $\kappa'' > e\cdot \rho$. Thus, $\frac{\log (\kappa''/\rho)}{\kappa''}$ is a decreasing function in $\kappa''$ for $\kappa'' > e\cdot \rho$. Consider the expression
    \begin{align*}
        \frac{2^p}{t} 
        &= \frac{4\kappa}{\kappa''} \frac{\log (\kappa''/\rho)}{2\rho}\\
        &\geq \frac{2\kappa \cdot \log{\bra{\frac{\kappa\log \kappa}{\rho^2} }}}{\kappa \log \kappa} \tag*{Since $\frac{\log (\kappa''/\rho)}{\kappa''}$ is a monotone decreasing function in $\kappa''$ and $\kappa'' \leq \frac{\kappa \log \kappa}{\rho}$.} \\
        &= 2\frac{\bra{\log {\kappa/\rho^2}} + \log{\log \kappa}}{\log \kappa} \\
        &\geq 2\frac{\log \kappa^{2\epsilon} + \log \log \kappa}{\log \kappa}
        \tag*{since $\rho\leq \kappa^{1/2 -\epsilon}$}
        \\
        &\geq 4 \epsilon.
    \end{align*}
Therefore 
\begin{equation}
\label{eq:2^pgeq4epst}
    2^p\geq 4 \epsilon t.
\end{equation}

Next, observe that
\begin{equation}
\label{eq: tt'mADtt'p}
    tt' = \frac{2\rho}{\log (\kappa''/\rho)} \bra{\frac{\kappa'' \log (\kappa''/\rho)}{\rho}}
        = 2\kappa''.
\end{equation}
We have the following properties of $\mAD_{t, t', p}$.
    \begin{itemize}
        \item Rank: 
            \begin{align*}  
                r(\mAD_{t,t',p}) &= \Theta\left(t \log t'\right) \tag*{by Claim~\ref{claim:properties of mAD}}\\
                &= \Theta\bra{ \frac{\rho}{\log(\kappa''/\rho)} \bra{\log \frac{\kappa''}{\rho} + \log \log \bra{ \frac{\kappa''}{\rho}} }
                } \tag*{by Equations~\eqref{eq:t mAD} and~\eqref{eq:t' mAD}}\\
                &= \Theta\bra{
                    \frac{\rho}{\log(\kappa''/\rho)} \bra{\log \frac{\kappa''}{\rho} }
                }\\
                &= \Theta(\rho).
            \end{align*}
        \item Sparsity:
        For the upper bound, we have
        \begin{align*}
            k(\mAD_{t,t',p}) &= O(2^ptt' + t^2t') \tag*{by Claim~\ref{claim:properties of mAD}}\\  
            &= O\bra{2^ptt' + \frac{2^ptt'}{4\epsilon}} \tag*{by Equation~\eqref{eq:2^pgeq4epst}}\\
            &=O\bra{\frac{1}{\epsilon} \cdot 2^ptt'} \tag*{since $0 < \epsilon < 1/2$}\\
            &= O \bra{\frac{1}{\epsilon} \cdot \frac{\kappa}{\kappa''} \cdot \kappa''} \tag*{by Equations~\eqref{eq: tt'mADtt'p} and~\eqref{eq:p mAD}}\\
            &= O\bra{\frac{\kappa}{\epsilon}}.
        \end{align*}
        For the lower bound, we have
        \begin{align*}
           k(\mAD_{t,t',p}) &= \Omega(2^p t t') \tag*{by Claim~\ref{claim:properties of mAD}}\\
            &= \Omega{ \left(\frac{\kappa}{\kappa''} \cdot \kappa''\right)} \tag*{by Equations~\eqref{eq: tt'mADtt'p} and~\eqref{eq:p mAD}}\\
            &= \Omega(\kappa).
        \end{align*}
        
        \item Max-rank-entropy:
        \begin{align*}
            k''(\mAD_{t,t',p}) 
            &= \Theta(tt') \tag*{by Claim~\ref{claim:properties of mAD}}\\
            &= \Theta(\kappa'') \tag*{by Equation~\eqref{eq: tt'mADtt'p}}.
        \end{align*}
        \item Weight:
        \begin{align*}
            \delta(\mAD_{t,t',p}) &= \frac{1}{t'} \tag*{by Claim~\ref{claim:properties of mAD}}\\
                    &= \frac{\rho}{\kappa'' \log (\kappa''/\rho). \tag*{by Equation~\eqref{eq:t' mAD}}}
        \end{align*}
    \end{itemize}
    \end{proof}


\section{Dominating Chang's lemma for all thresholds}
\label{sec: beating chang all thresholds}
In this section, we show that there exists a function such that for any choice of threshold, the lower bound on weight of that function that we obtain from Theorems~\ref{thm:delta lower bound in terms of rk and also k'} and~\ref{thm:delta lower bound in terms of rk and also k''} are stronger than the lower bound obtained from Chang's lemma (Lemma~\ref{lem:chang}).
\begin{claim}[Beating Chang's lemma for all thresholds for $\AD_{t, t}$]
\label{claim: beating changs lemma for all thresholds for ADtt'}
    Consider any integer $t > 4$ and define the function $f = \AD_{t, t} : \pmone^{\log t} \times \pmone^{t \log t} \to \pmone$ as in Definition~\ref{defi:ADtt'}. Then,
    \begin{itemize}
        \item $\delta(f) = \frac{1}{t}$.
        
        \item For all real $x > 0$ for which $\dim(\cS_x) > 1$, we have
        \begin{align*}
            \frac{\sqrt{\dim(\cS_x)}}{x\sqrt{\log (x^2/\dim(\cS_x))}} = O\left(\frac{1}{t^{3/2}}\right).
        \end{align*}
        \item 
        \[
        \frac{1}{k(f)}\bra{\frac{r(f)}{\log \sparsity(f)}}^2 = \Omega\bra{\frac{1}{t}}, \quad\frac{\sparsity(f)}{\seespectrum(f)^2} = \Omega\bra{\frac{1}{t}} \quad \text{and}\quad \frac{r(f)}{\seerank(f)\log \sparsity(f)} = \Omega\bra{\frac{1}{t}}.
        \]
    \end{itemize}
\end{claim}

In particular, Claim~\ref{claim: beating changs lemma for all thresholds for ADtt'} shows that our bounds can be strictly stronger than those given by Chang's lemma, in the following sense.
\begin{itemize}
\item All the lower bounds on $\delta(f)$ from Theorems~\ref{thm:delta lower bound in terms of rk and also k'} and~\ref{thm:delta lower bound in terms of rk and also k''} are tight, as witnessed by $f = \AD_{t, t}$. 

\item No matter what threshold $x$ is chosen in Lemma~\ref{lem:chang}, the best possible lower bound on $\delta(\AD_{t,t})$ that we get can get from Lemma~\ref{lem:chang} is $\Omega\bra{\frac{1}{t^{3/2}}}$, which is polynomially smaller than $1/t$, the actual weight of $\AD_{t,t}$.
\end{itemize}

\begin{proof}[Proof of Claim~\ref{claim: beating changs lemma for all thresholds for ADtt'}]

From Claim~\ref{claim:properties of ADtt'}, we have
\begin{align*}
    \delta(f) = \frac{1}{t}.
\end{align*}
Thus from Observation~\ref{obs:weight, empty Fourier},
\begin{align}
\label{eq:fhatempty adtt'}
    \wh{f}(\emptyset) = 1 - \frac{2}{t}.
\end{align}
First, we show that except for the Fourier coefficient of the empty set, all other Fourier coefficients of $f$ have magnitude equal to $\frac{2}{t^2}$.

Towards a contradiction assume that there exists $T \subseteq [\log t]\cup [t \log t]$ such that $\abs{\wh{f}(T)} > \frac{2}{t^2}$. We have,
\begin{align}
    1 
    &= \sum_{S \subseteq [\log t]\cup [t \log t]} \wh{f}^2(S)
    \tag*{from Theorem~\ref{thm:Parseval}}
    \nonumber
    \\
    &= \left(1 - \frac{2}{t}\right)^2 +  \sum_{S \subseteq [\log t]\cup [t \log t], S \neq \emptyset} \wh{f}^2(S)
    \tag*{by Equation~\eqref{eq:fhatempty adtt'}} \nonumber
    \\
    &= \left(1 - \frac{2}{t}\right)^2 +  \wh{f}^2(T) +
    \sum_{S \subseteq [\log t]\cup [t \log t], S \neq \emptyset, S \neq T} \wh{f}^2(S)
    \nonumber
    \\
    &> \left(1 - \frac{2}{t}\right)^2 +  \frac{4}{t^4} +
    \sum_{S \subseteq [\log t]\cup [t \log t], S \neq \emptyset, S \neq T} \wh{f}^2(S)
    \label{eq: beting changs eqn1}
\end{align}
From Claim~\ref{claim:properties of ADtt'}, $k(f) = 1 + t^2(t - 1) = t^3 - t^2 + 1$ and $k'(f) = t^2/2$. Using these, along with the definition of $k'(f)$ (Definition~\ref{defi:max entropy, max rank entropy}), in Equation~\eqref{eq: beting changs eqn1},
\begin{align*}
    1 &> \left(1 - \frac{2}{t}\right)^2 +  \frac{4}{t^4} + \frac{4(k(f) - 2)}{t^4}\\
    &= \left(1 - \frac{2}{t}\right)^2 +  \frac{4}{t^4} + \frac{4(t^3 - t^2 + 1 - 2)}{t^4}
    \\
    &= \left(1 - \frac{2}{t}\right)^2 +  \frac{4}{t^4} + \frac{4(t^3 - t^2 -1)}{t^4}
    \\
    &= \left(1 - \frac{2}{t}\right)^2 +  \frac{4}{t^4} +
    \left( \frac{4}{t} - \frac{4}{t^2} - \frac{4}{t^4}  \right)
    \\
    &= \left(1 + \frac{4}{t^2} - \frac{4}{t} \right) +  \frac{4}{t^4} +
    \left( \frac{4}{t} - \frac{4}{t^2} - \frac{4}{t^4}  \right)
    \\
    &= 1.
\end{align*}
Thus,
\begin{equation}
\label{eq:all fcoeffs except emptyset equal adtt'}
\abs{\wh{f}(S)} = \frac{2}{t^2} \quad \text{for all non-empty}~S \subseteq ([\log t]\cup [t \log t]).
\end{equation}
We now prove the second part of the claim. If $x < \frac{t^2}{2}$ then $\cS_x = \cbra{\emptyset}$ and has dimension $0$. On the other hand, for any $x \geq \frac{t^2}{2}$, we have $\cS_x = \supp(f)$ by Equation~\eqref{eq:all fcoeffs except emptyset equal adtt'}, and hence $\dim(\cS_x) = r(f) = (t + 1)(\log t)$ by Claim~\ref{claim:properties of ADtt'}.
Hence, in this case,
\[
\frac{\sqrt{\dim(\cS_x)}}{x\sqrt{\log (x^2/\dim(\cS_x))}} = O\left(\frac{\sqrt{t \log t}}{t^2\sqrt{\log t}}\right) = O\left(\frac{1}{t^{3/2}}\right),
\]

On the other hand, by Claim~\ref{claim:properties of ADtt'}, $r(f) = \Theta(t \log t)$, $k(f) = \Theta(t^3)$, $k'(f) = k''(f) = \Theta(t^2)$.
The third part of the claim follows.

\end{proof}

\section{Conclusions}
\label{sec: conclusions}
    In this paper, for Boolean functions $f$, we study the relationship between weight and other Fourier-analytic measures namely rank, sparsity, max-supp-entropy and max-rank-entropy. For a threshold $t > 0$, Chang's lemma gives a lower bound on the weight of a Boolean function $f$ in terms of  $\dim\bra{\cbra{S \subseteq [n]:|\wh{f}(S)|\geq \frac{1}{t}}}$.
    We consider three natural thresholds $t$ in Chang's lemma, namely $k(f)$, $k'(f)$ and $k''(f)$,  yielding three lower bounds on weight in terms of these measures. We prove new lower bounds on weight in Theorems~\ref{thm:delta lower bound in terms of rk and also k'} and~\ref{thm:delta lower bound in terms of rk and also k''}, and our bounds dominate all the above-mentioned bounds from Chang's lemma for a wide range of parameters.
    
    When $\log k(f) = \Theta(r(f))$, the function $f = \AND$ already shows that all the above lower bounds are tight. 
    To consider all other feasible relationships between $k(f)$ and $r(f)$, we divide our investigation of these lower bounds into two different parts.
    In the first part, we vary over all feasible settings of $r(f)$, $k(f)$ and $k'(f)$, and construct functions that witness tightness of our lower bounds in Theorem~\ref{thm:delta lower bound in terms of rk and also k'} for nearly all such feasible settings (Theorem~\ref{thm:delta upper bound in terms of rk and also k'}).
    In the second part, we vary over all feasible settings of $r(f)$, $k(f)$ and $k''(f)$, and construct functions that witness near-tightness of our lower bounds in Theorem~\ref{thm:delta lower bound in terms of rk and also k''} for nearly all such feasible settings (Theorem~\ref{thm:delta upper bound in terms of rk and also k''}).
    These functions are constructed by carefully composing the Addressing function with suitable inner functions.
    We show a composition lemma (Lemma~\ref{lem:properties of composition of addressing and g}), which relates the properties of the composed function with those of the inner functions; this allows us to come up with functions that match our lower bounds. 
    
    We also construct functions for which our lower bounds are asymptotically stronger than the lower bounds obtained from Chang's lemma for all choices of threshold (Claim~\ref{claim: beating changs lemma for all thresholds for ADtt'}). The functions that we construct in this work might be of independent interest.
    
    \paragraph{Open Problems.} 
    
    Claim~\ref{claim:setting parameters for tightstraightline} shows tightness of our upper bound on rank in terms of sparsity and weight (Theorem~\ref{thm:delta lower bound in terms of rk only}). Since our proof of Theorem~\ref{thm:delta lower bound in terms of rk only} is a generalization of the proof of the upper bound $r(f) = O(\sqrt{k(f)} \log k(f))$ due to Sanyal~\cite{San19}, it sheds light on the presence of the $\log k$ factor in Sanyal's upper bound. This still leaves the following question open: do there exist Boolean functions $f$ for which $r(f) = \omega(\sqrt{k(f)})$?

    There are some ranges of parameters where we were not able to construct functions with upper bounds matching our lower bounds from Theorem~\ref{thm:delta lower bound in terms of rk and also k''}. It will be interesting to see if our techniques can be extended to cover these ranges as well.    
    
    All thresholds $t$ considered for Chang's lemma in this work satisfy $\dim(\{S \subseteq [n]:|\wh{f}(S)|\geq \frac{1}{t}\}) = r(f)$.
    It is an interesting problem to obtain Chang's-lemma-type bounds for thresholds for which this dimension is strictly less than $r(f)$. 

%

{\bf Acknowledgements: } R.M.~thanks DST (India) for grant DST/INSPIRE/04/2014/001799. S.S.~is supported by an ISIRD Grant from SRIC, IIT Kharagpur. N.S.M.~is supported by the NWO through QuantERA project QuantAlgo 680-91-034. T.M.~would like to thank Prahladh Harsha and Ramprasad Saptharishi for helpful discussions.


    \bibliography{reference}

\begin{appendix}
\section{Equivalence of the two forms of Chang's lemma}
\label{sec: appendix}
Recall that for any Boolean function $f:\pmone^n \to \pmone$ and positive real number $t$, we define $\cS_t:=\{S \subseteq [n]:|\wh{f}(S)|\geq \frac{1}{t}\}$. Chang's lemma for the Boolean hypercube is usually stated in the literature as an upper bound on the dimension of $\cS_t$, as in Lemma~\ref{lem:chang original}. 
\begin{lemma}[Common form of Chang's lemma]
\label{lem:chang original}
Let $f:\pmone^n \to \pmone$ be a Boolean function and $t$ be any positive real number. Let $d = \dim(\cS_t) > 1$. Then
\[
d=O(t^2\delta(f)^2\log(1/\delta(f))).
\]
\end{lemma}
In this section we show that Lemma~\ref{lem:chang original} and Lemma~\ref{lem:chang} can be easily derived from each other. For convenience, we restate Lemma~\ref{lem:chang} below.
\begin{lemma}[Restatement of Lemma~\ref{lem:chang}]
\label{lem:chang restate}
There exists a universal constant $c>0$ such that the following is true. Let $f:\pmone^n \to \pmone$ be a Boolean function and $t$ be any positive real number. Let $d = \dim(\cS_t) > 1$. If $\delta(f)<c$, then
    \[
        \delta(f) = \Omega\left(\frac{\sqrt{d}}{t \sqrt{\log \bra{t^2/d}}}\right).
    \]
\end{lemma}
We need the following claim.
\begin{claim}
\label{clm:monotone}
Define a function $h:(0, 1/e] \to \R$ by $h(\eta) = \eta \log(1/\eta)$. Then $h$ is monotonically non-decreasing.
\end{claim}
\begin{proof} Define $g(\eta):=\eta \ln(1/\eta)$, where $\ln$ denotes the natural logarithm (i.e., logarithm with base $e$). Since, $h(\eta)=\log_2e\cdot g(\eta)$, it is sufficient to show that $g$ is monotonically non-decreasing for $\eta \in (0,1/e]$. 
Let $g'(\cdot)$ denote the derivative of $g$ with respect to $\eta$. We have that,
\[
g'(\eta)=\ln (1/\eta)-1,
\]
which is non-negative for $\eta \in (0, 1/e]$. The claim follows.
\end{proof}

\begin{proof}[Proof of equivalence of Lemma~\ref{lem:chang original} and Lemma~\ref{lem:chang restate}]~
\begin{description}
    \item[Lemma~\ref{lem:chang original} $\implies$ Lemma~\ref{lem:chang restate}.]~
    Assume Lemma~\ref{lem:chang original} and let $f,t,\delta(f)$ and $d$ be as in the statement of Lemma~\ref{lem:chang restate}. From Lemma~\ref{lem:chang original} we have that $d=O(t^2\delta(f)^2\log(1/\delta(f)))$. Now,

\begin{align}
    &d = O(t^2\delta(f)^2\log\left(\frac{1}{\delta(f)}\right) \label{eq:dupperbound}\\
    \implies &\frac{t^2}{d} =  \Omega\left(\frac{1}{\delta(f)^2 \log \left(\frac{1}{\delta(f)}\right)} \right). \label{eq:t^2/dlowerbound-interim}
\end{align}
Equation~\eqref{eq:t^2/dlowerbound-interim} implies that there exists a universal constant $c>0$ (that depends on the constant hidden in the asymptotic notation) such that $t^2/d > 1$ whenever $\delta(f)<c$. Assuming $\delta(f)<c$ and taking logarithm of both sides of Equation~\eqref{eq:t^2/dlowerbound-interim} we have that
\begin{align}
    0<\log \bra{t^2/d} = \Omega\left(\log (1/\delta(f)^2) - \log \log (1/\delta(f)) \right) = \Omega\left(\log (1/\delta(f)) \right). \label{eq:t^2/dlowerbound}
\end{align}
Equations~\eqref{eq:dupperbound} and \eqref{eq:t^2/dlowerbound} yield
\begin{align*}
    &\frac{d}{t^2 \log \bra{t^2/d}} = O\left( \frac{t^2 \delta(f)^2 \log (1/\delta(f))}{t^2 \log (1/\delta(f))} \right) = O(\delta(f)^2) \nonumber\\ 
    \implies &\delta(f) = \Omega\left(\frac{\sqrt{d}}{t \sqrt{\log \bra{t^2/d}}} \right).
\end{align*}
\item[Lemma~\ref{lem:chang restate} $\implies$ Lemma~\ref{lem:chang original}.]~
Assume Lemma~\ref{lem:chang restate}. Let $f,t,\delta(f)$ and $d$ be as in the statement of Lemma~\ref{lem:chang original} and $c$ be the universal constant in the statement of Lemma~\ref{lem:chang restate}. We assume without loss of generality (by replacing $f$ by $-f$ if necessary) that $\delta(f) \leq 1/2$. Now, we have that
\[
\frac{d}{t^2} \leq \frac{|\cS_t|}{t^2}\leq \sum_{S \in \cS_t}\wh{f}(S)^2 \leq 1,
\]
where the first inequality holds by the definition of $\calS_t$ and $d$, the second inequality holds from the definition of $\calS_t$, and the last inequality follows from Parseval's identity (Theorem~\ref{thm:Parseval}). We conclude that 
\begin{align}
\label{eq:t^2/d>1}
d \leq t^2.
\end{align}
We split the proof into two cases.
\begin{description}
\item[Case 1: $c \leq \delta(f) \leq 1/2$.\footnotemark]\footnotetext{This case is vacuous if $c>1/2$.} By Claim~\ref{clm:monotone} and the observation that $c>0$ and $\delta(f)^2 \leq 1/4 < 1/e$ we have that 
\begin{align}\delta(f)^2 \log(1/\delta(f)^2)\geq c^2 \log(1/c^2).\label{eq:mono}\end{align} By Equation~\eqref{eq:t^2/d>1} we have that
\[
d \leq t^2=O(t^2c^2\log(1/c^2))=O(t^2\delta(f)^2\log(1/\delta(f)^2))=O(t^2\delta(f)^2\log(1/\delta(f))),
\]
where the third step follows from Equation~\eqref{eq:mono}.
\item[Case 2: $0 \leq \delta(f)<c$.] 
From Lemma~\ref{lem:chang restate} we have that $\delta(f)^2 = \Omega\left(\frac{d}{t^2 \log (t^2/d)}\right)>0$. Recall that by our assumption, $\delta(f) \leq 1/2$ and hence $\log\bra{1/\delta(f)}\geq 1$. Now, if $t^2/d<2$, we have from Lemma~\ref{lem:chang restate} that $d=O(\delta(f)^2t^2\log\bra{t^2/d})=O(\delta(f)^2t^2)=O(\delta(f)^2t^2\log\bra{1/\delta(f)})$ and the proof is complete. Thus, assume henceforth that $t^2/d\geq 2$ and hence $\log\bra{t^2/d}\geq 1$. By Lemma~\ref{lem:chang restate}, Claim~\ref{clm:monotone} and the observation that $\delta(f)^2 \leq \frac{1}{4}<\frac{1}{e}$ we have that
\begin{align}
    t^2\delta(f)^2\log (1/\delta(f)^2) &= \Omega\left(t^2\bra{\frac{d}{t^2\log \bra{t^2/d}}}\cdot \log \bra{\frac{t^2\log(t^2/d)}{d}}\right)\nonumber \\
    &= \Omega\left(\bra{\frac{d}{\log \bra{t^2/d}}}\cdot \log \bra{\frac{t^2\log(t^2/d)}{d}}\right)\nonumber \\
    &=\Omega\left(\bra{\frac{d}{\log \bra{t^2/d}}}\cdot \log \bra{t^2/d}\right)\label{eq: cor1.2_123} \tag*{since $\log\bra{t/d^2}\geq 1$}\\ \nonumber
    &=\Omega(d).
\end{align}
This completes the proof.
\end{description}

\end{description}
\end{proof}
\end{appendix}

\end{document}